\documentclass[12pt]{article}
\usepackage{amsmath,amssymb,amsthm,amsfonts,mathtools}
\usepackage{xr-hyper}
\usepackage{epsfig,hyperref,xspace,natbib,graphicx}
\usepackage[font=small,labelfont=bf]{caption}
\usepackage{subcaption}
\usepackage{xcolor}
\usepackage{enumitem}
\usepackage{multirow}
\usepackage[space]{grffile}

\graphicspath{{./output/}}

\newtheorem{proposition}{Proposition}

\newtheorem{remark}{Remark}
\newtheorem{example-continued}{Example}[section]

\DeclareMathOperator*{\argmin}{arg\,min}

\author{
  Simon Freyaldenhoven\\
  \textit{Federal Reserve Bank of Philadelphia}
  \and
  Christian Hansen\\ 
  \textit{University of Chicago\thanks{We thank Thorsten Drautzburg, Lutz Killian, Christian Leuz, Paul Mohnen, Matthew Notowidigdo, Mikkel Plagborg-M\o{}ller, Jon Roth, Jesse Shapiro, as well as audiences at the Federal Reserve Bank of Philadelphia, the University of Bonn, Georgetown University, Maastricht University, and NY Camp Econometrics for comments. The views expressed herein are those of the authors and do not necessarily reflect the views of the Federal Reserve Bank of Philadelphia or the Federal Reserve System. Emails: \href{mailto:simon.freyaldenhoven@phil.frb.org}{simon.freyaldenhoven@phil.frb.org}, \href{mailto:chansen1@chicagobooth.edu}{chansen1@chicagobooth.edu}}}
}

\title{(Visualizing) Plausible Treatment Effect Paths}

\date{\today}

\begin{document}

\maketitle

\vspace{-36pt}

\begin{abstract}
\vspace{-8pt}
\noindent
We consider point estimation and inference for the treatment effect path of a policy. Examples include dynamic treatment effects in microeconomics, impulse response functions in macroeconomics, and event study paths in finance.
We present two sets of plausible bounds to quantify and visualize the uncertainty associated with this object. 
Both plausible bounds are often substantially tighter than traditional confidence intervals, and can provide useful insights even when traditional (uniform) confidence bands appear uninformative. Our bounds can also lead to markedly different conclusions when there is significant correlation in the estimates, reflecting the fact that traditional confidence bands can be ineffective at visualizing the impact of such correlation.
Our first set of bounds covers the average (or overall) effect rather than the entire treatment path.
Our second set of bounds imposes data-driven smoothness restrictions on the treatment path. Post-selection Inference (\cite{berk2013valid}) provides formal coverage guarantees for these bounds. The chosen restrictions also imply novel point estimates that perform well across our simulations. 
\end{abstract}



\thispagestyle{empty}

\newpage
\setcounter{page}{1}

\section{Introduction}\label{sec:idea}

We are interested in the treatment effect path of a policy at discrete horizons $h = 1,...,H$. Examples include dynamic treatment effects in microeconomics, impulse response functions in macroeconomics, and event study paths in finance.
We write $\beta=\{\beta_h\}_{h={1}}^{H}$ for the vector that collects this dynamic treatment effect path up to the fixed maximum horizon of interest $H$.
We assume access to point estimates of the parameters $\beta_h$, denoted by $\hat{\beta}_h$, that correspond to the cumulative effect of the policy at horizon $h=1, \ldots H$. Throughout, we assume the vector that collects the estimated dynamic treatment effect path, $\hat\beta$, satisfies $\hat\beta \sim N(\beta,V_{\beta})$ and that we have access to the covariance matrix $V_{\beta}$. Leading examples to obtain such estimates include distributed lag models, local projections, and event studies.\footnote{We abstract away from approximation issues, but note that standard asymptotic approximations within these settings along with access to consistent asymptotic variance estimators motivate this setup.}
We consider both point estimation and uncertainty quantification, though our focus will be on the latter. In particular, we introduce two approaches to visualize the uncertainty about the treatment path, which we call \textit{cumulative} and \textit{restricted} plausible bounds. Both bounds are often substantially tighter than traditional confidence intervals, and can provide useful insights even when traditional (uniform) confidence bands appear uninformative.

The standard approach in economics to quantify and visualize the uncertainty associated with parameter estimates is to construct confidence regions. 
Intuitively, a confidence region visualizes to the reader what values of the parameter, in this case $\beta$, are ``plausible'' based on the observed data. The idea being that values inside this region appear plausible, while values outside of the region do not. 
The two predominant confidence regions in practice are pointwise and sup-t confidence regions (e.g. \cite{callaway2021}; \cite{jorda2023local}; \cite{boxell2024}). 
A third alternative is the Wald confidence region $CR^{Wald}$. This region simply collects all parameter values $b$ that are not rejected by a standard Wald test of the null hypothesis that $\beta = b$ at level $\alpha$. 
While a confidence region constructed from pointwise confidence intervals does not achieve correct coverage for the vector $\beta$, both sup-t and Wald confidence regions achieve valid coverage: $\mathbb{P}(\beta \in CR^{Wald})= \mathbb{P}(\beta \in CR^{sup-t}) = (1-\alpha)$.\footnote{We discuss these regions, and their construction, in more detail in Appendix \ref{app-sec:crs}. For further discussion of uniform confidence bands, and, in particular, the merits of sup-t confidence bands, also see \cite{freyberger2018} and \cite{olea:pm}.}

\begin{figure}[tb]
		\centering
		\includegraphics[width=.5\linewidth]{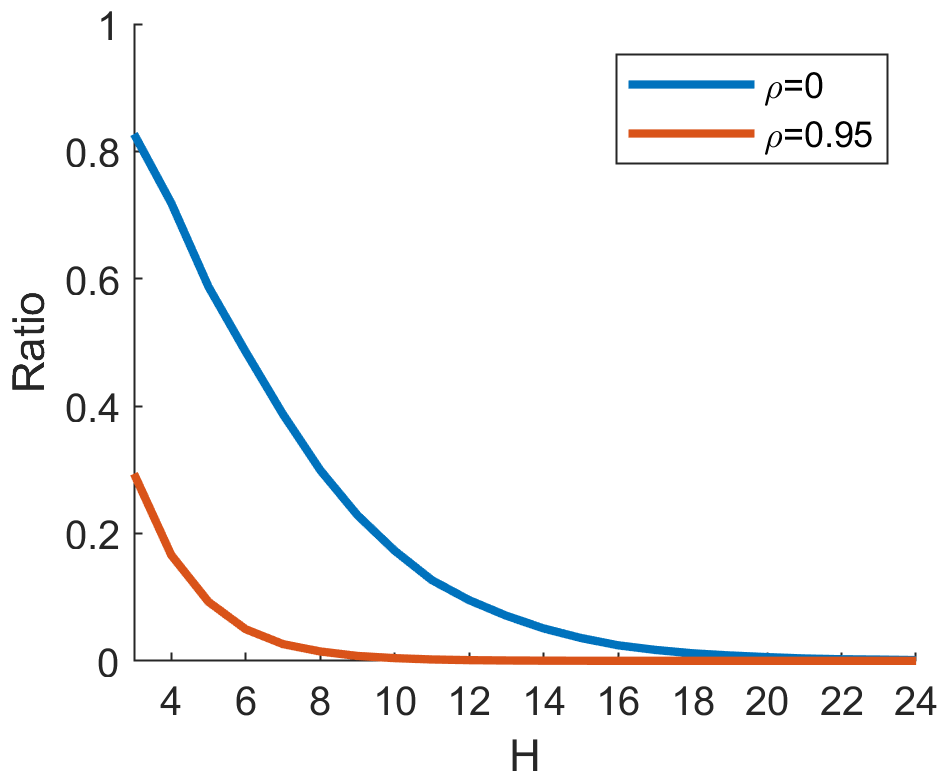}
		\caption{Volume of Wald confidence region relative to sup-t confidence region as a function of $H=\dim(\beta)$. $V_\beta$ is a symmetric Toeplitz matrix with entries $v_{ij}=\rho^{|i-j|}$.}
		\label{fig:volume}
\end{figure}

Sup-t and Wald confidence regions both come with some advantages and disadvantages.
Since the Wald region is an ellipsoid, a disadvantage of the Wald confidence region is that it becomes infeasible to visualize in higher dimensions (i.e. when $H > 3$).
The sup-t confidence region has the advantage of being easy to visualize. However, the volume of the sup-t confidence region quickly explodes relative to the volume of the Wald region. We illustrate this difference in confidence region volume in Figure \ref{fig:volume}, which plots the volume of the Wald confidence region relative to the volume of the sup-t confidence region as a function of the dimension $H$ for two exemplary covariance matrices $V_{\beta}$.  We see that the volume of the sup-t region tends to be orders of magnitude larger than the volume of the Wald region for the typical horizon that is depicted in event studies and impulse responses.\footnote{In what follows, we refer to the visualizations of $\hat{\beta}$ simply as treatment effect plots.} When $V_{\beta}$ is the identity matrix ($\rho=0$), the relative volume of the Wald region is less than 10\% and around 0.1\% of the volume of the sup-t region for $H=12$ and $H=24$ respectively. These numbers are generally even smaller if the entries in $\hat{\beta}$ have non-zero correlation: When $V_{\beta}$ is a symmetric Toeplitz matrix with entries $v_{ij}=0.95^{|i-j|}$, this ratio drops to 3.5\% and 0.0001\% for $H=12$ and $H=24$ respectively. One immediate consequence is that, for even moderate horizons $H$, the overwhelming majority of paths inside the sup-t bands 
would be rejected by a simple joint hypothesis test. This property seems unappealing to us and serves as a first indication that sup-t confidence bands may not always be appropriate for visualizing what dynamic treatment effect paths are plausible.

\begin{figure}[tb]
	\begin{subfigure}[t]{.49\textwidth}
		\centering
		\includegraphics[width=\linewidth]{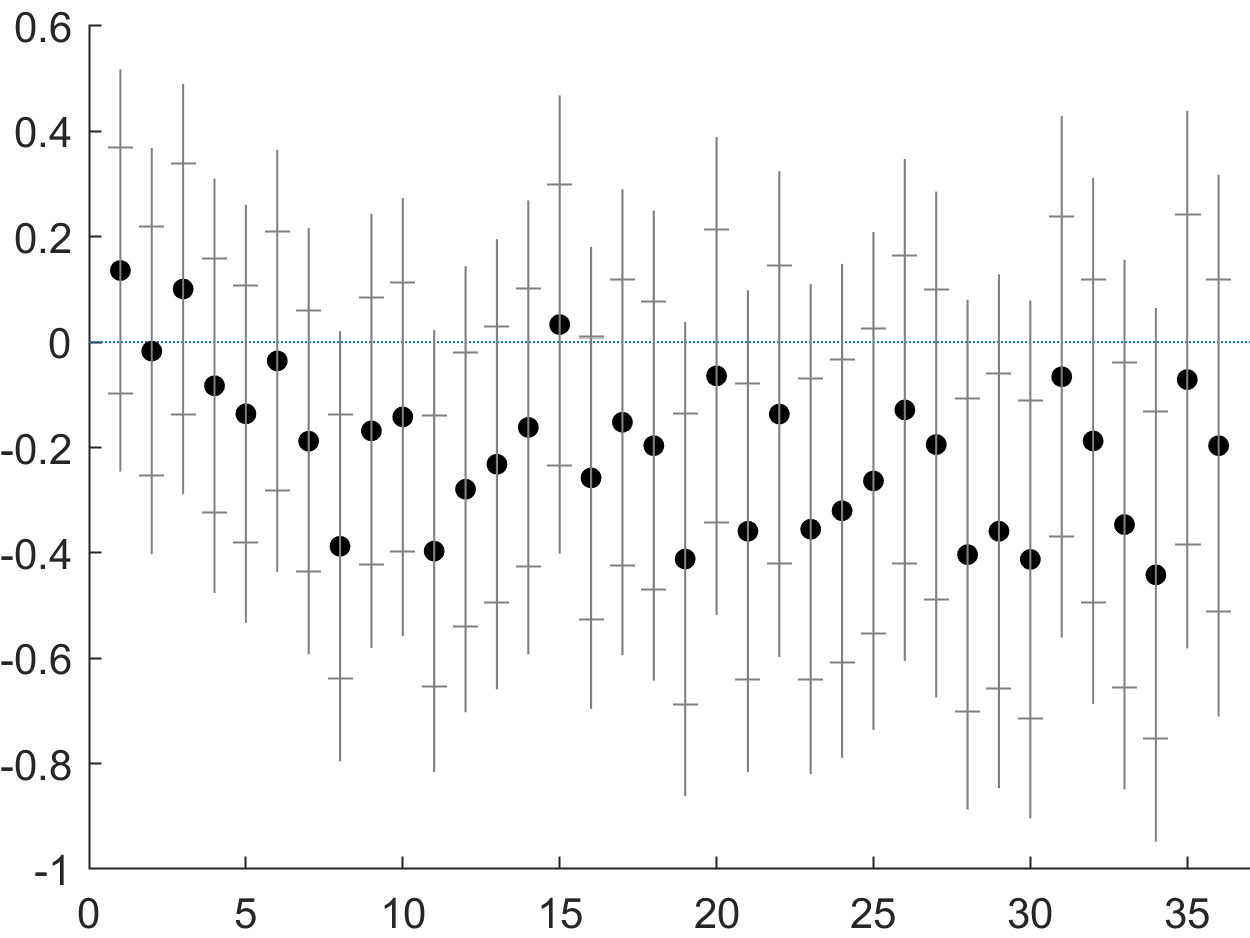}
		\caption{Simulated data with $Cov(\hat{\beta}_i,\hat{\beta}_j)=0$ for $i \ne j$. 
		}
		\label{fig:example_existing}
	\end{subfigure}\hfill
	\begin{subfigure}[t]{.49\linewidth}
		\centering
		\includegraphics[width=\linewidth]{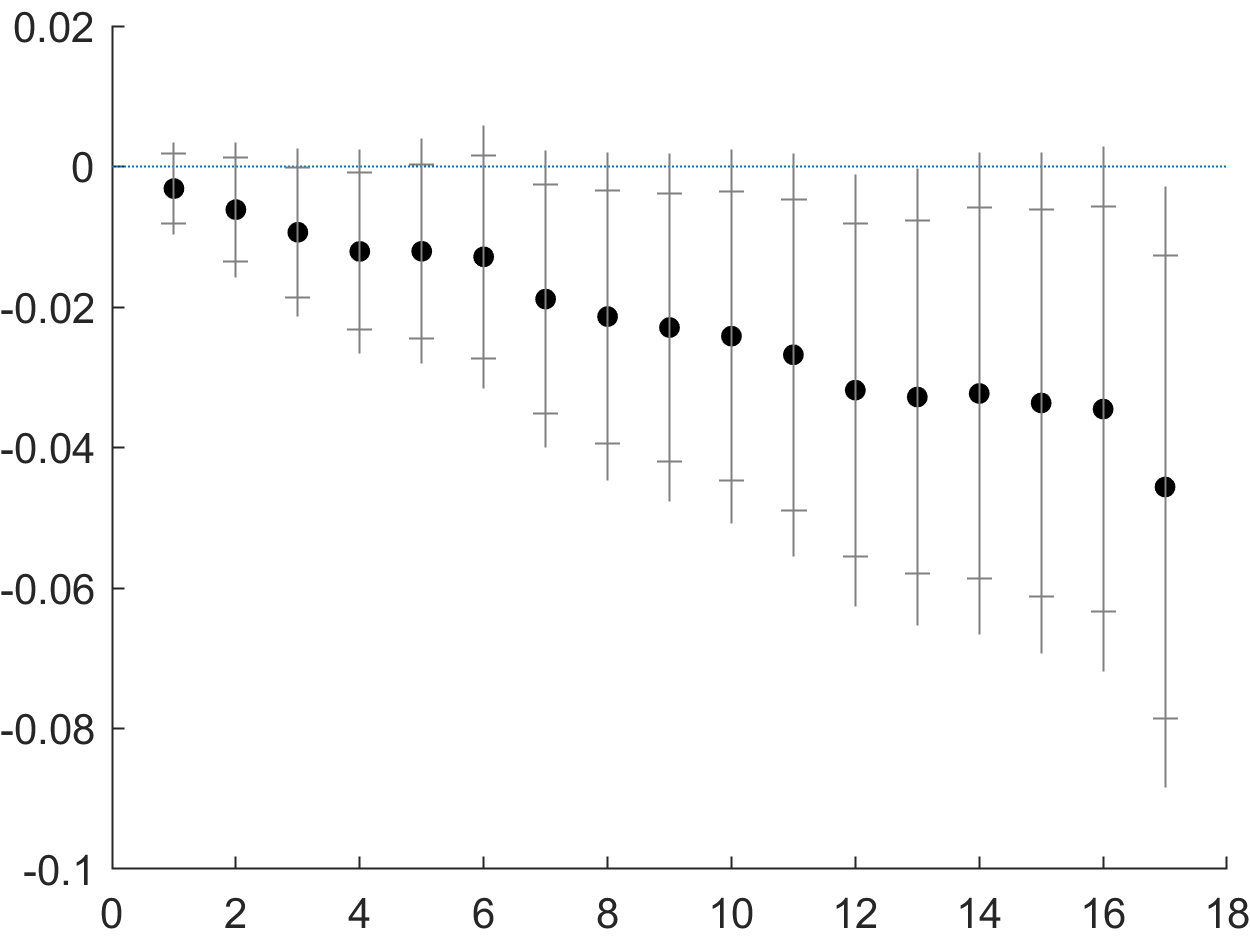}
		\caption{Replicated from \cite{bosch2014} (Figure 4b).}
		\label{fig:roth_example_existing}
	\end{subfigure}
	\caption{Two exemplary treatment effect plots, including point estimates, and pointwise and sup-t confidence intervals.}
	\label{fig:illustration}
\end{figure}

To illustrate this further, Figure \ref{fig:illustration} depicts two exemplary treatment effect plots.  The object of interest is the treatment path of a policy over the depicted horizon. We have access to jointly normal estimates $\{\hat{\beta}_h\}_{h=1}^{H}$, with observed point estimates $\hat{\beta}$ given by the black dots. Both panels further include pointwise 95 percent confidence intervals (inner confidence set as indicated by the dashes) and uniform 95 percent sup-t confidence bands (outer confidence set). While the pointwise confidence intervals only permit testing of pre-selected hypotheses for individual coefficients $\beta_h$, the sup-t bands contain the entire true path $\beta$ in 95 percent of realized samples.
Figure \ref{fig:example_existing} depicts a hypothetical example with zero correlation between the estimated coefficients.\footnote{We give more detail on the underlying DGP in Section \ref{sec:sim}.} 
Figure \ref{fig:roth_example_existing} is based on the same estimates as Figure 4b in \cite{bosch2014}. In this example, all off-diagonal entries in $V_{\beta}$ are positive, and the average correlation between adjacent coefficients is 0.95.

The sup-t region in Figure \ref{fig:example_existing} includes treatment paths that imply an overall positive effect (paths with $\sum_{h=1}^H \beta_h > 0$) and treatment paths with very different shapes. In fact, $\beta=0$ falls inside the sup-t bands, suggesting that the null of ``no treatment effect" is plausible. However, a joint test of the null hypothesis that $\beta=0$ yields a p-value of $1.54 \times 10^{-9}$. 
In contrast, $\beta=0$ falls outside the sup-t bands in Figure \ref{fig:roth_example_existing}, suggesting that the null of ``no treatment effect" is not plausible. However, a joint test of the null hypothesis that $\beta=0$ yields a p-value of $0.33$.
These discrepancies between the easy-to-visualize sup-t region and the results of simple joint hypothesis tests again suggest to us that the sup-t confidence region may not always be providing an empirically effective visualization of what treatment effect paths are plausible.

\begin{figure}[tbhp!]
	\begin{subfigure}[t]{.45\textwidth}
		\centering
		\includegraphics[width=\linewidth]{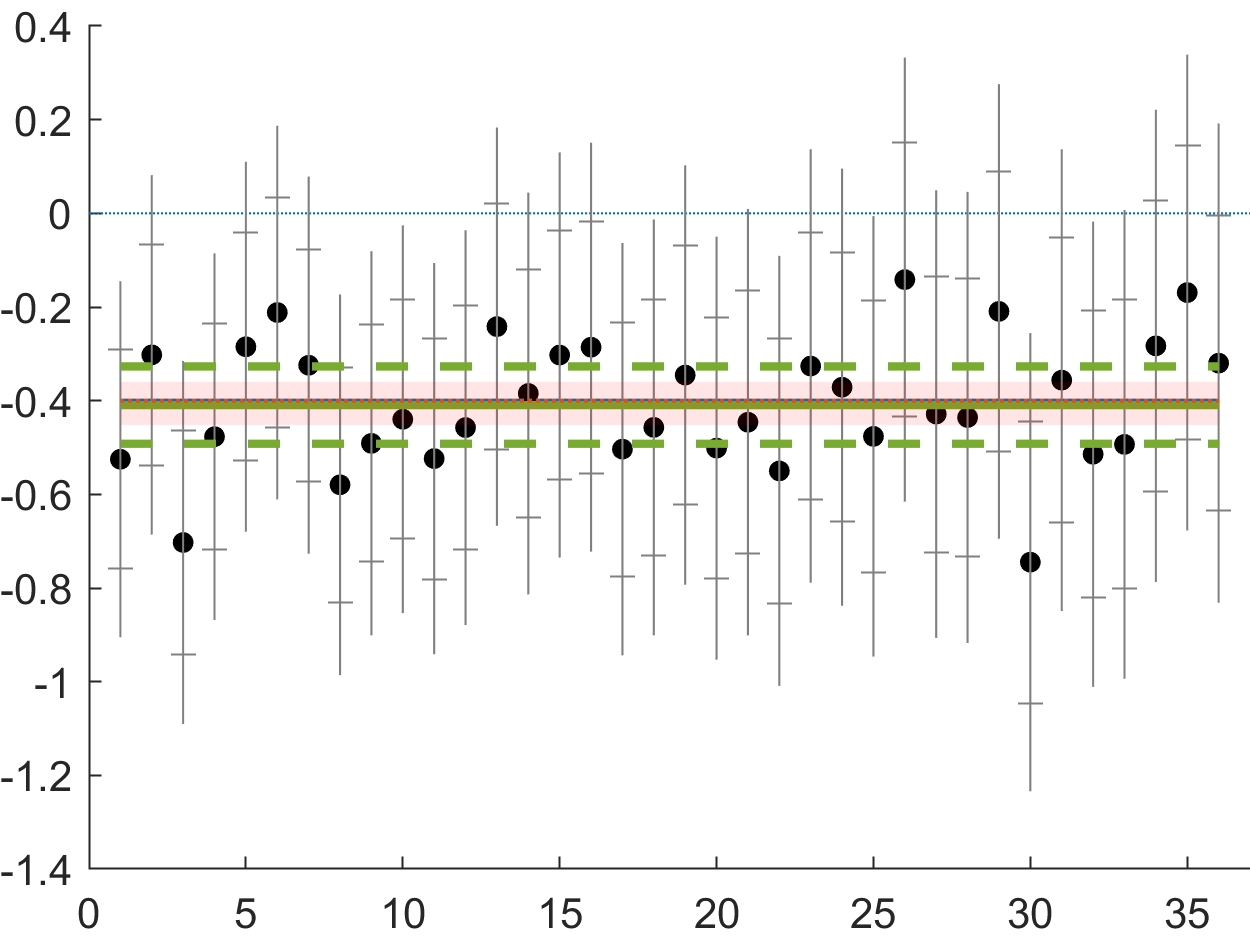}
		\caption{treatment path constant}
		\label{fig:new_viz_constant}
	\end{subfigure}\hfill
	\begin{subfigure}[t]{.45\linewidth}
		\centering
		\includegraphics[width=\linewidth]{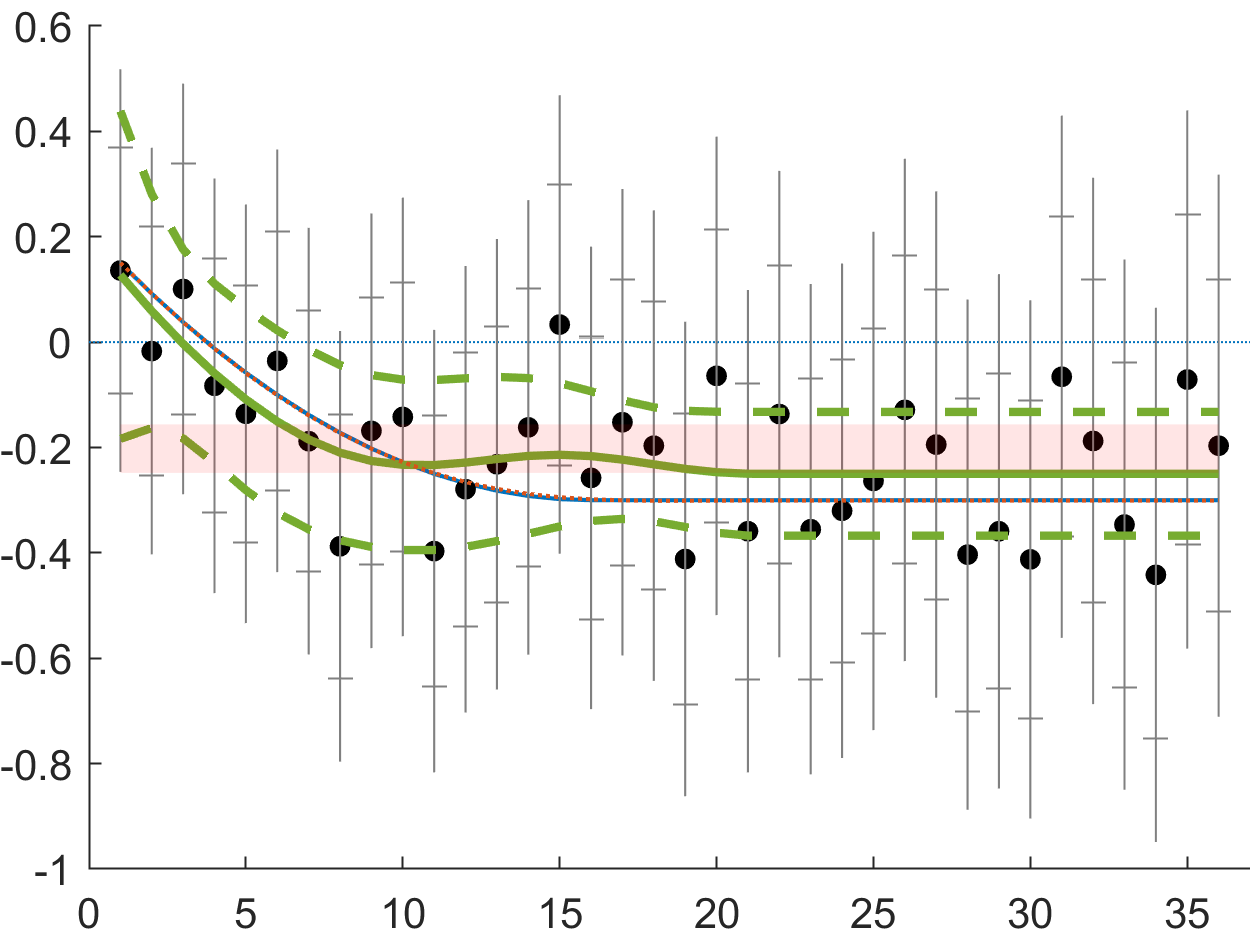}
		\caption{treatment path smooth, eventually flat}
		\label{fig:new_viz_quadratic}
	\end{subfigure}
	
	\begin{subfigure}[t]{.45\textwidth}
		\centering
		\includegraphics[width=\linewidth]{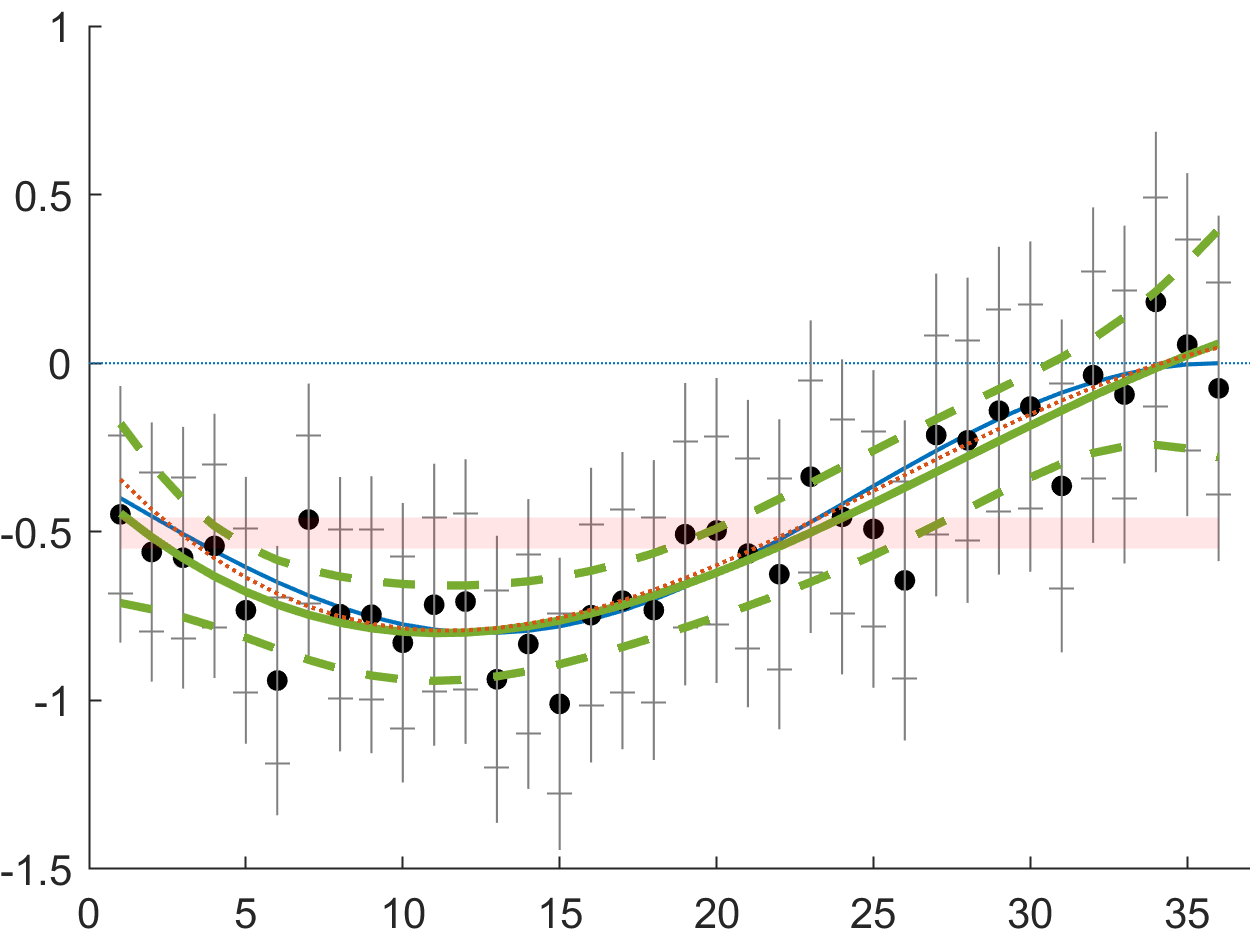}
		\caption{treatment path hump-shaped}
		\label{fig:new_viz_no_flat}
	\end{subfigure}\hfill
	\begin{subfigure}[t]{.45\linewidth}
		\centering
		\includegraphics[width=\linewidth]{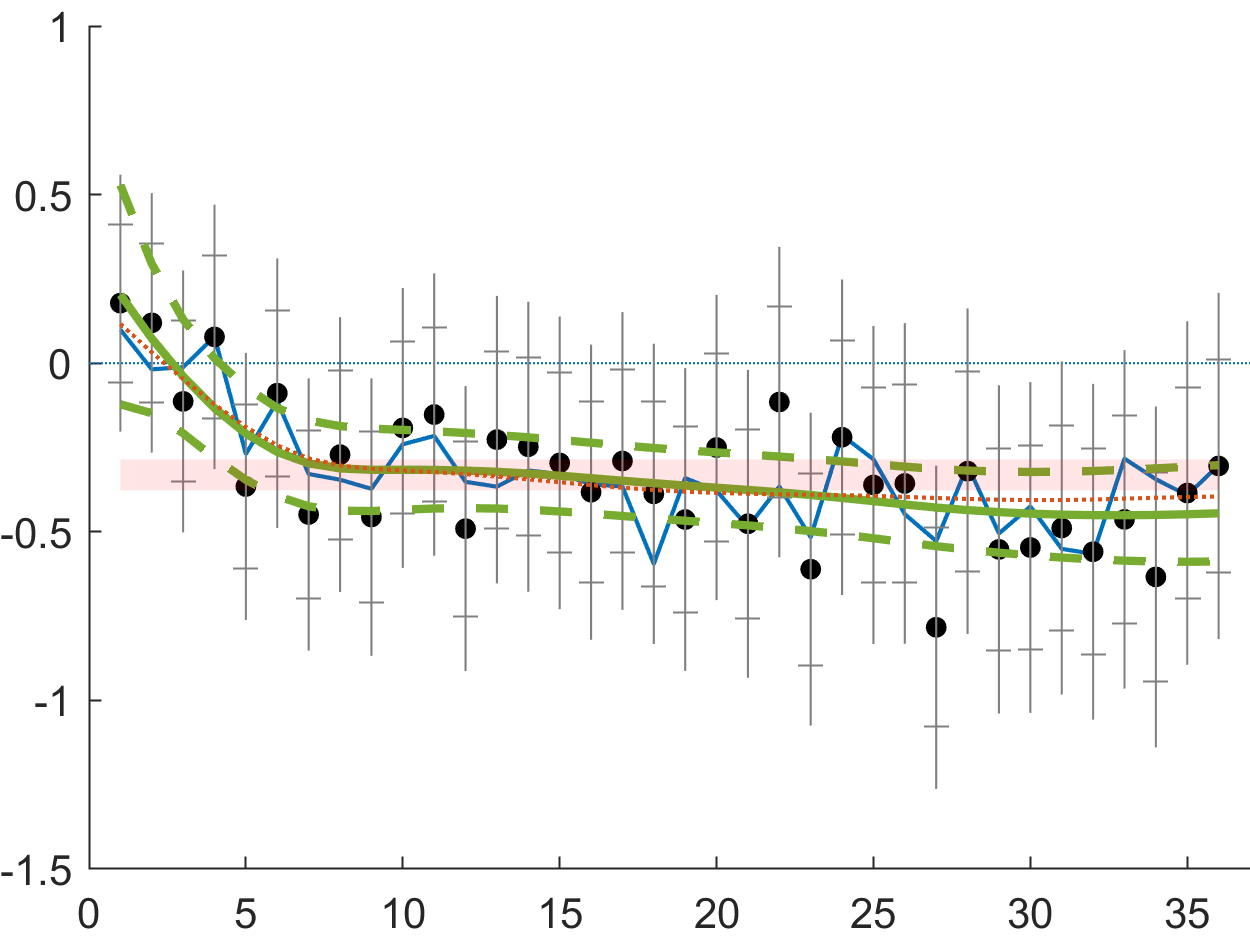}
		\caption{treatment path wiggly}
		\label{fig:new_viz_wiggly}
	\end{subfigure}

		\begin{subfigure}[t]{.45\textwidth}
		\centering
		\includegraphics[width=\linewidth]{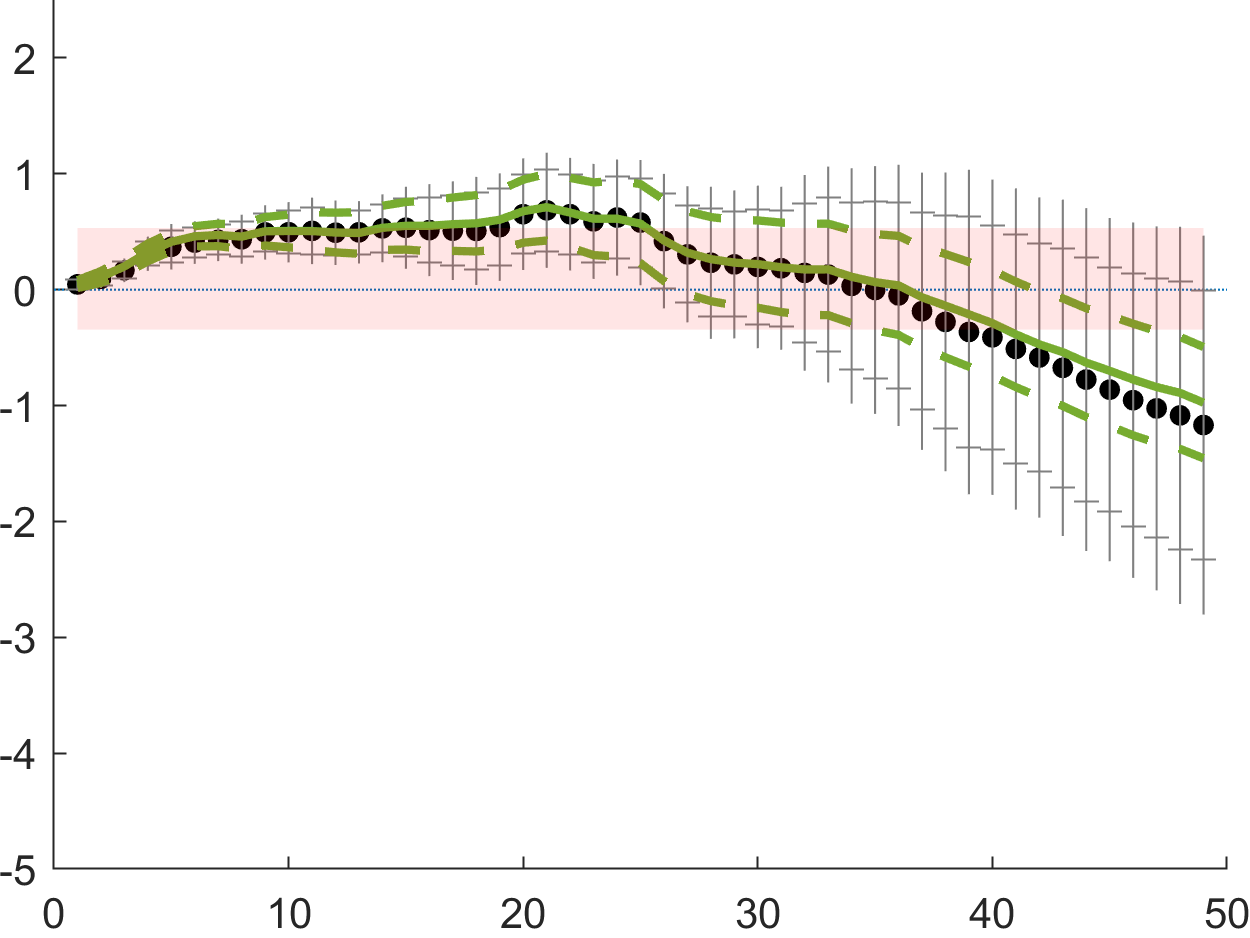}
		\caption{Replicated from \cite{nakamura2018identification} (Figure A.2)}
		\label{fig:new_viz_nakamura}
	\end{subfigure}\hfill
	\begin{subfigure}[t]{.45\linewidth}
		\centering
		\includegraphics[width=\linewidth]{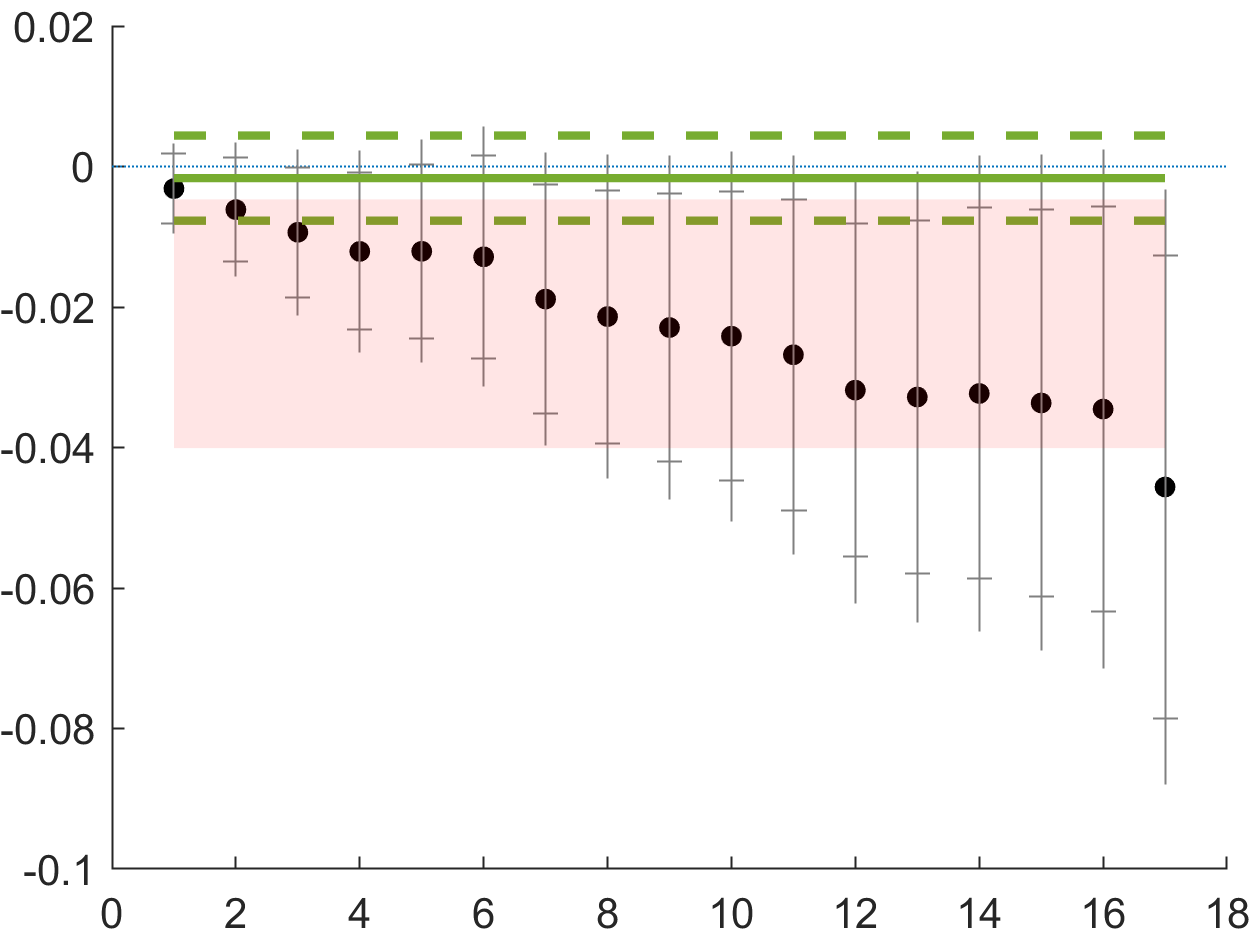}
		\caption{Replicated from \cite{bosch2014} (Figure 4b)}
		\label{fig:new_viz_bosch}
	\end{subfigure}
	\caption{Exemplary treatment effect plots including our proposals. Our proposed visualization includes two additional objects. The shaded red areas provides the \emph{cumulative plausible bounds}. The dashed green lines provide the \emph{restricted plausible bounds}, and the thick solid green line provides the corresponding restricted estimates. In the simulated panels a)-d), we further include the true treatment path (thin blue line) and its \emph{surrogate} (dotted red line).}
	\label{fig:new_viz}
\end{figure}

In Figure \ref{fig:new_viz}, we therefore introduce two alternative ways to visualize plausible treatment effect paths. 
Panels \ref{fig:new_viz_constant}-\ref{fig:new_viz_wiggly} are based on different simulated data generating processes, while panels \ref{fig:new_viz_nakamura} and \ref{fig:new_viz_bosch} are based on published figures in Macroenomics (\cite{nakamura2018identification}) and Applied Micro (\cite{bosch2014}). Elements from the ``standard'' treatment effect plot are provided by the black dots which give point estimates of the treatment effect at each horizon and by the inner and outer confidence intervals corresponding respectively to the usual pointwise and sup-t confidence intervals.\footnote{Figure \ref{fig:new_viz_quadratic} is based on the same estimates $\hat{\beta}$ as Figure \ref{fig:example_existing}. Figure \ref{fig:new_viz_bosch} is based on the same estimates $\hat{\beta}$ as Figure \ref{fig:roth_example_existing}.} In addition, panels \ref{fig:new_viz_constant}-\ref{fig:new_viz_wiggly} include a solid blue line to represent the true treatment effect path and an additional dotted red line, which we explain below. Finally, each plot includes two new features: (i) the shaded red area, and (ii) the dashed and solid green lines. Importantly, these new features shift the goal posts relative to the Wald and sup-t bounds, and the inferential target of our bounds is \emph{not} the true treatment path.

The shaded red area represents our proposed 95\% \emph{cumulative plausible bounds}. We construct these bounds so that the average treatment effect across the depicted horizons will be within these bounds for 95\% of all realizations of the data. 
For example, in Figure \ref{fig:new_viz_quadratic}, these bounds suggest that the average effect of the policy over the 36 periods depicted is between (-0.248, -0.156), and thus that the overall effect of the policy over the 36 periods is strictly negative and inside the window (-8.93, -5.62). In contrast to the standard sup-t region, these bounds suggest that a treatment path  with no overall effect of the policy is not plausible. These cumulative plausible bounds have an alternative interpretation in terms of the overall treatment effect path. Specifically, the cumulative plausible bounds are such that the treatment path $\beta$ (depicted as the solid blue line) will \emph{on average} be within these bounds for 95\% of all realizations of the data.

The dashed green lines represent our proposed 95\% \emph{restricted plausible bounds}, which are centered around \emph{restricted estimates} provided by the solid green line. These restricted estimates and bounds are motivated by envisioning a researcher who is interested in understanding key features of the treatment effect path but is not concerned with necessarily covering the entire true path at every horizon. However, we also imagine the researcher as being \emph{ex ante} unsure about what the important features are and wanting to use the data to help select a restricted model for summarizing the treatment effect path. 

More concretely, we construct the restricted estimates and plausible bounds by using a statistical model selection procedure to select an approximating model from within a pre-specified universe of candidates. We consider a default set of models motivated by a preference for smooth dynamics that eventually die out induced by shrinking first and third differences of $\hat\beta$, though we note that the procedure could be applied with any finite, pre-specified universe of models. The restricted estimates are then simply the point estimates of the treatment path based on the selected model. We construct the restricted plausible bounds to provide uniform (95\%) coverage accounting for data-dependent model selection by applying Berk et al.'s [\citeyear{berk2013valid}] Post-Selection Inference (PoSI) to our setting. 

Looking at the restricted estimates and restricted plausible bounds in each panel paints a starkly different picture compared to the sup-t intervals. In all cases, the restricted plausible bounds are relatively narrow and seemingly quite informative about the broad features of the treatment effect paths. Figure \ref{fig:new_viz_bosch} stands out and merits further discussion. In this instance, our model selection procedure selects a constant treatment effects model. Our restricted estimates then coincide with the MLE estimate of a constant treatment effects model.\footnote{That is, a model with $\hat\beta \sim N(\beta,V_{\beta})$, where $\beta$ is constant across $h$.} Remarkably, this estimate is $-0.0017$, which is outside of the convex hull of the individual estimates $\hat{\beta}_h$.\footnote{Intuitively, this behavior results from the strong positive correlation in the estimates combined with more precise estimates in early periods. We note that this strong positive correlation in the estimates cannot be inferred from the traditional plot.} Further, a Wald test of the null hypothesis that $\beta_h=-0.0017$ for $h=1, \ldots, H$ gives a p-value of 0.36. That is, a traditional joint hypothesis test suggests there is relatively little evidence against this hypothesis, which thus appears relatively ``plausible,'' contrary to what a visual inspection of the traditional treatment effect plot might suggest. This difference underscores that traditional treatment effect plots may be ineffective at visualizing the impact of the off-diagonal entries in $V_{\beta}$ (we illustrate this further in Appendix \ref{app-sec:crs}). In contrast, the entire covariance matrix $V_{\beta}$ is reflected in our restricted plausible bounds. Accounting for the covariance structure can lead to interestingly different results, improving the informativeness of these plots.

Finally, we reiterate that the inferential target of the restricted plausible bounds in the population is \emph{not} the true treatment path. Rather, the restricted plausible bounds provide uniform coverage of a \emph{surrogate} path given by the approximation that would be obtained by applying the selected model to the true effect path. We depict the selected surrogate for each of the simulated scenarios in Figure \ref{fig:new_viz} with a dotted red line. In Figure \ref{fig:new_viz_constant} and Figure \ref{fig:new_viz_quadratic}, this surrogate is indistinguishable from the true treatment path. In Figure \ref{fig:new_viz_no_flat} and Figure \ref{fig:new_viz_wiggly}, the surrogate differs from the true treatment path but visually captures what seem to be key features of the overall treatment path. Indeed, we suspect many empirical researchers, if given the true treatment path from Figure \ref{fig:new_viz_wiggly}, would actually be more interested in the smooth approximation provided by the surrogate in this case. That is, we view the fact that the restricted plausible bounds cover a data-dependent approximation to the population treatment effect path as a potentially appealing feature. 

In summary, we propose augmenting standard event-study plots with two additional elements: a shaded red region (the cumulative plausible bounds) and dashed and solid green lines (the restricted plausible bounds and estimates, respectively). Together with the usual pointwise and sup-t intervals, these visualizations offer a more comprehensive view of plausible effect paths, each serving distinct inferential purposes. Sup-t bands provide a simple, assumption-free summary of plausible paths. Pointwise intervals target effects at specific horizons. Cumulative bounds inform average treatment effects, while restricted bounds and estimates capture approximations of the effect path obtained using data-driven smoothing.

Our paper connects to several strands in the literature. We obtain our restricted plausible bounds and estimates by considering a finite set of candidate models for $\beta$. This approach is thus closely related to work that considers parametric models and approximations to $\beta$. Such approaches have been studied going back at least to \cite{almon1965}, who imposes a parametric model on distributed lag coefficients. More recently, \cite{barnichon2018functional} propose approximating impulse responses with a set of basis functions, and \cite{barnichon2019} propose to shrink impulse response estimates towards polynomials. While related, we differ from these approaches by focusing on inference for a data-dependent surrogate effect path; see, for example, \cite{genovese2008adaptive} for a general discussion of inference on surrogates.

We operationalize our restricted plausible bounds by using data-dependent selection from a universe of candidate models with different fixed degrees of shrinkage over first and third differences. This model universe is closely related to the structure employed in \cite{shiller1973} which takes a fully Bayesian approach to estimating a distributed lag model under a normal prior on the $d^{\textrm{th}}$ order difference of $\beta$. Our restricted estimates are thus akin to point estimates that could be obtained by taking an empirical Bayes approach within the framework of \cite{shiller1973}. From the empirical Bayes perspective, one could then potentially adapt \cite{armstrong2022robust} to the present context to obtain interval estimates.\footnote{Also see the SmIRF estimator of \cite{plagborg2016essays} for a related approach that includes confidence sets with guaranteed frequentist coverage.}
In contrast, our restricted plausible bounds provide frequentist coverage for the population value of the selected surrogate path. 

To maintain coverage guarantees for the selected surrogate, accounting for data dependent model selection, we use a version of post-selection inference (PoSI) confidence intervals (\cite{berk2013valid}). Given that our inferential target is the population value of the selected model, we note that one could adopt other approaches from the literature on selective inference; see, e.g., \cite{taylor2015statistical} and \cite{kuchibhotla2022post} for excellent reviews.

There are a variety of other approaches to quantifying and visualizing uncertainty about treatment effect paths available in the literature.  For example, \cite{sims1999error} argues that conventional pointwise bands common in the literature should be supplemented with measures of shape uncertainty, and proposes such measures. \cite{jorda2009simultaneous} suggests a method to construct simultaneous confidence regions for impulse responses given propagation trajectories. \cite{freyberger2018inference} propose a uniformly valid inference method for an unknown function or parameter vector satisfying certain shape restrictions. More generally, inference for the treatment effect path is tightly tied to more general nonparametric inference problems; see, e.g., \cite{chenadaptive2024} for an interesting recent example that explicitly allows for use of a data-dependent sieve dimension. The ``shotgun plot'' of \cite{inoue2016}, which depicts a random sample of $B$ impulse responses contained in the joint Wald confidence set, provides an alternative approach to visualizing plausible treatment effect paths. We believe our proposal to provide simple additional visual elements to the usual treatment effect plot provides a useful complement to this existing literature.

\section{Cumulative Plausible Bounds}

We first present a simple visual feature, the \emph{cumulative plausible bounds}, that can be added to a standard treatment effect plot. This visualization does not impose any functional form or smoothness assumptions on the underlying treatment path, but, in terms of the full treatment effect path, it also does not achieve uniform coverage. Rather, the cumulative plausible bounds use a weaker notion of ``cumulative coverage": The true treatment path will \textit{on average} be within the cumulative plausible bounds in ($1-\alpha$)\% of all realizations of the data for a given significance level $\alpha$. 
That is, by providing valid inference for the average effect over the horizon $H$, our cumulative plausible bounds provide a simple visual element that conveys uncertainty about the average treatment effect. 

These cumulative plausible bounds are simply visualizations of the dynamic treatment path corresponding to the largest and smallest sum of all treatment effects up to horizon $H$ not rejected by a standard hypothesis test. They can be interpreted as boundary paths that would be consistent with the upper and lower limits of a confidence interval for the overall effect of the policy over $H$ periods.
Formally, let 
\begin{align}\label{eq:wald}
	u^{1-\alpha} = \max \sum_{h=1}^H \beta^*_h \qquad \text{s.t.} \ (\beta^*-\hat{\beta})'V_{\beta}^{-1}(\beta^*-\hat{\beta}) = \kappa^{(1-\alpha)},
\end{align}
where $\kappa^{(1-\alpha)}$ denotes the inverse of the chi-square cdf with one degree of freedom at chosen significance level $(1-\alpha)$. We further define $l^{1-\alpha}$ analogously, replacing the $\max$ in \eqref{eq:wald} with $\min$.
Since both $u^{1-\alpha}$ and $l^{1-\alpha}$, corresponding to the upper and lower limit of the overall effect are scalars, there are infinitely many treatment paths that correspond to these bounds on the overall treatment effect. To visualize the bounds, we use
$(U,L)^{1-\alpha}=\{U^{1-\alpha}_h, L^{1-\alpha}_h\}_{h=1}^H$, where $U^{1-\alpha}_h = \frac{u^{1-\alpha}}{H}$ and $L^{1-\alpha}_h = \frac{l^{1-\alpha}}{H}$.\footnote{Instead of using bounds $(U,L)^{1-\alpha}$ that are constant across $h$, one could alternatively depict bounds that reflect the shape of the unrestricted estimates.} We choose this visualization as the interval $(U,L)^{1-\alpha}$ is a $(1-\alpha)$\% Wald confidence interval for the average effect of the policy over the horizon $H$, $\frac{1}{H} \sum_{h=1}^H \beta_h$. 

The following trivial proposition clarifies how coverage of these cumulative plausible bounds relates to coverage of the treatment path. 

\begin{proposition}\label{prop:wald}
	The true treatment path $\beta$ will on average (over H) be within the cumulative plausible bounds for ($1-\alpha$) of all realizations: 
	\begin{align}
		\mathbb{P}\left(\frac{\sum_{h=1}^H L^{1-\alpha}_h}{H}< \frac{\sum_{h=1}^H \beta_h}{H}  < \frac{\sum_{h=1}^H U^{1-\alpha}_h}{H} \right) = (1-\alpha).
	\end{align}
\end{proposition}
\begin{proof}
	By construction, a Wald test on the cumulative treatment effect with significance level $(1-\alpha)$ will reject 
	\begin{enumerate}
		\item[a)] any treatment path $\tilde \beta_h$ with $\sum_{h=1}^H \tilde \beta_h > \sum_{h=1}^H U^{1-\alpha}_h$,
		\item[b)] any treatment path $\tilde \beta_h$ with $\sum_{h=1}^H \tilde \beta_h < \sum_{h=1}^H L^{1-\alpha}_h$.
	\end{enumerate}
	Since a Wald test for the cumulative effect has correct size, the result follows immediately.
\end{proof}

Proposition \ref{prop:wald} states that, for a given significance level $\alpha$, the true treatment path will \emph{on average} be within our bounds for ($1-\alpha$)\% of all realizations. This follows immediately from the fact that any path that is not, on average, inside the cumulative plausible bounds implies an overall treatment effect over $H$ periods that is rejected by the corresponding hypothesis test.

\section{Restricted Plausible Bounds}

The second idea we pursue is to present confidence regions that cover approximations of the true effect path that have ``reasonable shapes." We term these confidence regions \emph{restricted plausible bounds}. Here, we define ``reasonable shapes'' by pre-specifying a universe of models. We then use data-dependent model selection to choose a good representation for $\hat{\beta}$ from among this set. Intuitively, this approach is related to directly imposing a functional form restriction as is often done in empirical work, for example by
\begin{itemize}
\item specifying a parametric model for $\beta$, e.g. imposing a constant treatment effects model ($\beta_{h}=\beta_{h'} \ \forall h, h'$),
\item aggregating the underlying dataset over time (e.g. monthly to quarterly), which effectively restricts $\beta$ to ``step functions,''
\item estimating an impulse response function (IRF) via a vector auto regression (VAR), which restricts the IRF to functional forms compatible with the chosen VAR (cf. the discussions in \cite{li2024local} and \cite{olea2024double}).
\end{itemize}

One key feature of our approach is that we do not rely on a fixed functional form restriction or make use of some other implicit or ad hoc device to choose a restricted model. Rather, we select a model, and then take model selection explicitly into account when constructing confidence bounds. That is, we propose a model selection procedure that is explicit, transparent, and will allow us to maintain formal coverage guarantees instead of implicitly using the data to select a restricted model, which leads to invalid inference. 

Before we formally define our proposal, we introduce some necessary notation. We first borrow from the nonparametric statistics literature to introduce the notion of a ``surrogate" (cf. \cite{genovese2008adaptive}).  A surrogate path $\beta_M$ is close to, but potentially simpler than, $\beta$. We note that the surrogate path is a population object that approximates $\beta$, the true treatment path.
For example, we may define a constant treatment effects surrogate of $\beta$ as $\beta_s=\argmin_{b} \ (\beta - b)'(\beta - b) \ \text{ s.t. } \Delta b =0$. If the surrogate model $M$ is fixed a priori (and not itself a function of the data), inference for $\beta_M$ is straightforward, though we stress that any inferential statements in this case will be about $\beta_M$ and not $\beta$.\footnote{Targeting a simple surrogate function is akin to the standard approach in economics of estimating linear models even when the conditional expectation function is not believed to be linear. One can think of the linear model as a ``surrogate model'' capturing the best linear predictor. Inference will then be about the linear surrogate, and not the ``truth.''}  However, failing to take into account that the data is used to select the surrogate creates a problem for inference (e.g. \cite{leeb2005} or \cite{roth2022pretest}). In our setting the surrogate is explicitly a function of the data (or more precisely, of the unrestricted estimates $\hat{\beta}$), and we may thus write $\beta_{M(\hat{\beta})}$ to denote a data dependent surrogate path. In a first step, we use the data to select the surrogate model. In a second step, we then create a uniformly valid confidence region for the selected surrogate path, taking into account that the choice of surrogate is also random (i.e. a function of the data).

Given that we are doing model selection from a specified universe of models, a key choice is the specific model universe we consider. We consider a model universe motivated by the following economic intuition:
\begin{enumerate}
\item The dynamics of the treatment effect die off eventually. That is, after $K$ periods, the treatment effect is constant. We treat $K$ as unknown and allow $K$ to be as large as $H$, thus allowing dynamics across the entire depicted horizon.
\item The dynamic treatment path is ``smooth,'' where we measure smoothness using the third differences of the treatment path.
\end{enumerate}
In practice, we use shrinkage over first and third differences of $\hat\beta$ to implement 1. and 2.

Formally, we assume that the estimates of the treatment path $\hat{\beta}$ are jointly normal with $\hat{\beta} \sim N(\beta,V_{\beta})$, where $V_{\beta}= \sigma^2 V$, $\sigma^2=\frac{1}{H}\sum_{h=1}^H V_{\beta}(h,h)$, and $V$ is positive-definite. Taking $\hat{\beta}$ as input, we define the following object:
\begin{equation}\label{eq:obj}
\begin{aligned}
  \tilde\beta(\lambda_1,\lambda_2,K) &= \arg\min_{b} \ Q(b,\lambda_1,\lambda_2,K) \\
   &= \arg\min_{b} \ \underbrace{(\hat{\beta} - b)'V^{-1}(\hat{\beta} - b)}_{\text{distance from }\hat{\beta}} + \lambda_1 \underbrace{b'D_1'W_1(K)D_1 b}_{\substack{\text{penalty on first difference} \\ \text{after horizon K}}} +
   \lambda_2 \underbrace{b'D_3'W_3 D_3 b,}_{\substack{\text{penalty on} \\ \text{third difference}}}
\end{aligned}
\end{equation}
where 
\begin{itemize}
	\item $D_1$ and $D_3$ are the $(H \times H-1)$ and $(H \times H-3)$ first and third difference operators, 
	\item $\eta = (\lambda_1$, $\lambda_2$, $K$) are tuning parameters, 
	\item $W_1(K)$ and $W_3$ are weighting matrices where $W_1(K)$ only places weight on first differences for horizon $K \leq H$ and beyond (See Appendix \ref{app-sec:algorithm} for further details.).
\end{itemize}

Solving \eqref{eq:obj} provides a closed form solution\footnote{For intuition, note that the problem in \eqref{eq:obj} is closely related to the following constrained optimization with tuning parameters $c_1$, $c_2$, and $K$, explicitly bounding the first and third difference:
	\begin{align*}
		\hat\beta(c_1,c_2,K)  &= \arg\min_{b} \ \underbrace{(\hat{\beta} - b)'V^{-1}(\hat{\beta} - b)}_{\text{distance from }\hat{\beta}} \\
		&\text{such that} \quad \underbrace{b'D_1'W_1(K)D_1 b \le c_1}_{\substack{\text{small first difference}, \\ \text{after horizon K}}} &\text{and} \qquad \underbrace{b'D_3'W_2 D_3 b \le c_2}_{\text{small third difference}}. 
	\end{align*}
	However, this formulation does not have a closed form solution and is computationally more challenging, making it less appealing in practice.} 
for $\tilde\beta(\lambda_1,\lambda_2,K) := \tilde\beta(M)$ given by
\begin{align*}
  \tilde\beta(M) &= \left(V^{-1} + \lambda_1 D_1'W_1(K)D_1 + \lambda_2 D_3'W_3 D_3\right)^{-1} V^{-1}\hat{\beta} \\
  &= P(M) \hat{\beta}.
\end{align*}
For fixed $M = (\lambda_1, \lambda_2, K)$, it  immediately follows that 
\begin{align}\label{eq:normal}
	\tilde\beta(M) - P(M)\beta \sim N(0, V_{M}),
\end{align}
where $V_{M}=P(M) V_{\beta}P(M)'$.
Here, $P(M)\beta= \beta_M$ defines a particular surrogate path for $\beta$. Intuitively, $P(M)\beta$ corresponds to a ``projection'' of the true treatment path $\beta$ into a lower dimensional space. 
Given \eqref{eq:normal}, it would be straightforward to construct a confidence region for $\{\beta_{M,h}\}_{h=1}^H$, where $\beta_{M,h}$ denotes the $h^{\text{th}}$ entry in vector $\beta_{M}$, for a given, fixed value of the tuning parameters. However, knowing \textit{ex ante} what values to use for $\lambda_1$, $\lambda_2$, and $K$ seems challenging. We thus use model selection to choose $\lambda_1$, $\lambda_2$, and $K$ --- or, equivalently, to choose the surrogate model $M$. 

Specifically, we use the estimated $\hat\beta$ and an object akin to an information criterion to select the surrogate $M$. 
First, note that we can construct the ``residuals" $\hat\beta - \tilde\beta(M) = \hat\beta - P(M)\hat\beta = (I-P(M))\hat \beta$.
We use this residual formulation to define an analog of model degrees of freedom given by df$(M)$ = trace$(P(M))$.
We then select a model that minimizes a BIC analog over $\mathcal{M}$ where $\mathcal{M}$ denotes the universe of values for $M=(\lambda_1,\lambda_2,K)$: $$
\hat{M} = \argmin_{M \in \mathcal{M}} (\hat{\beta} - \tilde\beta(M))'V_{\beta}^{-1}(\hat{\beta} - \tilde\beta(M)) + \log(H) \text{df}(M).$$

We tie the researcher's hands by pre-specifying $\mathcal{M}$, the universe of models considered. In our implementation, $\mathcal{M}$ includes surrogate models corresponding to a constant, linear, quadratic and cubic treatment effect path (with one, two, three, and four degrees of freedom respectively), as well as an unrestricted model corresponding to the unrestricted estimates $\hat{\beta}$ (with $H$ degrees of freedom). $\mathcal{M}$ further includes surrogate models corresponding to surrogate paths of the form $P(M)\beta=\beta_M$ using a grid over $(\lambda_1,\lambda_2,K)$. We discuss our implementation in more detail in Appendix \ref{app-sec:algorithm} and visualize the model universe $\mathcal{M}$ for four exemplary treatment paths $\beta$ in Online Appendix Figure \ref{app-fig:model_universe}, but note that $\mathcal{M}$ does not depend on $\hat\beta$ or $\sigma$. 

\begin{remark}
	One could use other model universes and shrinkage methods. Examples of alternative approaches include \cite{barnichon2018functional} and \cite{barnichon2019}. We have chosen a class that we believe will be a reasonable representation of beliefs in many applications. Our approach with this model class is also particularly easy computationally, which allows us to nest a large universe of models. Likewise, one could select the surrogate model $M$ using methods other than minimizing our BIC analog. We found that the specific structure and estimation we employ performed well across our simulations and believe it provides a good default. 
\end{remark}

Given the selected surrogate $\hat{M}$, we define the \emph{restricted estimates} as $\tilde\beta(\hat{M})$. However, we cannot directly apply \eqref{eq:normal} to obtain a valid confidence region for the population value of the surrogate path $\beta_{\hat{M}}=P(\hat{M})\beta$ because $\hat{M}$ was selected by looking at the data, $\hat\beta$. Thus, in a second step, we use Valid Post-Selection Inference (\cite{berk2013valid}), which explicitly accounts for data-dependent (and thus random) model selection, to construct a uniformly valid confidence region for $\beta_{\hat{M}}$.  
These confidence intervals, our \emph{restricted plausible bounds}, are rectangular regions of the form $CR^{POSI} = \{\ell_h(X), u_h(X)\}_{h=1}^H$ for $[\ell_h(X), u_h(X)] = [\tilde{\beta}(\hat{M})_h \pm C^{\alpha} V_{\hat{M}}^{1/2}(h,h)]$ where $\tilde{\beta}(\hat{M})_h$ denotes the restricted estimate of the effect at horizon $h$ and $V_{\hat{M}}^{1/2}(h,h)$ is the square root of the $h^{\text{th}}$ diagonal entry of $V_{\hat{M}}$. 
To ensure uniform validity we use the ``PoSI constant'' of \cite{berk2013valid} as $C^\alpha$, defined as the minimal value that satisfies
 \begin{align*}
\mathbb{P}\left( \max_{M\in \mathcal{M}} \max_{h} |t_{h \cdot M}| \le C^\alpha \right) \ge (1-\alpha),
 \end{align*}
where $t_{h \cdot M} =  V_{M}^{-1/2}(h,h) \xi_h$, and $\xi_h$ is the $h^{th}$ element of multivariate normal vector $\xi$ with mean $\mathbf{0}_H$ and variance $V_M$.
Importantly, $C^{\alpha}$ depends on $\mathcal{M}$, the universe of models considered, but not on the model selection procedure.

The following proposition is a direct application of \cite{berk2013valid}.
\begin{proposition}\label{prop:surrogate}
For any treatment path $\beta$, we obtain valid coverage for its surrogate $\beta_{\hat{M}}$:
\begin{align}
\mathbb{P}[\beta_{\hat{M}} \in CR^{POSI}] \ge 1-\alpha. 
\end{align}
\end{proposition}
\begin{proof}
This follows immediately from the guarantees in \cite{berk2013valid}: 
\[
	\mathbb{P}(\beta_M \in CR^{POSI} | \hat{M}=M) \ge 1-\alpha.
\]
\end{proof}
Proposition \ref{prop:surrogate} guarantees that our restricted plausible bounds cover the selected surrogate to the truth in at least $(1-\alpha)$\% of sample realizations. 

\begin{remark}\label{rem:truth}
An immediate consequence of Proposition \ref{prop:surrogate} is that $\mathbb{P}(\beta \in CR^{POSI}) \ge 1-\alpha$ if $\mathbb{P}\left(\beta_{\hat{M}}=\beta_M=\beta\right)=1$. That is, in cases where model selection is effectively non-random and the selected surrogate path coincides with $\beta$, the restricted plausible bounds will also provide valid coverage for the true treatment path. Given the form of our BIC type objective for selecting $\hat{M}$ and that the unrestricted estimates are always included in our default model universe, one could provide conditions for $\mathbb{P}\left(\beta_{\hat{M}}=\beta_M=\beta\right)=1$ under a sequence of models where $\sigma^2 \rightarrow 0$ and surrogate paths were well-separated -- e.g. where $\|\beta - \beta_M\| \geq \delta > 0$ for all candidate models $M$ such that $\beta_M \neq \beta$. While technically possible, we i) question the utility of this perspective in offering a useful finite sample approximation and ii) view the surrogate as an economically interesting summary of the treatment path in itself.
\end{remark}

\begin{remark}
Throughout, we work with the unrestricted estimates $\hat{\beta}$. Our motivation is to augment visualizations of treatment effect paths within the setting where $\hat\beta \sim N(\beta,V_{\beta})$ provides a reasonable approximation. We further note that the interpretation of treatment effect plots as currently displayed with point estimates and pointwise confidence intervals essentially relies on exactly this approximation. An alternative would be to estimate the restricted models directly on the data. In settings where $\hat\beta \sim N(\beta,V_{\beta})$ provides a good approximation, we suspect that such an approach will yield qualitatively similar results. It may be interesting to explore directly estimating restricted models in settings where the approximation $\hat\beta \sim N(\beta,V_{\beta})$ is questionable.
\end{remark}

\section{Numerical Results}\label{sec:sim}

In this section, we illustrate the properties of our restricted plausible estimates and bounds as well as our cumulative plausible bounds in simulation experiments with treatment paths generated to resemble treatment path dynamics that practitioners may encounter. We illustrate these treatment paths in Figure \ref{fig:DGPs}. We consider a constant treatment effect path (cf. Figure \ref{fig:DGPs_constant}); a treatment path that smoothly declines before flattening out after 17 periods (cf. Figure \ref{fig:DGPs_quadratic});  a hump-shaped treatment path with dynamics that continue for the entire $H$ periods (cf. Figure \ref{fig:DGPs_noflat}); and a ``wiggly'' treatment path (cf. Figure \ref{fig:DGPs_wiggly}). We describe the exact DGP for each of the four panels in more detail in Online Appendix Table \ref{app-tab:DGPs}.\footnote{Figure \ref{fig:new_viz} provides one example realization from each of these DGPs.} The object of interest is the treatment path over a 36-month horizon. We have access to jointly normal estimates $\{\hat{\beta}_h\}_{h=1}^{36}$. 
For each of these four treatment paths, we then draw 1,000 realizations of $\hat{\beta} \sim N(\beta,V_{\beta})$.\footnote{In the figures that follow, $V_{\beta}$ is diagonal with its entries specified in Online Appendix \ref{app-sec:add_details}. We repeat our exercise with more general covariance matrices in Online Appendix \ref{app-sec:general_cov}.}
\begin{figure}[tb!]
\centering
\begin{subfigure}[h]{0.49\textwidth}
\includegraphics[width=\linewidth]{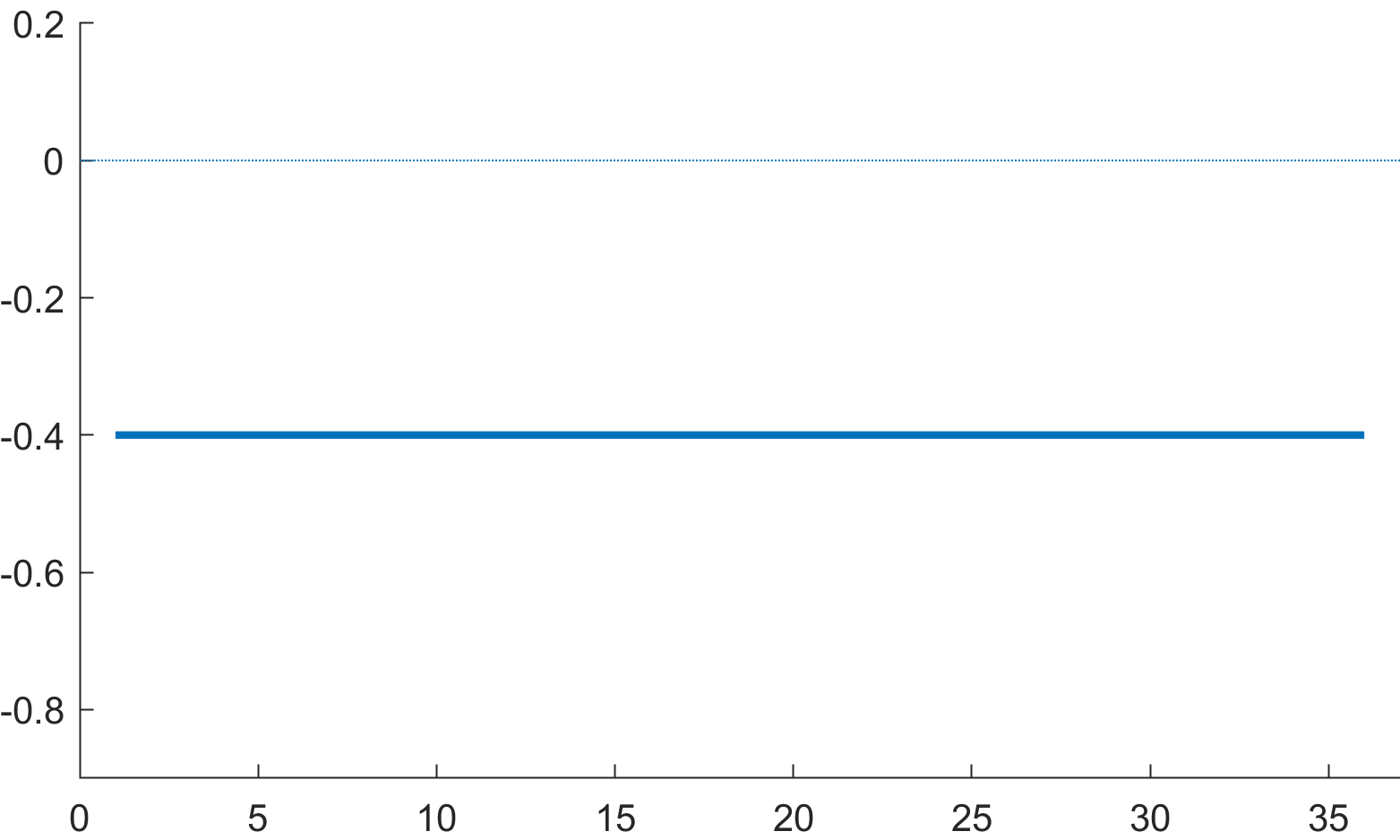}
\caption{Constant Treatment Effect}
 \label{fig:DGPs_constant}
\end{subfigure}
\hfill
\begin{subfigure}[h]{0.49\textwidth}
\includegraphics[width=\linewidth]{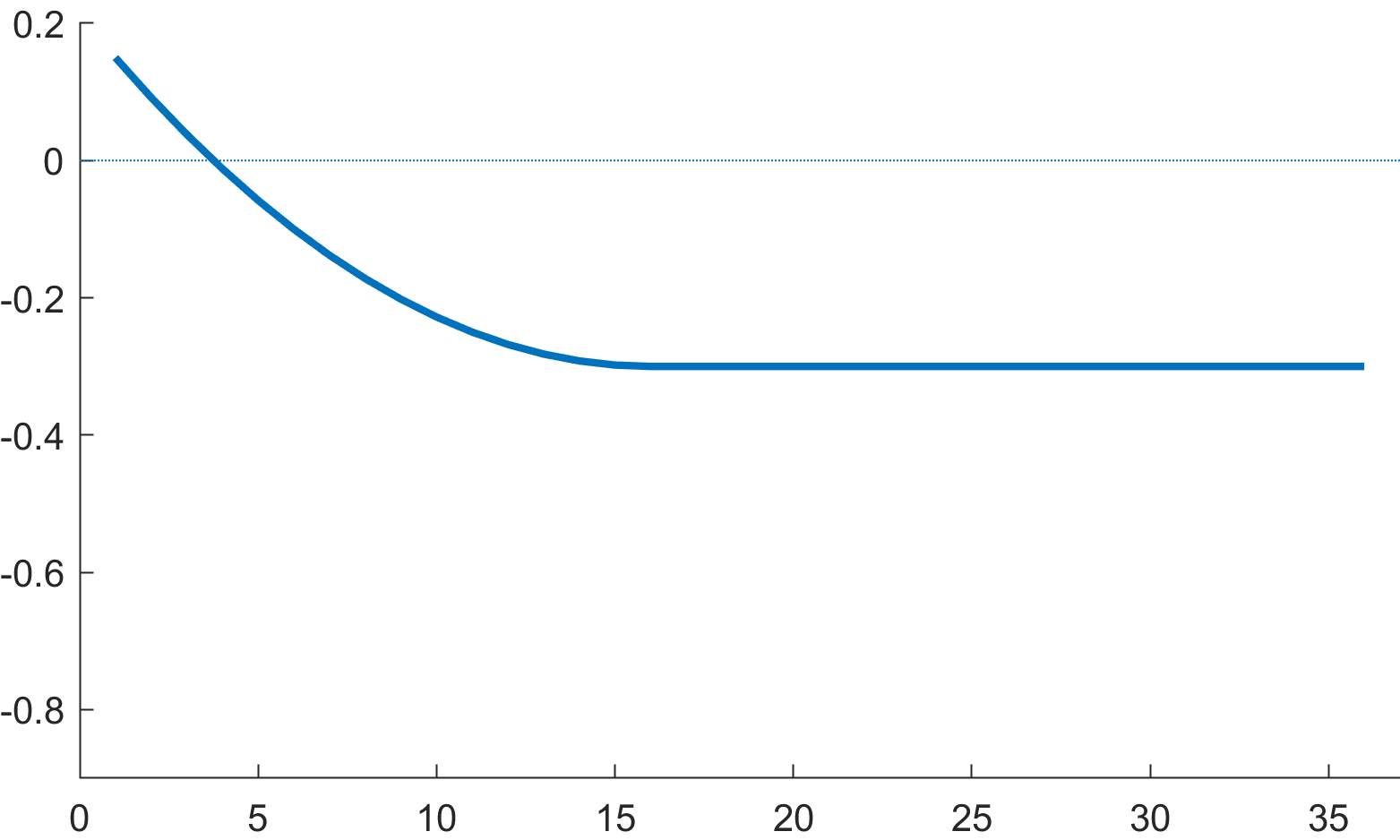}
\caption{Smooth, eventually flat}
 \label{fig:DGPs_quadratic}
\end{subfigure}

\begin{subfigure}[h]{0.49\textwidth}
\includegraphics[width=\linewidth]{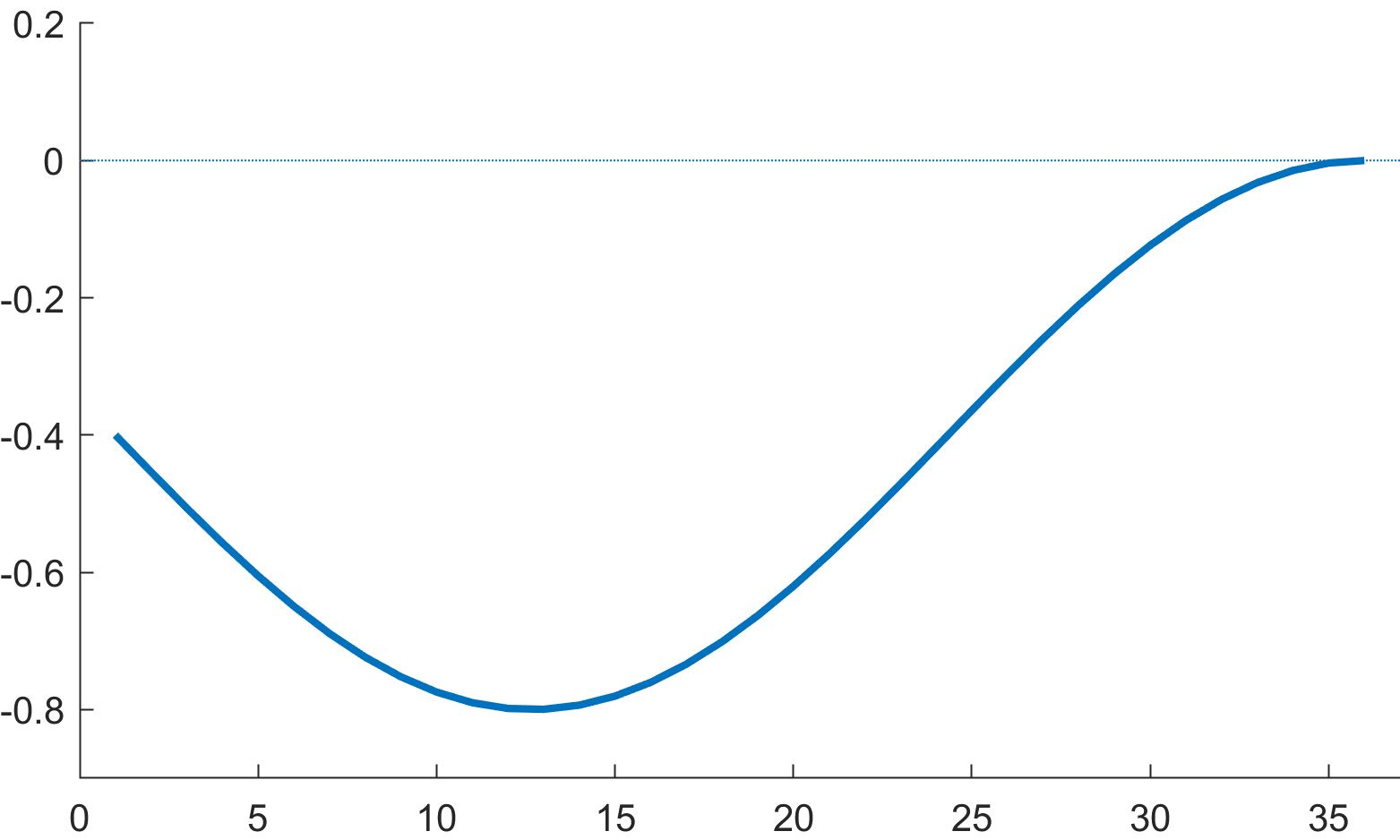}
\caption{Hump-shaped}
 \label{fig:DGPs_noflat}
\end{subfigure}
\hfill
\begin{subfigure}[h]{0.49\textwidth}
\includegraphics[width=\linewidth]{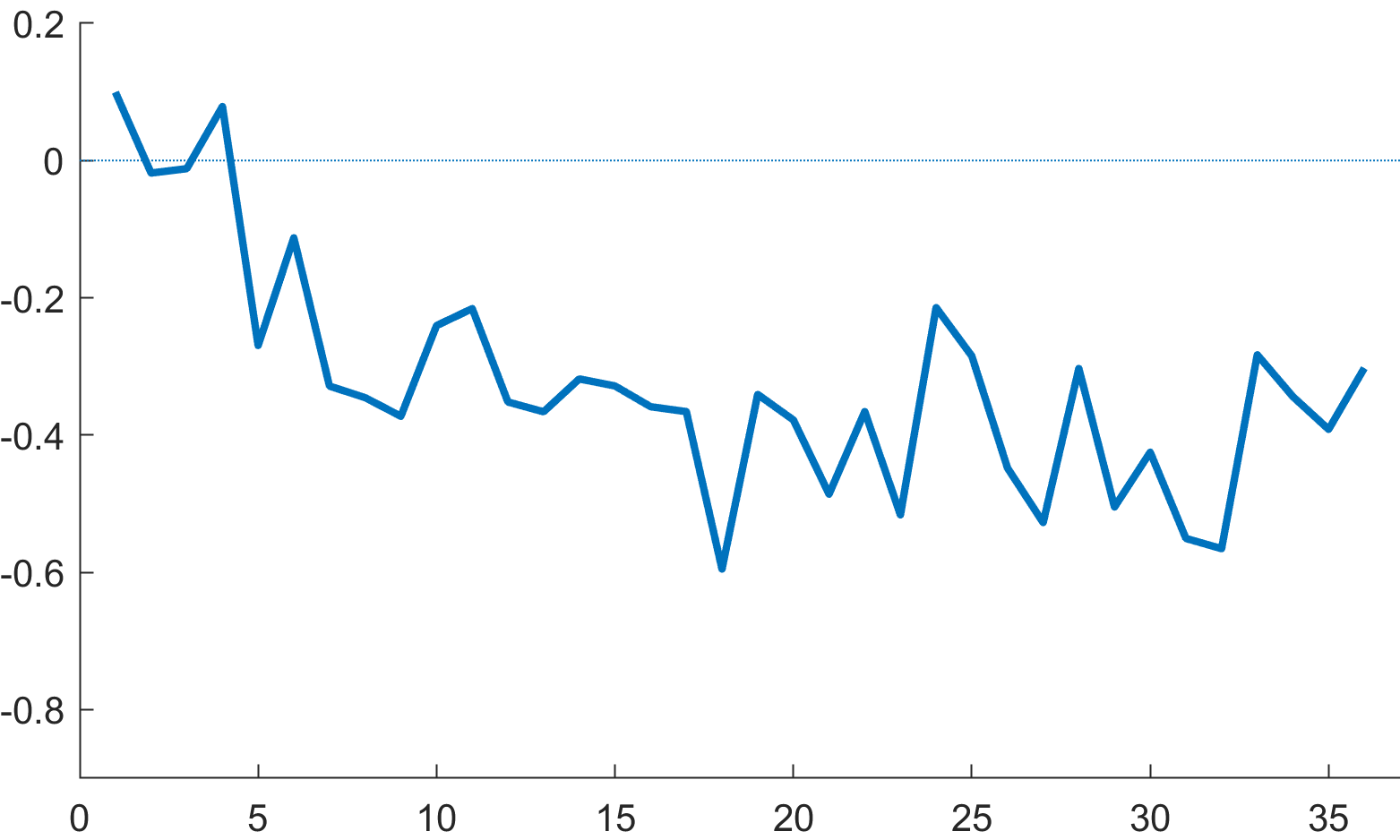}
\caption{Wiggly}
 \label{fig:DGPs_wiggly}
\end{subfigure}
\caption{Four exemplary treatment effect paths $\beta$.}
 \label{fig:DGPs}
\end{figure}

We first compare the point estimation properties of the unrestricted estimates $\hat{\beta}$ with our restricted estimates $\tilde{\beta}(\hat{M})$ for each of these four scenarios.  In particular, Figure \ref{fig:MSEs} depicts the ratio in mean-squared error, $MSE_{\tilde{\beta}(\hat{M})}/MSE_{\hat{\beta}}$, as a function of $\sigma^2$, which scales the covariance matrix of the estimates, $V_\beta$ (see Online Appendix \ref{app-sec:add_details} for more detail).  
\begin{figure}[tb!]
\centering
\begin{subfigure}[h]{0.49\textwidth}
\includegraphics[width=\linewidth]{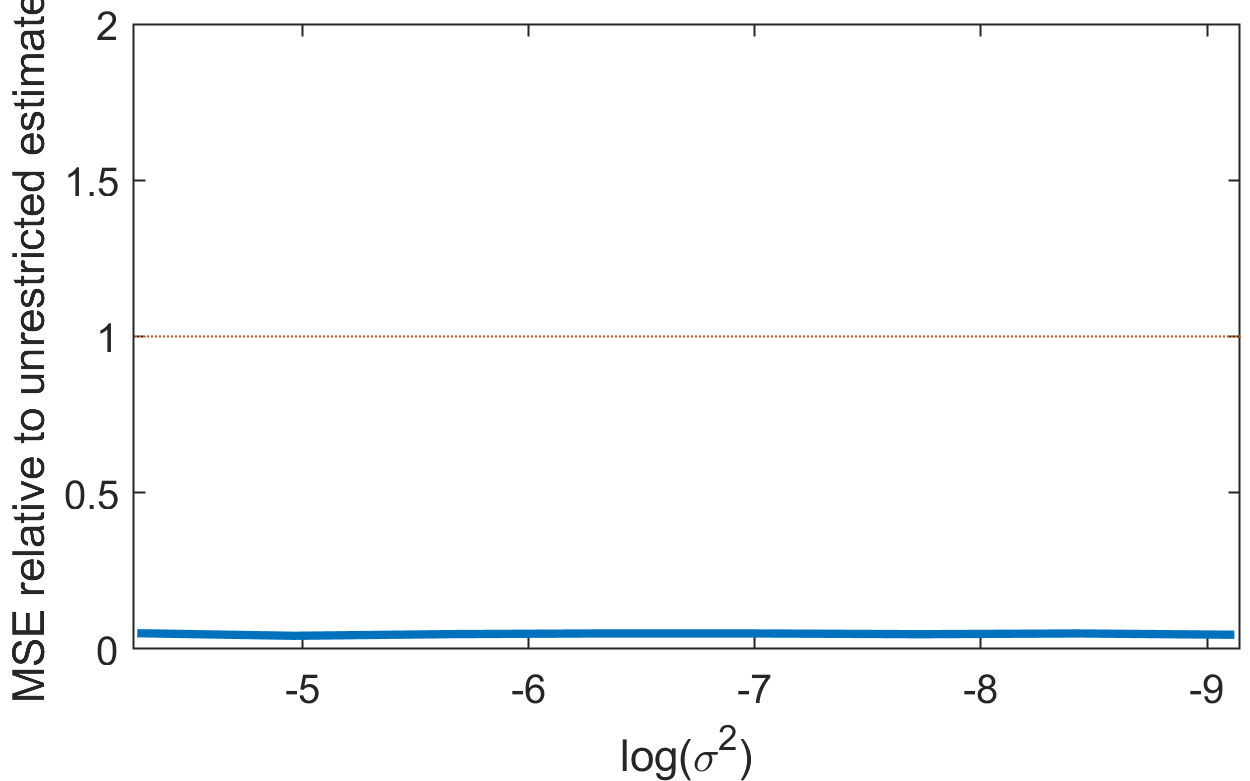}
\caption{Constant Treatment Effect}
 \label{fig:mse_constant}
\end{subfigure}
\hfill
\begin{subfigure}[h]{0.49\textwidth}
\includegraphics[width=\linewidth]{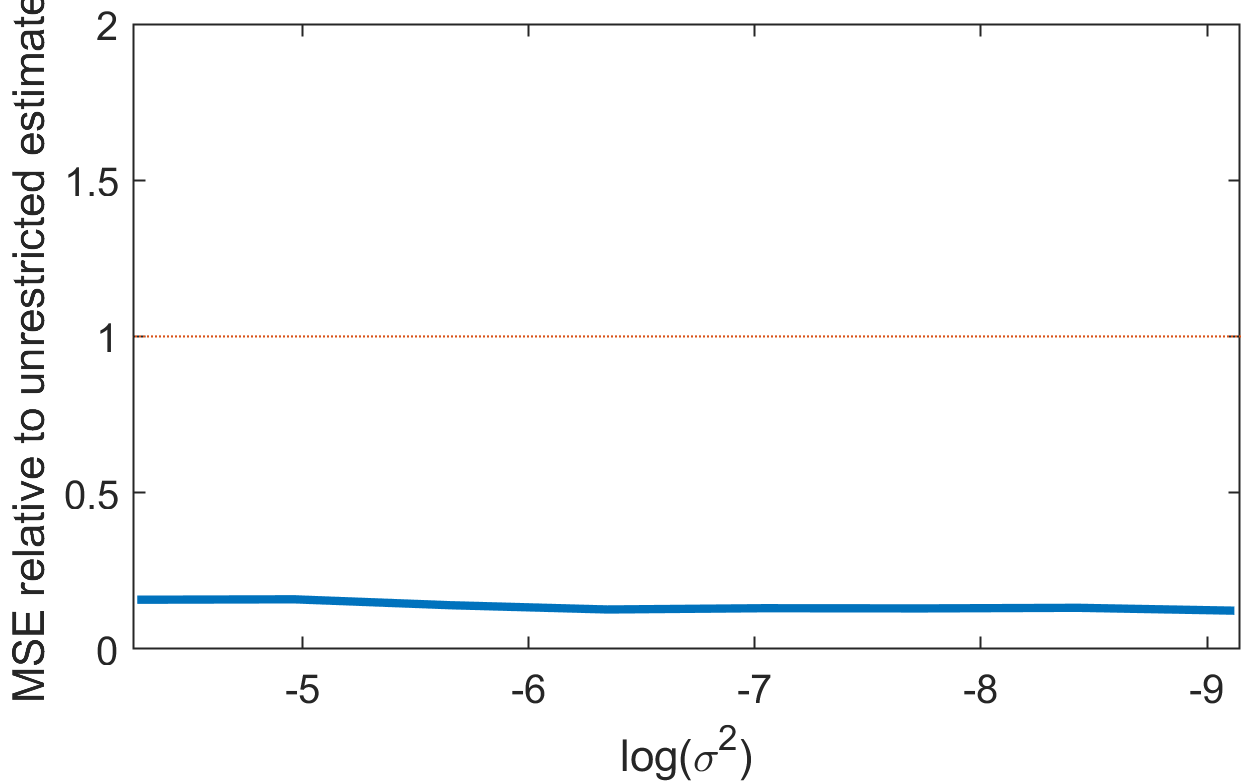}
\caption{Smooth, eventually flat}
 \label{fig:mse_quadratic}
\end{subfigure}
\begin{subfigure}[h]{0.49\textwidth}
\includegraphics[width=\linewidth]{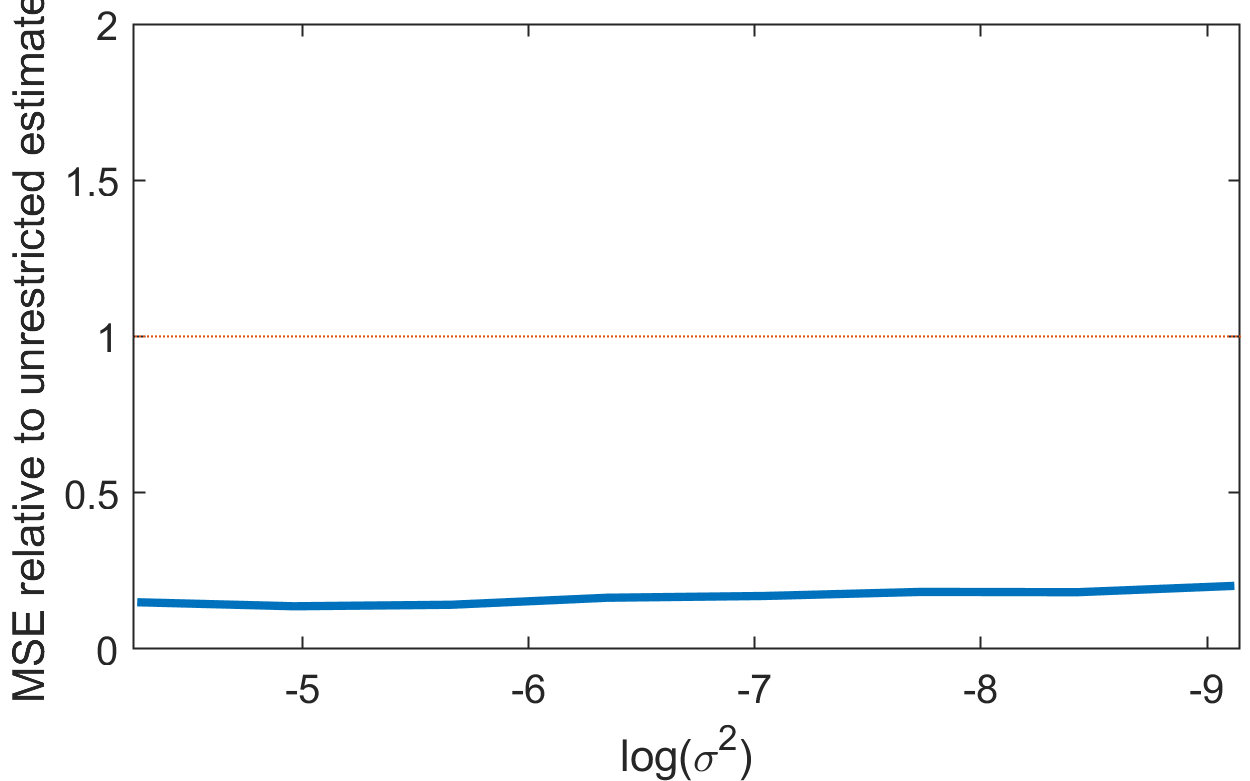}
\caption{Hump-shaped}
 \label{fig:mse_noflat}
\end{subfigure}
\hfill
\begin{subfigure}[h]{0.49\textwidth}
\includegraphics[width=\linewidth]{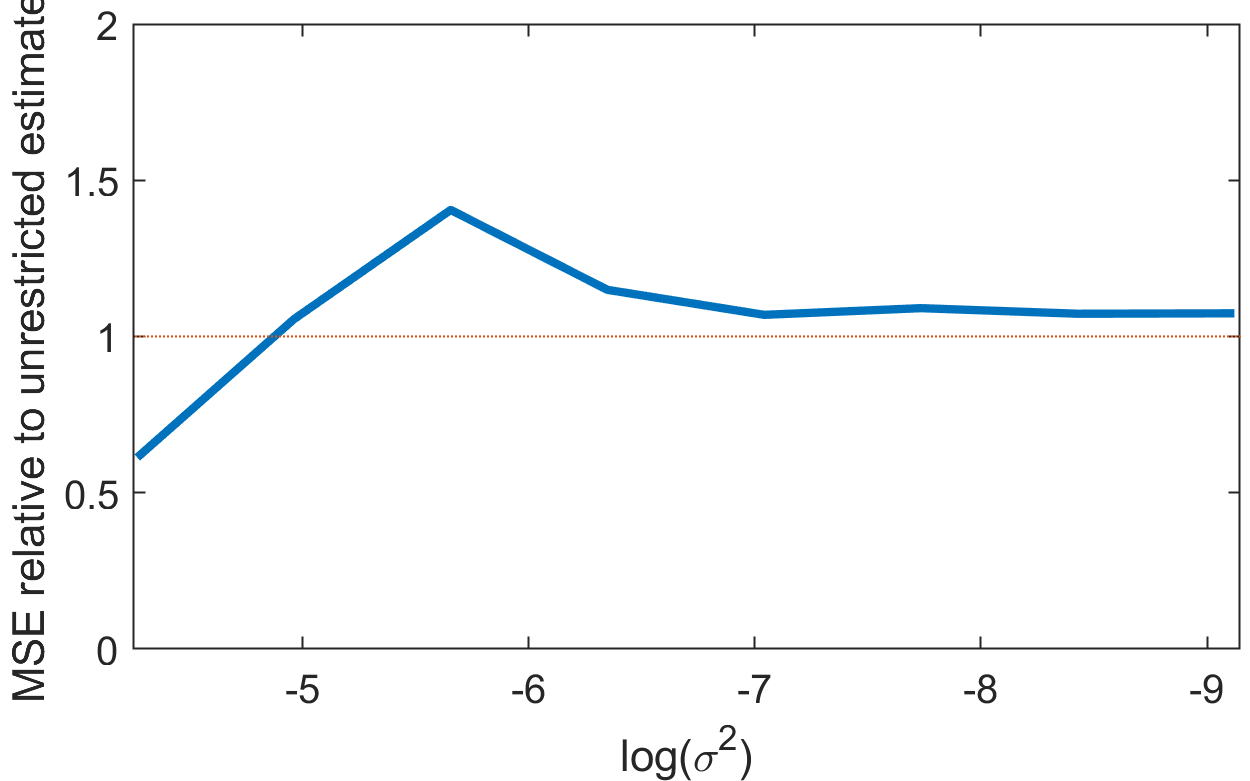}
\caption{Wiggly}
 \label{fig:mse_wiggly}
\end{subfigure}
\caption{Relative performance of restricted and unrestricted estimators. Depicted is the ratio in mean-squared error between restricted and unrestricted estimates, $\frac{MSE(\tilde{\beta}(\hat{M}))}{MSE(\hat{\beta})}$, as a function of the amount of noise $\sigma^2$ in the initial estimates $\hat{\beta}$.}
 \label{fig:MSEs}
\end{figure}
The largest value of $\sigma^2$ in Figure \ref{fig:MSEs} (corresponding to the left most point) thus represents relatively noisy estimates.\footnote{Figure \ref{fig:new_viz} was created from a single realization of the the left most point in Figure \ref{fig:MSEs}. To give the reader a sense of the scale of the x-axes, Online Appendix Figure \ref{app-fig:new_viz_largen} also illustrates a single realization of the right most point in Figure \ref{fig:MSEs}.}
We conclude that in most cases our restricted estimator has excellent point estimation properties when the target is the true treatment path $\beta$. In the three panels that have a smooth treatment path (Figures \ref{fig:mse_constant}-\ref{fig:mse_noflat}), the MSE of our restricted estimate is a full order of magnitude lower compared to the unrestricted estimate. In Figure \ref{fig:mse_wiggly}, our restricted estimate has a lower MSE when estimates are very noisy, a higher MSE for intermediate sizes of $\sigma^2$, and a similar MSE when $\sigma^2$ is small. However, we suspect that a lower-dimensional summary of $\beta$, as provided by the surrogate path, may in fact be the policy relevant object in cases where the true treatment path exhibits complicated dynamics as in this panel (cf. Figure \ref{fig:DGPs_wiggly}).

We next turn our attention to inference and first look at coverage. We again consider the four previous DGPs and vary the amount of noise in the estimation of $\beta$ by varying $\sigma^2$. The result is depicted in Figure \ref{fig:coverage}, where we set $\alpha=0.05$.
\begin{figure}[tb!]
\centering
\begin{subfigure}[h]{0.49\textwidth}
\includegraphics[width=\linewidth]{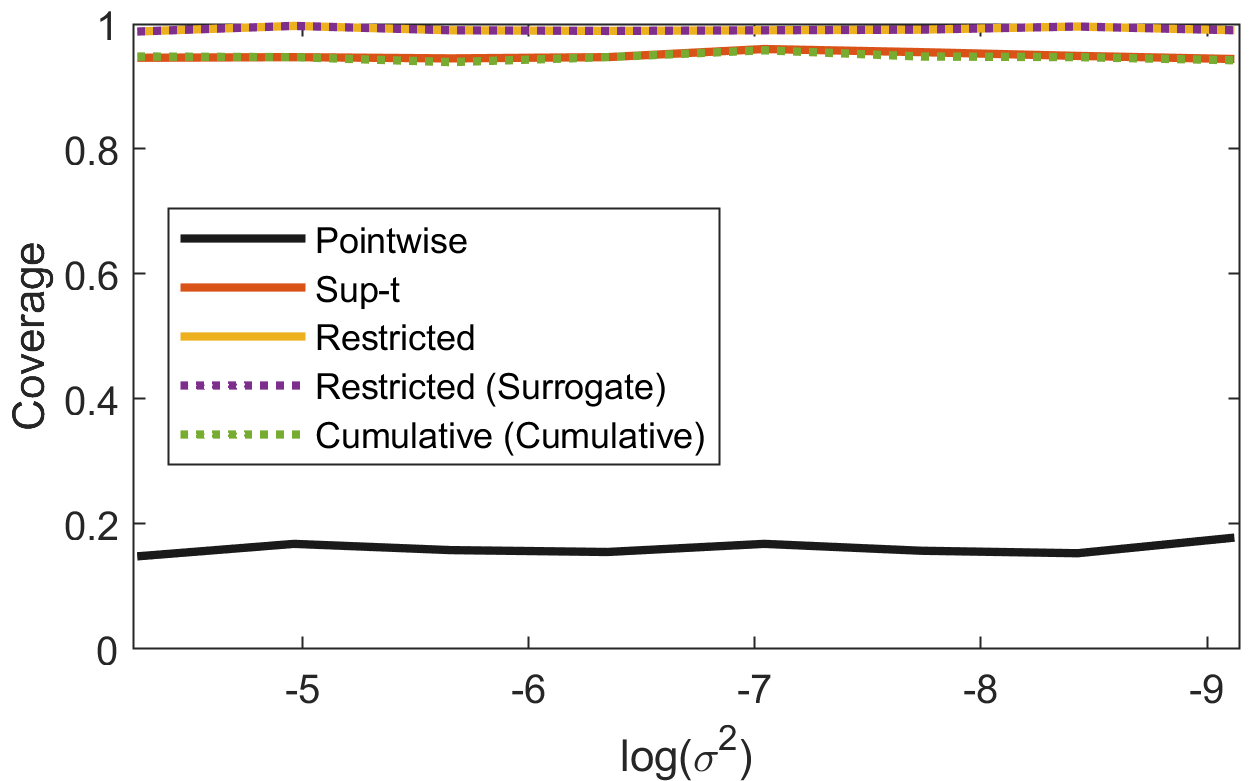}
\caption{Constant Treatment Effect}
 \label{fig:coverage_constant}
\end{subfigure}
\hfill
\begin{subfigure}[h]{0.49\textwidth}
\includegraphics[width=\linewidth]{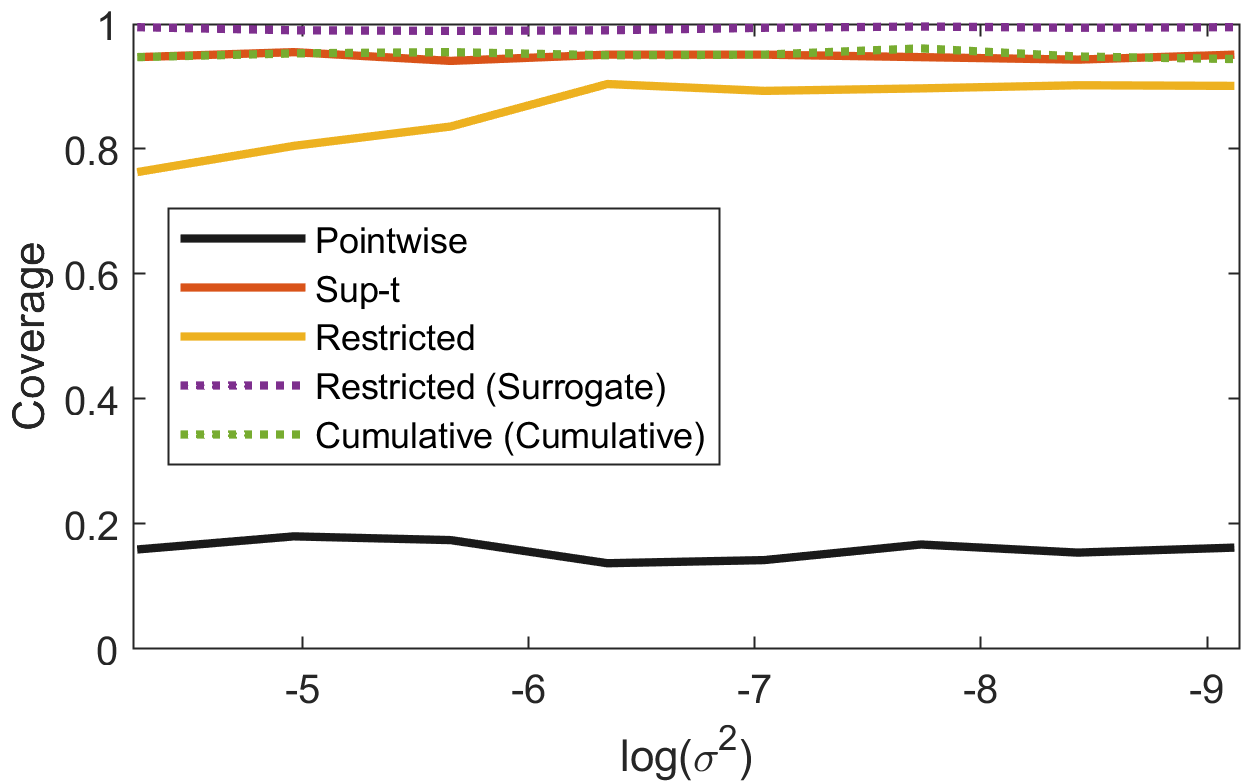}
\caption{Smooth, eventually flat}
 \label{fig:coverage_quadratic}
\end{subfigure}
\begin{subfigure}[h]{0.49\textwidth}
\includegraphics[width=\linewidth]{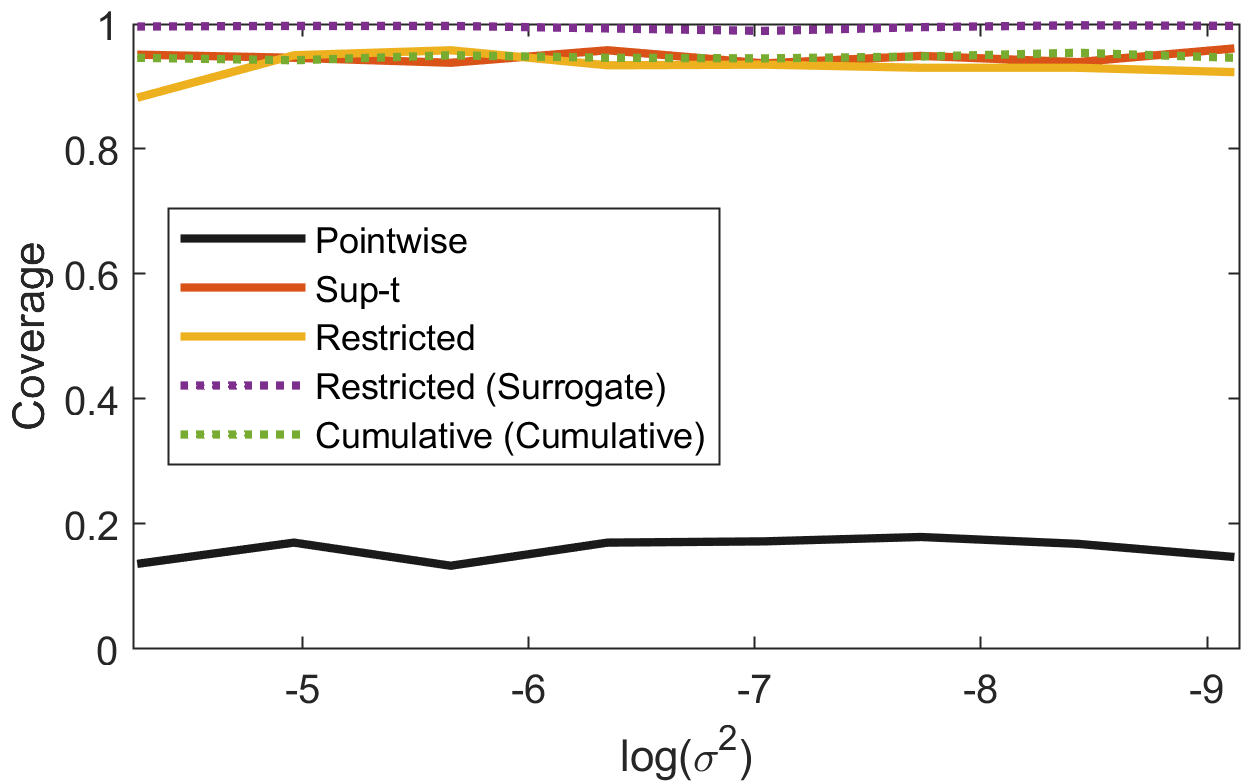}
\caption{Hump-shaped}
 \label{fig:coverage_noflat}
\end{subfigure}
\hfill
\begin{subfigure}[h]{0.49\textwidth}
\includegraphics[width=\linewidth]{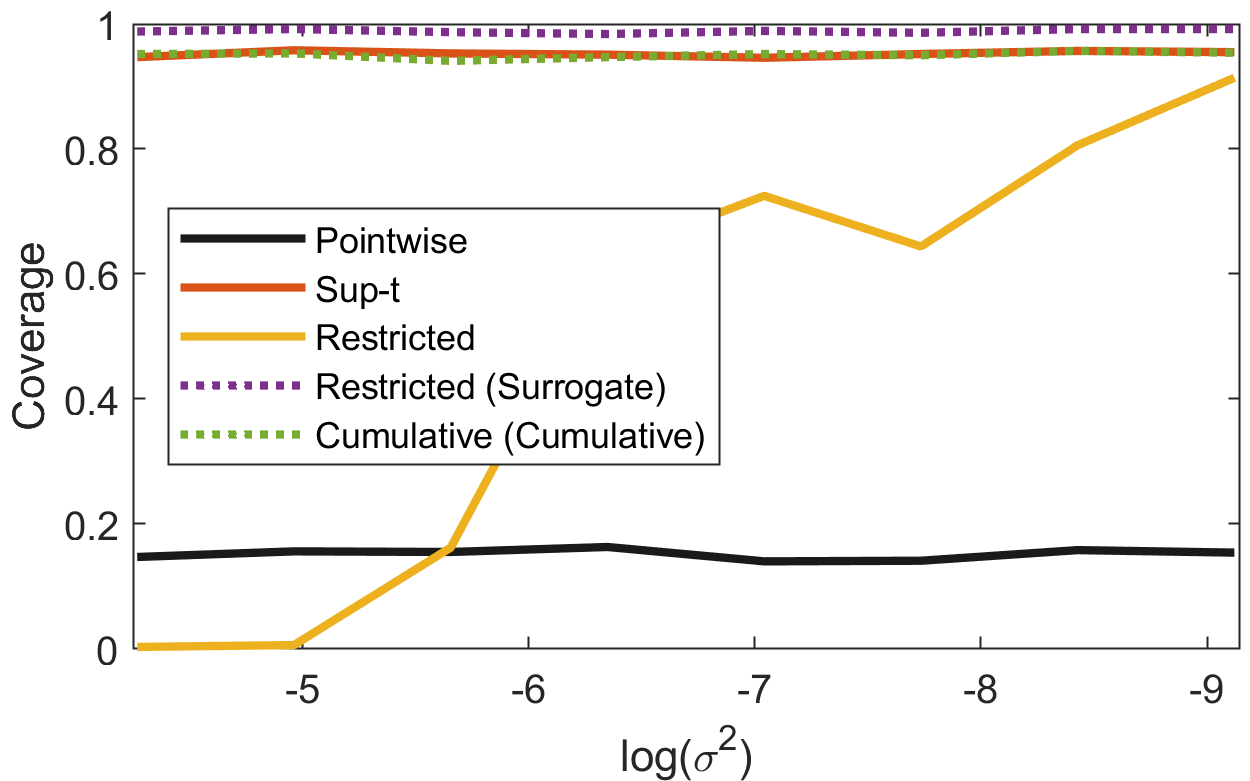}
\caption{Wiggly}
 \label{fig:coverage_wiggly}
\end{subfigure}
\caption{Coverage properties of various confidence regions as a function of the amount of noise $\sigma^2$ in the initial estimates $\hat{\beta}$.}
 \label{fig:coverage}
\end{figure}
The pointwise, sup-t, and restricted coverage numbers, as indicated by the three solid lines, represent the empirical analogue to the usual notion of uniform coverage for $\beta$: It reflects the fraction of simulations in which the true treatment path $\beta$ falls entirely inside the pointwise confidence intervals (black), sup-t intervals (red), and restricted plausible bounds (yellow), respectively.

Across all four DGPs, the pointwise confidence intervals only cover the true treatment path in 15-20\% of simulations. On the other hand, the sup-t intervals achieve nominal coverage across DGPs. The restricted plausible bounds are not constructed to provide uniform coverage for the true treatment path. It is thus not surprising that these bounds do not cover the entire treatment path in 95\% of simulations across all DGPs. However, we do see that their coverage appears to converge towards 95\% as the amount of noise in the initial estimates decreases (cf. Remark \ref{rem:truth}).
Further, the restricted plausible bounds perform substantially better than the pointwise confidence intervals in covering the true treatment path when the true path is smooth. In the scenario of a wiggly treatment path, the restricted bounds exhibit poor coverage properties when the amount of noise in the initial estimates is large.

Proposition \ref{prop:surrogate} guarantees that the plausible bounds cover the selected surrogate to the truth in at least 95\% of realizations. In Figure \ref{fig:coverage}, the restricted plausible bounds' coverage of the surrogate is illustrated by the dashed purple line. We see that this coverage is above 95\% for all levels of $\sigma^2$ across all DGPs. We reemphasize that, in the scenario of a wiggly treatment path, a smooth approximation to the true path, for which we obtain valid coverage, may be a policy relevant object. 

Finally, we note that, in line with Proposition \ref{prop:wald}, the cumulative bounds (as indicated by the dashed green lines) indeed cover the cumulative effect of the policy in 95\% of simulations. That is, the true treatment path is on average within these bounds in 95\% of simulations.

In order to cover the true path in a higher fraction of simulations, the sup-t bands are generally wider than the pointwise bands (cf. Figure \ref{fig:illustration}). We illustrate this across our simulations in Figure \ref{fig:ratio_widths}, which depicts the average width of the sup-t bands and the restricted bounds, relative to the pointwise intervals.  
\begin{figure}[thb!]
\centering
\begin{subfigure}[h]{0.49\textwidth}
\includegraphics[width=\linewidth]{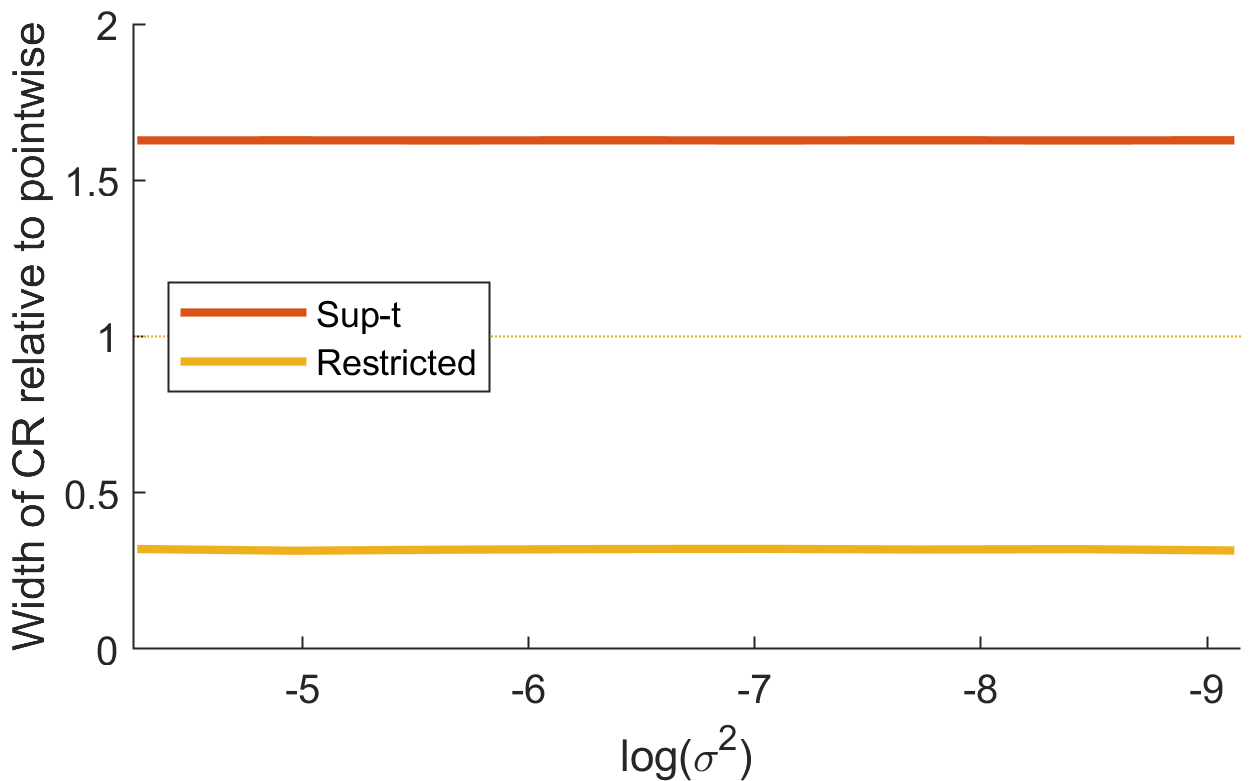}
\caption{Constant Treatment Effect}
 \label{fig:relwidth_constant}
\end{subfigure}
\hfill
\begin{subfigure}[h]{0.49\textwidth}
\includegraphics[width=\linewidth]{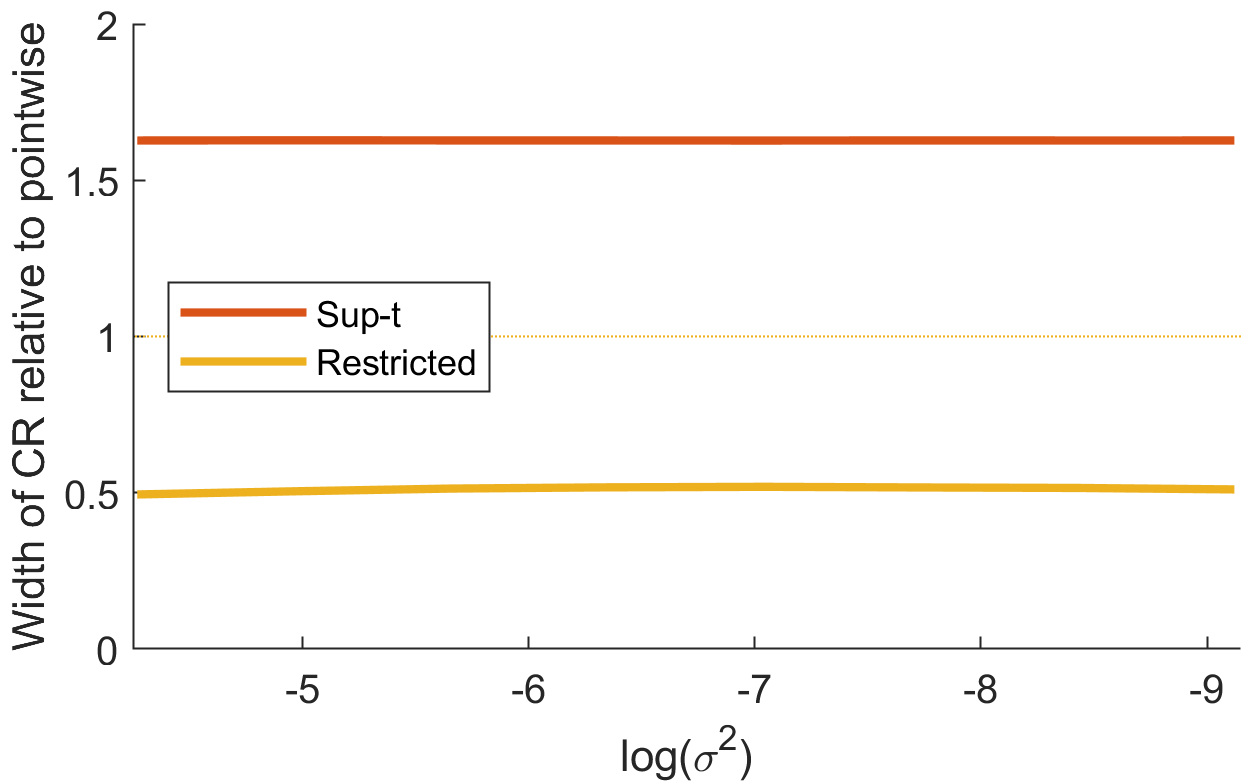}
\caption{Smooth, eventually flat}
 \label{fig:relwidth_quadratic}
\end{subfigure}

\begin{subfigure}[h]{0.49\textwidth}
\includegraphics[width=\linewidth]{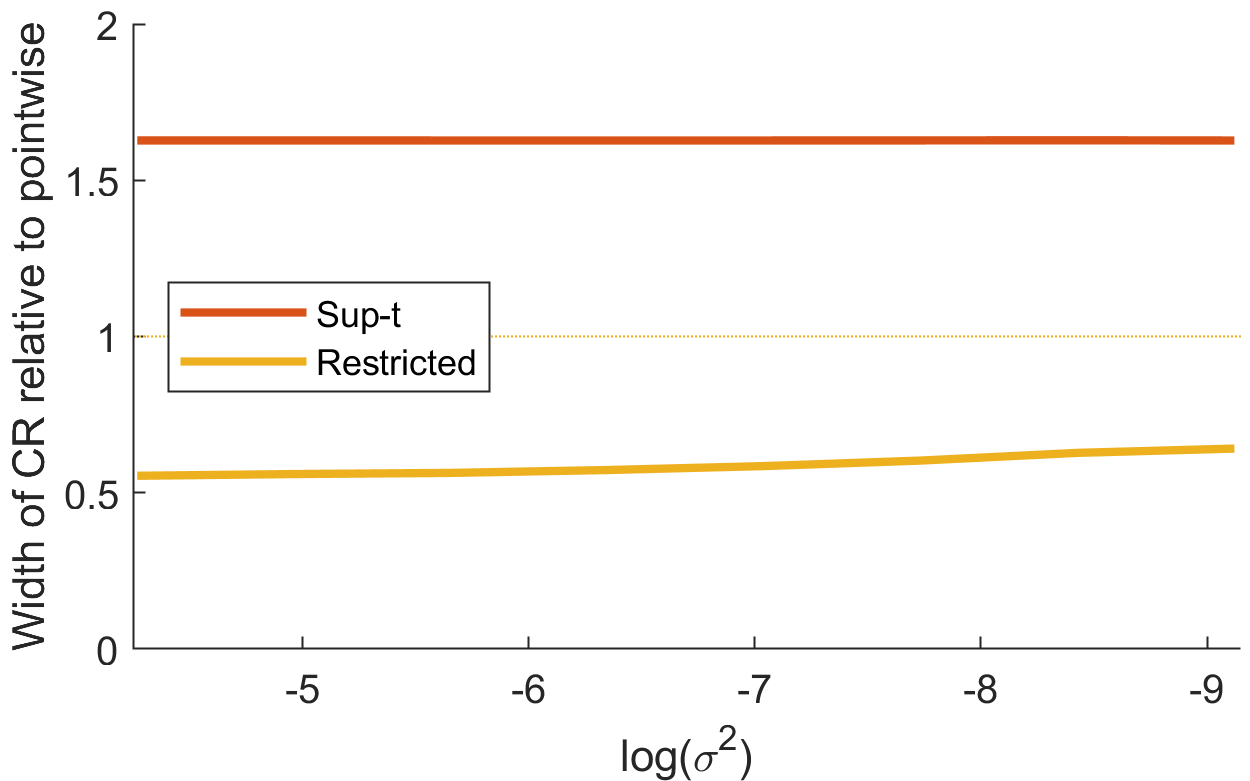}
\caption{Hump-shaped}
 \label{fig:relwidth_noflat}
\end{subfigure}
\hfill
\begin{subfigure}[h]{0.49\textwidth}
\includegraphics[width=\linewidth]{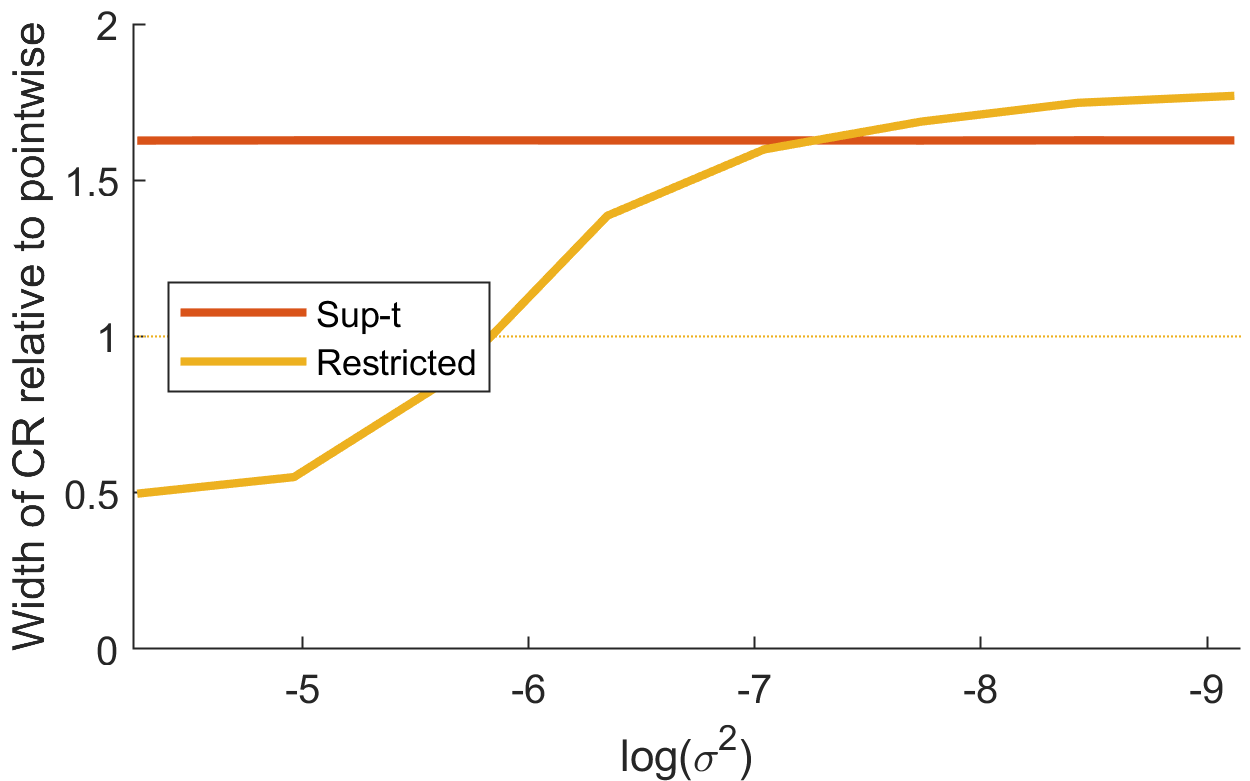}
\caption{Wiggly}
 \label{fig:relwidth_wiggly}
\end{subfigure}
\caption{Average width of confidence regions relative to pointwise confidence intervals as a function of the amount of noise $\sigma^2$ in the initial estimates $\hat{\beta}$.}
 \label{fig:ratio_widths}
\end{figure}
While the sup-t bands are wider than the pointwise bands across the parameter space, the restricted bounds are in fact narrower throughout much of the parameter space, despite their improved coverage properties. For example, in the ``Constant Treatment Effect" scenario, the restricted bounds are less than half as wide as the pointwise bands and less than a quarter as wide as the sup-t bands for all noise levels $\sigma^2$. Yet, they cover close to 95\% of all realizations, while the pointwise confidence intervals only cover the true treatment path in 15-20\% of simulations.\footnote{As we show in Online Appendix Figure \ref{app-fig:df}, we select surrogate paths $\beta_M$ with significantly reduced degrees of freedom, relative to the unrestricted estimates, in the first three DPGs. This effective dimension reduction explains both the improvement in point estimation properties of $\tilde{\beta}(\hat{M})$ relative to the unrestricted estimates $\hat\beta$ documented in Figure \ref{fig:MSEs}, and the decrease in width of the confidence intervals documented in Figure \ref{fig:ratio_widths}.}
The restricted bounds are only wider than the sup-t bands when (i) the treatment path is wiggly and (ii) there is very little noise. In this region of the parameter space, our model selection mechanism infers that the wiggles in the treatment path are large enough relative to the noise level that heavy smoothing will incur a large cost in terms of losing fit. With relatively little smoothing, our restricted plausible bounds are then slightly more conservative than the standard sup-t bands, which follows immediately from the results in \cite{berk2013valid}.

Finally, we shed some more light on the model selection of our procedure in Figure \ref{fig:surrogates}. 
\begin{figure}[tb!]
\centering
\begin{subfigure}[h]{0.49\textwidth}
\includegraphics[width=\linewidth]{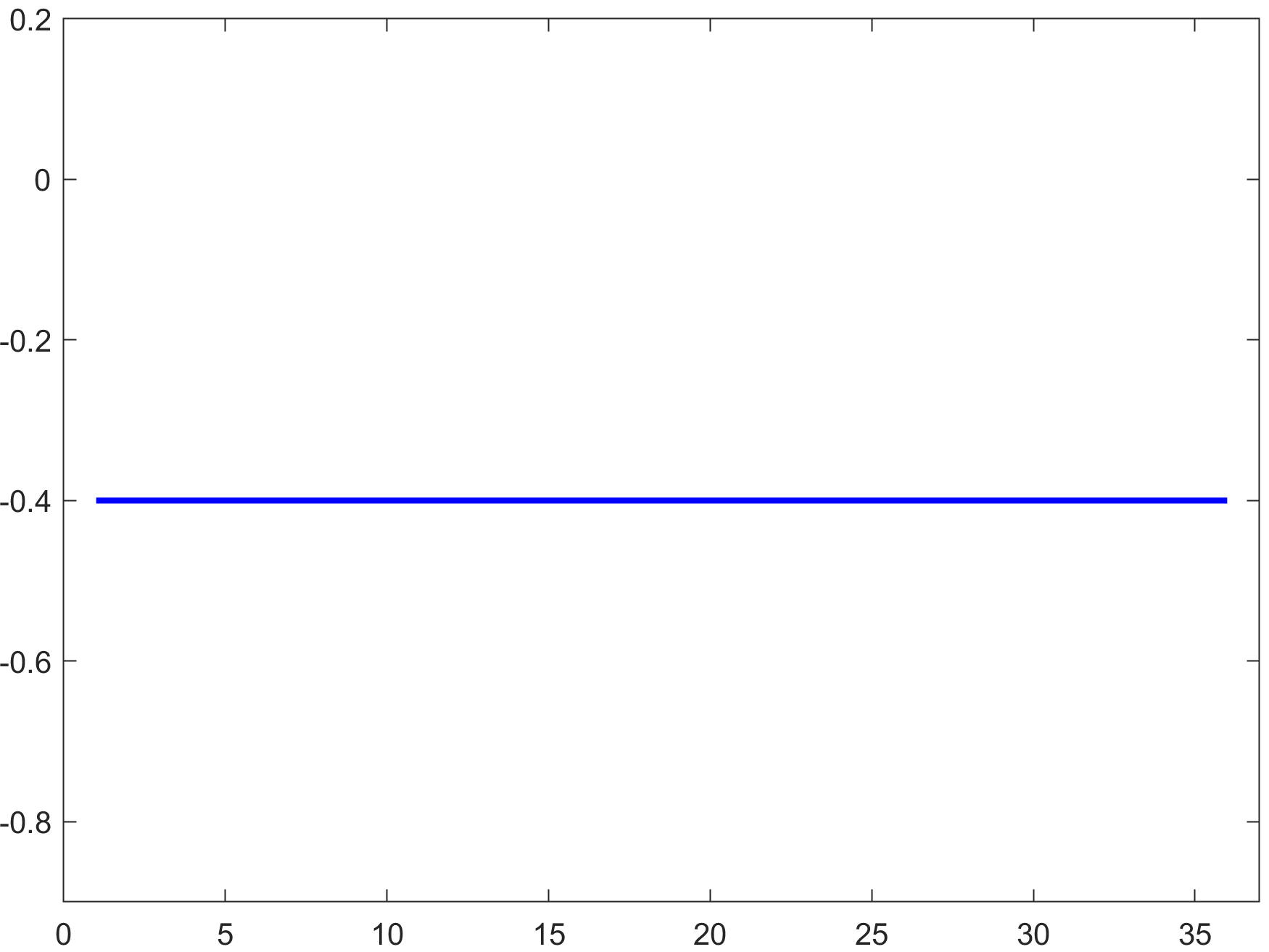}
\caption{Constant Treatment Effect}
 \label{fig:surrogates_constant}
\end{subfigure}
\hfill
\begin{subfigure}[h]{0.49\textwidth}
\includegraphics[width=\linewidth]{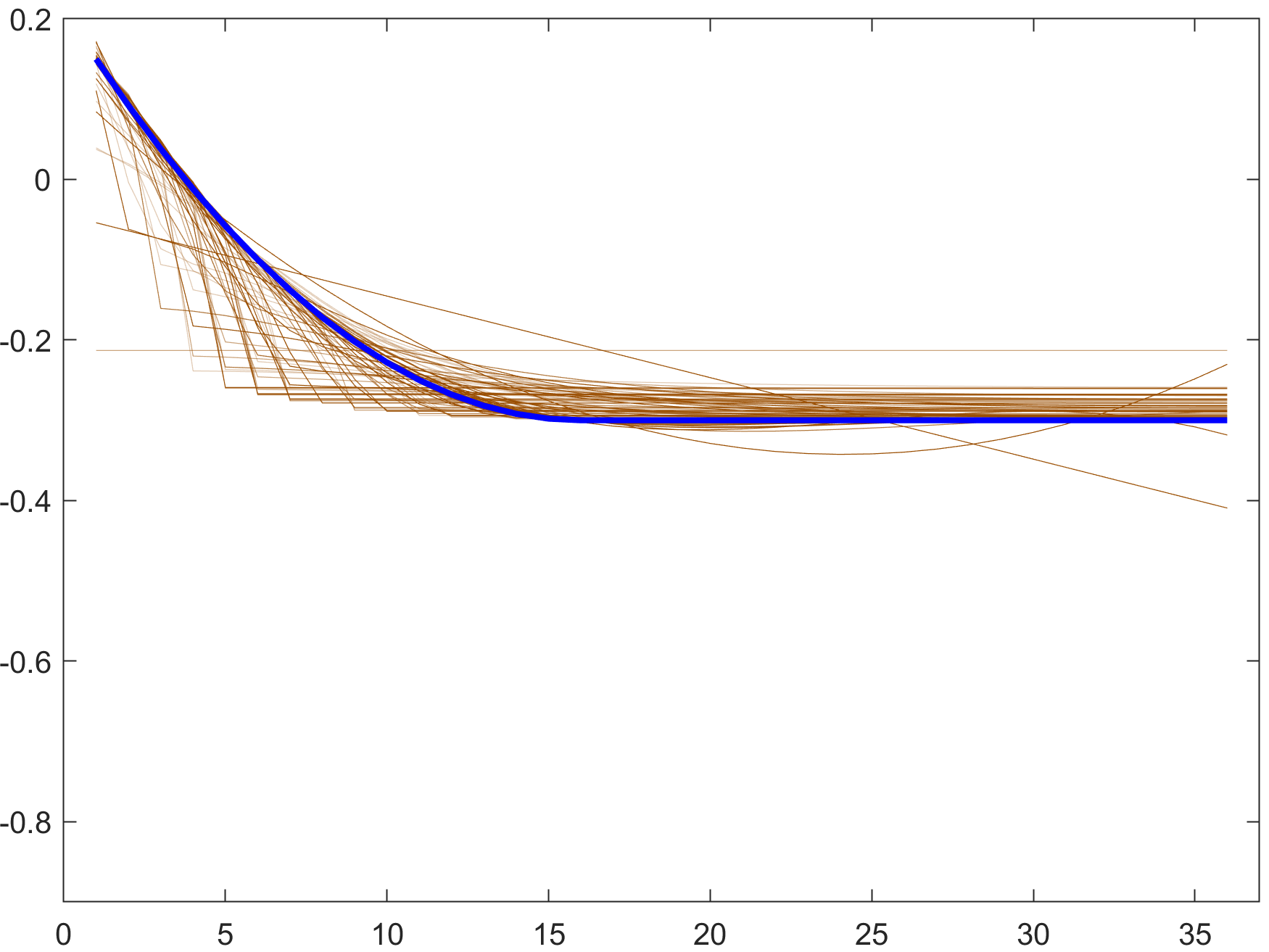}
\caption{Smooth, eventually flat}
 \label{fig:surrogates_quadratic}
\end{subfigure}

\begin{subfigure}[h]{0.49\textwidth}
\includegraphics[width=\linewidth]{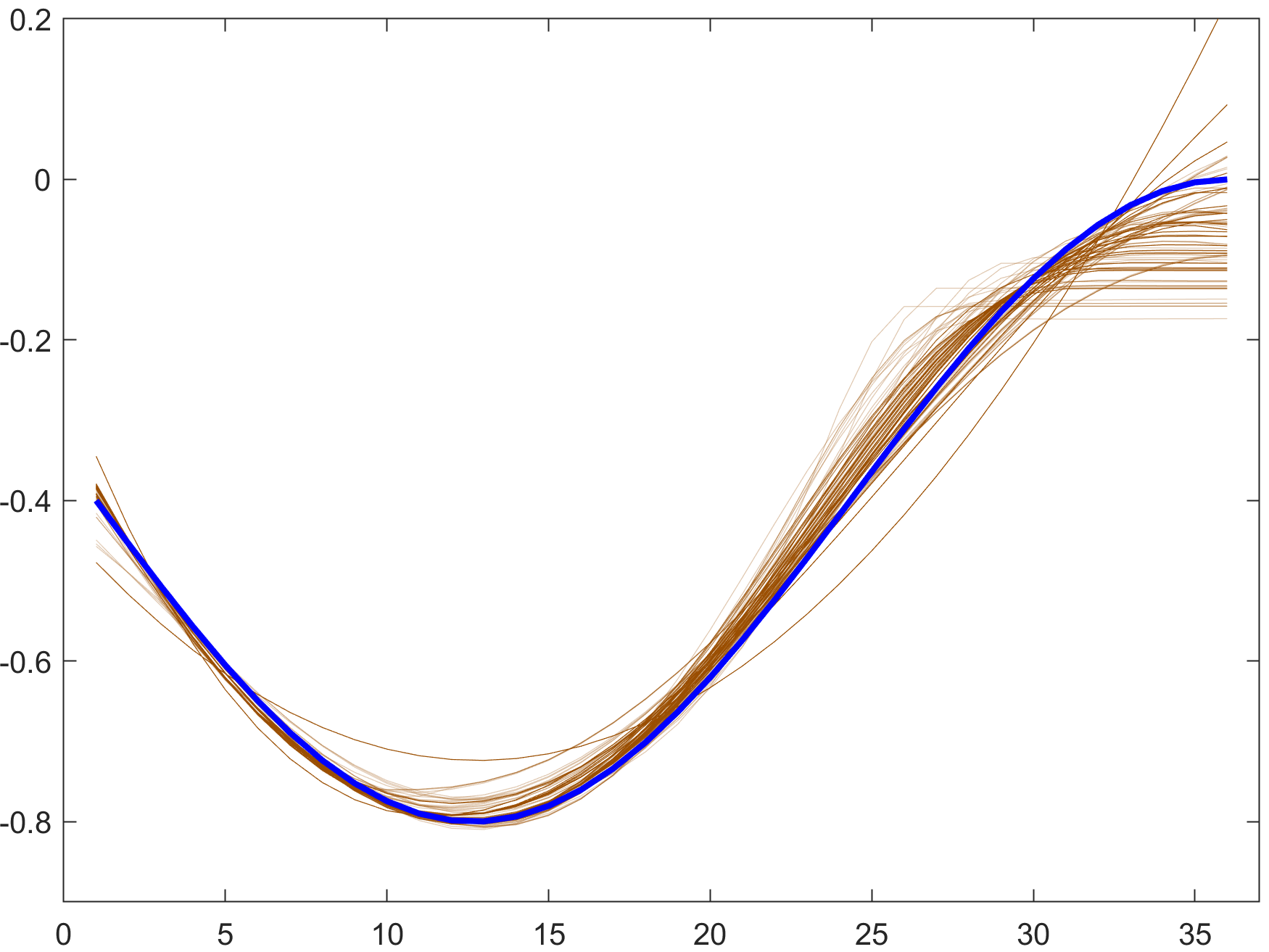}
\caption{Hump-shaped}
 \label{fig:surrogates_noflat}
\end{subfigure}
\hfill
\begin{subfigure}[h]{0.49\textwidth}
\includegraphics[width=\linewidth]{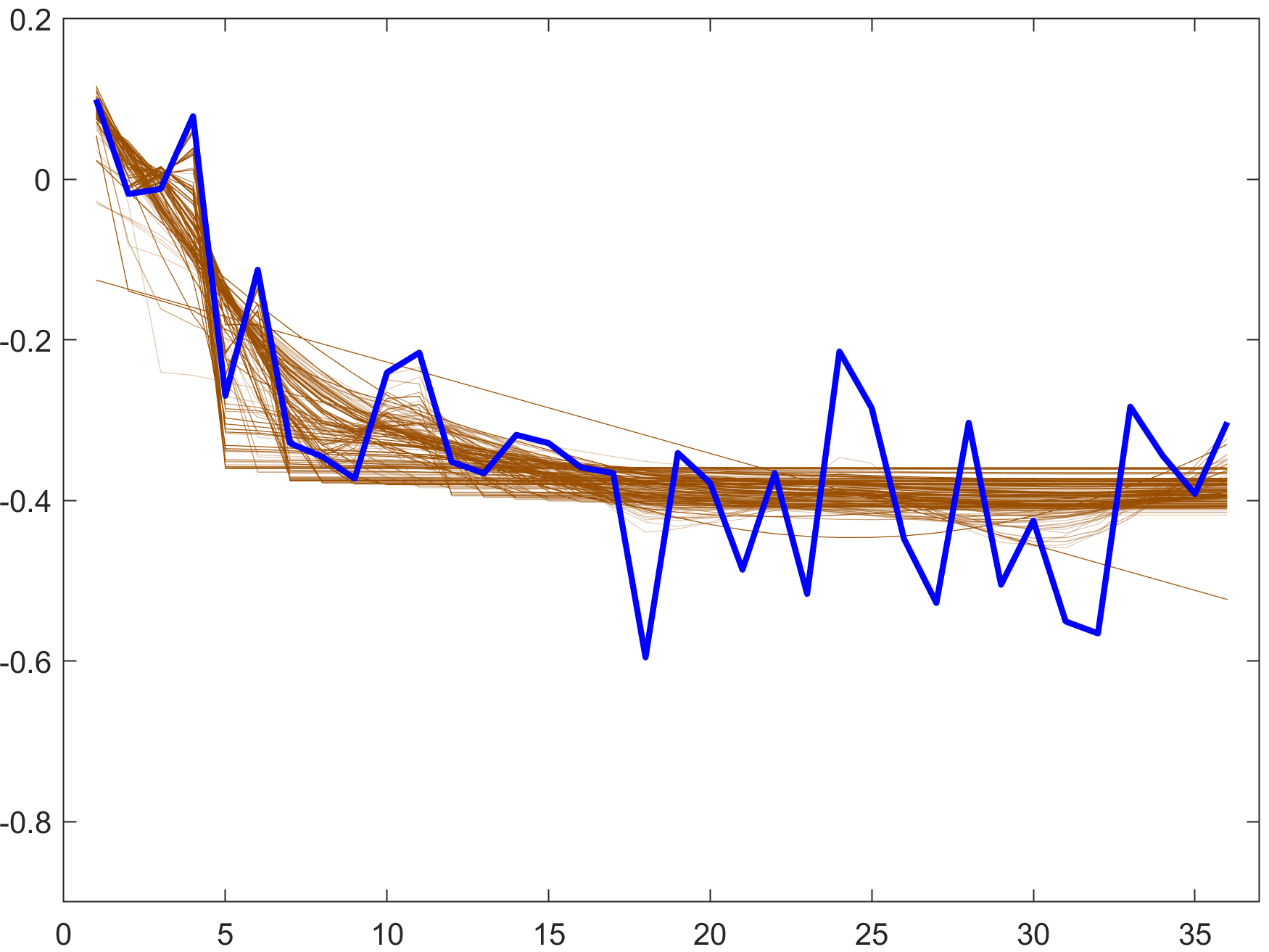}
\caption{Wiggly}
 \label{fig:surrogates_wiggly}
\end{subfigure}
\caption{Illustration of the chosen surrogates across our 1,000 simulations for $\sigma^2=0.014$ ($\log(\sigma^2)=-4.27$), corresponding to the left most point in Figures \ref{fig:MSEs}-\ref{fig:ratio_widths}, and thus a large amount of noise in $\hat{\beta}$. }
 \label{fig:surrogates}
\end{figure}
Each panel depicts the 1,000 selected surrogates across our simulations in light brown. Since the chosen surrogates will depend on the precision of the initial estimates $\hat{\beta}$, we fix $\sigma^2$ at $0.014$, corresponding to the amount of noise in the examples in Figure \ref{fig:new_viz}. When the true treatment path is smooth (Panels \ref{fig:surrogates_constant}-\ref{fig:surrogates_noflat}), the chosen surrogate tends to closely approximate the truth.  In contrast, the chosen surrogate is often simpler than $\beta$ while retaining its main features when the true treatment path is not smooth (Panel \ref{fig:surrogates_wiggly}).\footnote{In Online Appendix Figure \ref{app-fig:surrogates_wiggly_noise}, we further illustrate how the chosen surrogates depend on $\sigma^2$ for our Wiggly DGP.}

\section{Conclusion}

We are interested in (joint inference on) the dynamic treatment effect of a policy. We propose two visualization devices, which we term restricted and cumulative plausible bounds, to include in standard visualizations of estimated treatment effect paths. Because both bounds explicitly target objects other than uniformly covering the entire treatment path, they can be substantially tighter than standard pointwise and uniform confidence bands. Our bounds may thus provide useful insights about features of treatment effect paths when traditional confidence bands appear uninformative. 
Our bounds can also lead to markedly different conclusions relative to traditional visualizations when there is significant correlation in the estimates. 
As an auxiliary benefit, producing our restricted plausible bounds also provides additional restricted estimates that capture simplified representations of the dynamic effect path. Unsurprisingly, these restricted estimates have improved point estimation properties relative to the unrestricted estimates in settings where the treatment effect path is smooth.

It may be interesting to explore the notion of explicitly covering data-dependent surrogates 
in other economically relevant settings. While alternative model universes or model selection rules may be needed outside of the present setting, our proposal can easily be adapted to do so. For example, it would be interesting to explore related ideas in more structural settings.

\bibliographystyle{plainnat}
\bibliography{bounds}

\newpage

\noindent \textbf{\Huge Online Appendix}

\appendix

\section{A primer on confidence regions}\label{app-sec:crs}

If $\beta$ is a scalar, the standard approach in economics to quantify and visualize the uncertainty associated with an estimate for $\beta$ is to construct a confidence interval. For a chosen size $\alpha$, such a confidence interval covers the true value $\beta$ in $100*(1-\alpha)$\% of all realizations of the data: $ \mathbb{P}(\ell(X)< \beta < u(X)) = 1-\alpha$ where $\ell(X)$ and $u(X)$ denote the lower and upper bounds of the confidence interval and observed data are a realization from random variable $X$. Intuitively, these intervals visualize to the reader what values of $\beta$ are ``plausible'' based on the observed data. The idea being that values inside this confidence interval appear ``plausible,'' while values outside of the interval do not. More formally, values outside this interval are rejected by a standard t-test at level $\alpha$, while values inside the interval are not rejected. 

Since in this paper a dynamic treatment path is the object of interest, $\beta=\{\beta_h\}_{h={1}}^{H}$ is an ordered vector instead.
We start with a diagram in Online Appendix Figure \ref{app-fig:diagram} that illustrates standard methods in the case of a two dimensional parameter $\beta=(\beta_1, \beta_2)$ where estimates are $\hat\beta = (\hat\beta_1,\hat\beta_2) = (2,1)$ and $V_\beta = I_2$ is the 2 $\times$ 2 identity matrix.

\begin{figure}[tb]
	\centering
	\includegraphics[width=.7\linewidth]{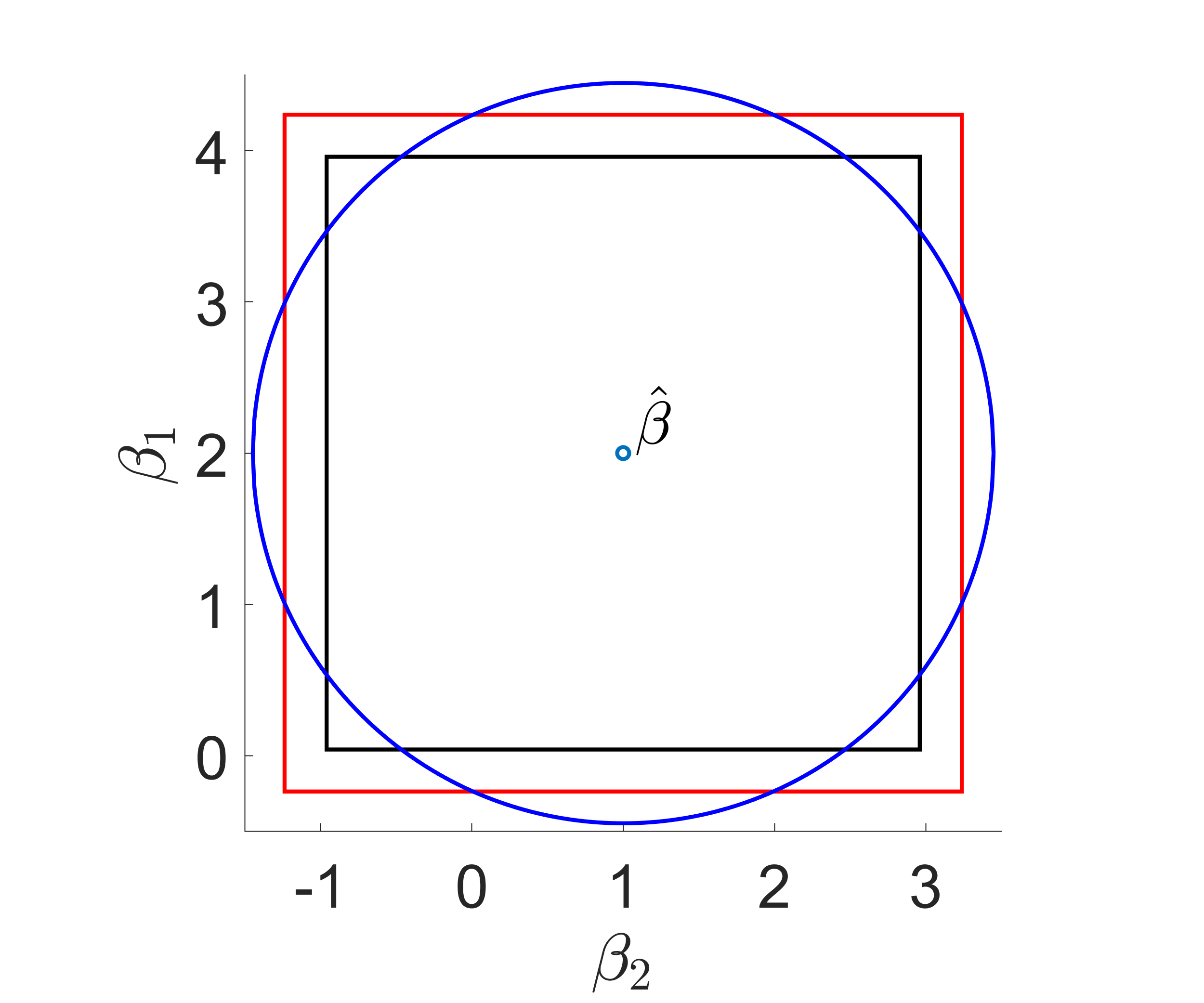}
	\caption[Illustration of different confidence regions.]{Illustration of different confidence regions. Pointwise (black), sup-t (red), and Wald (blue) 95\% confidence region in two dimensions. $\hat{\beta}=(2,1)$ with covariance matrix $V_\beta=I_2$.}
	\label{app-fig:diagram}
\end{figure}

The predominant practice today is to include pointwise confidence intervals in depictions of estimated treatment effect paths.  $100*(1-\alpha)$\% pointwise intervals for a specific $\beta_h$ simply correspond to choosing $(\ell_h(X),u_h(X)) = (\hat\beta_h - z_{1-\alpha/2}\sqrt{V_\beta[h,h]},\hat\beta_h + z_{1-\alpha/2}\sqrt{V_\beta[h,h]})$ where $V_\beta[h,h]$ is the variance of $\hat\beta_h$ and $z_{1-\alpha/2}$ is the $1-\alpha/2$ quantile of a standard normal distribution. For example, the pointwise 95\% confidence intervals in the case of the example in Online Appendix Figure \ref{app-fig:diagram} for $\beta_1$ and $\beta_2$ are respectively 2 $\pm$ 1.96 and 1 $\pm$ 1.96. Correspondingly, the black square depicts the Cartesian product of these pointwise confidence intervals for $\beta_1$ and $\beta_2$. Denote the region provided by the black square as $CR^{pw}$.
Treated as a confidence region for ($\beta_1, \beta_2$), $CR^{pw}$ \textit{a}) ignores any information in the off-diagonal entries in the covariance matrix of $\hat{\beta}$ and \textit{b}) is only valid for testing pre-specified hypotheses involving single coefficients. Thus, it does not achieve correct coverage for the vector $\beta=(\beta_1, \beta_2)$: For a chosen significance level $\alpha$, it will generally be true that $\mathbb{P}(\beta \in CR^{pw}) = \mathbb{P}(\ell_h(X)< \beta_h < u_h(X) \ \forall \ h) < (1-\alpha)$, such that the black square will generally cover the true parameter $\beta$ in less than $100(1-\alpha)$\% of realizations of the data. For example, if $Cov(\hat{\beta_1}, \hat{\beta_2})=0$, the probability that the pointwise confidence region covers the vector $\beta$ will be $P(\beta \in CR^{pw})= (1-\alpha)^2$.

One way to construct a uniformly valid confidence region is to take the Cartesian product of sup-t confidence intervals (depicted in red) instead. Denote this region $CR^{sup-t}$. Sup-t intervals are easy to construct, and simply use a slightly large critical value compared to pointwise confidence intervals. Specifically, $100(1 - \alpha)$\% sup-t intervals are constructed by choosing $\big(\ell_h(X),u_h(X)\big) = \big(\hat\beta_h - c_{\alpha}\sqrt{V_{\beta}[h,h]},\hat\beta_h + c_{\alpha}\sqrt{V_{\beta}[h,h]}\big)$ where $c_{\alpha}$ is set such that $\mathbb{P}(\ell_h(X)< \beta_h < u_h(X) \; \forall \ h) \ge (1-\alpha)$. For a chosen significance level $\alpha$, $CR^{sup-t}$ thus achieves valid coverage, since $\mathbb{P}(\beta \in CR^{sup-t}) \ge (1-\alpha)$ by construction. For example, the sup-t 95\% confidence intervals in the case of the example in Online Appendix Figure \ref{app-fig:diagram} for $\beta_1$ and $\beta_2$ are respectively 2 $\pm$ 2.24 and 1 $\pm$ 2.24. See, e.g., \cite{freyberger2018} and \cite{olea:pm} for details about sup-t interval construction as well as further discussion of different (rectangular) confidence regions. We focus on pointwise and sup-t confidence intervals, since these two are the predominant intervals used in practice (e.g. \cite{callaway2021}; \cite{jorda2023local}; \cite{boxell2024}).  

Finally, the blue circle in Online Appendix Figure \ref{app-fig:diagram} corresponds to an alternative confidence region for $\beta$, namely the Wald confidence region. Denote this region $CR^{Wald}$. This region simply collects all parameter values $(b_1,b_2)$ that are not rejected by a standard Wald test of the null hypothesis that $(\beta_1,\beta_2) = (b_1,b_2)$ at level $\alpha$. For a chosen significance level $\alpha$, this region also achieves valid coverage: $\mathbb{P}(\beta \in CR^{Wald}) = (1-\alpha)$ by construction.

\begin{figure}[tb]
	\begin{subfigure}[t]{.49\textwidth}
		\centering
		\includegraphics[width=\linewidth]{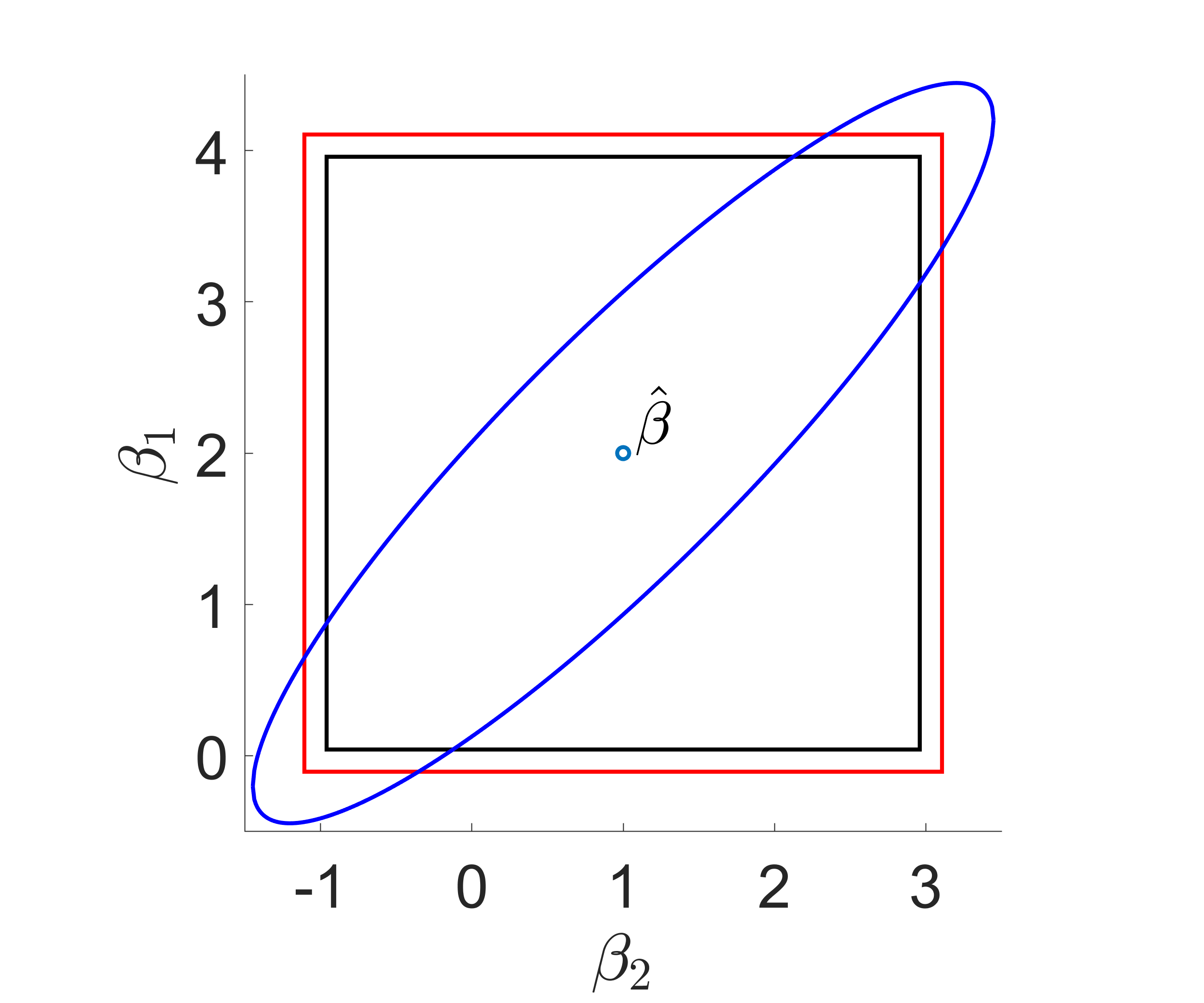}
		\caption{Positive correlation ($Cov(\beta_1, \beta_2)=0.9$).}
		\label{app-fig:diagram_corr_pos}
	\end{subfigure}\hfill
	\begin{subfigure}[t]{.49\linewidth}
		\centering
		\includegraphics[width=\linewidth]{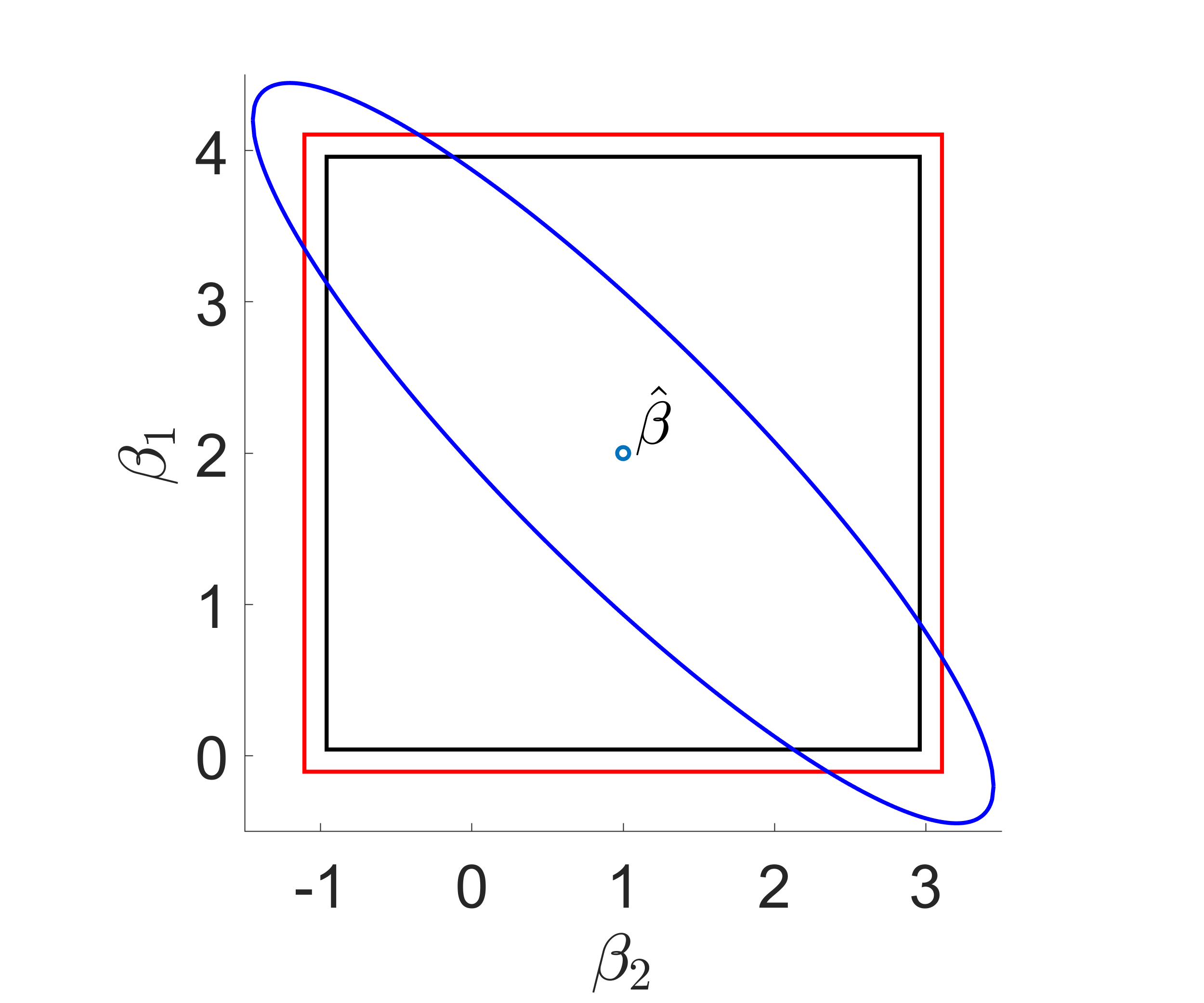}
		\caption{Negative correlation ($Cov(\beta_1, \beta_2)=-0.9$).}
		\label{app-fig:diagram_corr_neg}
	\end{subfigure}
	\caption[Illustration of different confidence regions with correlation.]{Illustration of different confidence regions with non-zero correlation. Pointwise (black), Sup-t (red), and Wald (blue) 95\% confidence region in two dimensions. $\hat{\beta}=(2,1)$. $Var(\beta_1)=Var(\beta_2)=1$.}
	\label{app-fig:diagram_corr}
\end{figure}

Online Appendix Figure \ref{app-fig:diagram_corr} illustrates the three considered confidence regions in settings with non-zero correlation between the two estimates.  
We first note that the pointwise confidence region is identical in Online Appendix Figures \ref{app-fig:diagram} and \ref{app-fig:diagram_corr}, reflecting the fact that this region does not depend on the off-diagonal entries in $Var(\hat{\beta})$. 
Second, while the sup-t region does depend on the off-diagonal entries in $Var(\hat{\beta})$, it only does so in a limited way through the critical value $c_\alpha$. In this example, the sup-t region is different between Online Appendix Figures \ref{app-fig:diagram} and \ref{app-fig:diagram_corr}, but identical in Online Appendix Figures \ref{app-fig:diagram_corr_pos} and \ref{app-fig:diagram_corr_neg}. 
In contrast, the Wald confidence region differs meaningfully in all three figures. This illustration suggests that the off-diagonal entries in $V_{\beta}$ can have important implications for the construction of confidence regions, but that traditional event study plots may be ineffective at visualizing these implications.

\newpage

\section{Algorithmic implementation}\label{app-sec:algorithm}

Recall that the algorithm takes as input jointly normal estimates of the treatment path $\hat{\beta} \sim N(\beta,V_{\beta})$, where $V_{\beta}= \sigma^2 V$, $\sigma^2=\frac{1}{H}\sum_{h=1}^H V_{\beta}(h,h)$, and $V$ is positive-definite. We define the following object:
\begin{equation*}\label{app-eq:obj}
	\begin{aligned}
		\tilde\beta(\lambda_1,\lambda_2,K) &= \arg\min_{b} \ Q(b,\lambda_1,\lambda_2,K) \\
		&= \arg\min_{b} \ \underbrace{(\hat{\beta} - b)'V^{-1}(\hat{\beta} - b)}_{\text{distance from }\hat{\beta}} + \lambda_1 \underbrace{b'D_1'W_1(K)D_1 b}_{\substack{\text{penalty on first difference} \\ \text{after horizon K}}} +
		\lambda_2 \underbrace{b'D_3'W_3 D_3 b,}_{\substack{\text{penalty on} \\ \text{third difference}}}
	\end{aligned}
\end{equation*}
where 
\begin{itemize}
	\item $\lambda_1$, $\lambda_2$, $K$ are tuning parameters that determine the surrogate $M$
	\item $D_1$ and $D_3$ are the $H \times (H-1)$ and $H \times (H-3)$ first and third difference operators 
	\item $V_1 = D_1 V D_1'$, $V_3 = D_3 V D_3'$ are (scaled) variance matrices for first and third differences
	\item $V_1(K)$ is the $(H-K) \times (H-K)$ matrix consisting of the lower right entries of $V_1$, $V_1(K:H-1, K:H-1)$ 
	\item $\bar{V}_3=\frac{1}{H-3}\sum_{h=1}^{H-3} V_3(h,h)$, $\bar{V}_1(K)=\frac{1}{H-K}\sum_{h=K}^{H-1} V_1(h,h)$
	\item 
	$\begin{aligned}[t]
		W_1(K) = \begin{pmatrix}
			\begin{array}{ll}
				\mathbf{0}_{(K-1) \times (K-1)} & \mathbf{0}_{(K-1) \times (H-K)} \\
				\mathbf{0}_{(H-K) \times (K-1)} & diag(V_1(K))/\bar{V}_1(K)
			\end{array}
		\end{pmatrix}
	\end{aligned}$
	\item $W_3 = diag(V_3)/\bar{V}_3$
\end{itemize}
Intuitively, $W_1(K)$ and $W_3$ are analogs to natural scaling in standard ridge with independent columns but different variances.

To select the surrogate $M$ from the data we choose $M = (\lambda_1,\lambda_2,K)$ that minimizes a BIC analog over $\mathcal{M}$: $\hat{M} = \argmin_{M \in \mathcal{M}} (\hat{\beta} - \tilde\beta(M))'V_{\beta}^{-1}(\hat{\beta} - \tilde\beta(M)) + \log(H) \text{df}(\lambda_1,\lambda_2,K)$. The universe of models considered, $\mathcal{M}$, includes 
\begin{enumerate}[label=(\alph*)]
	\item A constant, linear, quadratic, and cubic treatment effect model (with one, two, three, and four degrees of freedom respectively)
	\item Surrogates of the form $P(M)\beta=\beta(\lambda_1,\lambda_2,K)$ using a grid over $(\lambda_1,\lambda_2,K)$ \label{item:grid}
	\item the unrestricted estimates $\hat{\beta}$ ($df=H$) 
\end{enumerate}
We construct the grid for the surrogates under \ref{item:grid} as follows. First, we set lower and upper bounds for $\lambda_1$ and $\lambda_2$. Independent of $K$, these bounds are equal to $(\underline{\lambda}_1, \bar{\lambda}_1)=(e^{-10}, e^{10})$, and $(\underline{\lambda}_2, \bar{\lambda}_2)=(e^{-10},\bar{\lambda}_2)$, where $\bar{\lambda}_2$ is defined as the $\lambda_2$ such that $df(e^{-10}, \lambda_2, K) = 4$.\footnote{Recall that $\textrm{df}(\lambda_1,\lambda_2,K) = \textrm{trace}\left(\left(V^{-1} + \lambda_1 D_1'W_1(K)D_1 + \lambda_2 D_3'W_3 D_3\right)^{-1} V^{-1}\right)$.} Note that, with $\lambda_1=0$, $\bar{\lambda}_2$ also does not depend on $K$. We then consider the Cartesian product of an equal spaced grid of 20 points between $(\log(\underline{\lambda}_1), \log(\bar{\lambda}_1))$ and equal spaced grid of 20 points between $(\log(\underline{\lambda}_2), \log(\bar{\lambda}_2))$, and retain those grid points with $df\in [4,H-1]$.

\newpage

\section{Simulation design}\label{app-sec:add_details}

\begin{table}[hb!]
	\centering
	\begin{tabular}{lc} 
		\hline \hline \\[-0.9em] 
		Scenario & treatment path $\beta$ \\ \hline \hline \\[-0.9em]
		Constant treatment effect  & $\beta_h = -0.4 \; \forall \ h$\\[18pt]
		Smooth, eventually flat   & $ \beta_h = \left\{ \begin{array}{cc} -0.289 + \frac{(18-h)^2}{1000} & h \le 17 \\ -0.289 & h \ge 18 \end{array} \right. $ \\[15pt]
		Hump-shaped   & $\beta_h = -0.4 - 0.4 \sin\bigg(\frac{3}{70}  \pi (h-1)\bigg) \; \forall \ h$\\[10pt]  
		\multirow{3}{*}{Wiggly}   & $\beta_h= \breve \beta_h+ \xi_h$, where $\xi_h \sim N(0,0.1)$ iid across $h$ and \\
		& $\breve \beta_h = \left\{ \begin{array}{cc} -0.4 \sin\bigg(\frac{1}{35}\pi(h-1)\bigg) & h \le 19 \\ -0.4 & h \ge 20 \end{array} \right.$ 
	\end{tabular}
	\caption[Description of different treatment paths]{Detailed description of the four different treatment paths $\beta= \{\beta\}_{h=1}^{36}$ considered in the simulations. We draw a single realization of the ``Wiggly" scenario (which is depicted in Figure \ref{fig:DGPs_wiggly}) to use throughout our simulations.}
	\label{app-tab:DGPs}
\end{table}

We generate the covariance matrix of $\hat{\beta}$ as $V_{\beta} = \sigma^2*diag(S) R \ diag(S)$, where $S_h=(100+h)/100$, and $R$ is a $H \times H$ Toeplitz matrix with $R_{ij}=\rho^{|i-j|}$. 
For all results in the main text, we set $\rho=0$ (such that $R$ becomes the identity matrix). 

\newpage

\section{Additional simulation results}\label{app-sec:add_results}
\begin{figure}[tbhp!]
	\begin{subfigure}[t]{.48\textwidth}
		\centering
		\includegraphics[width=\linewidth]{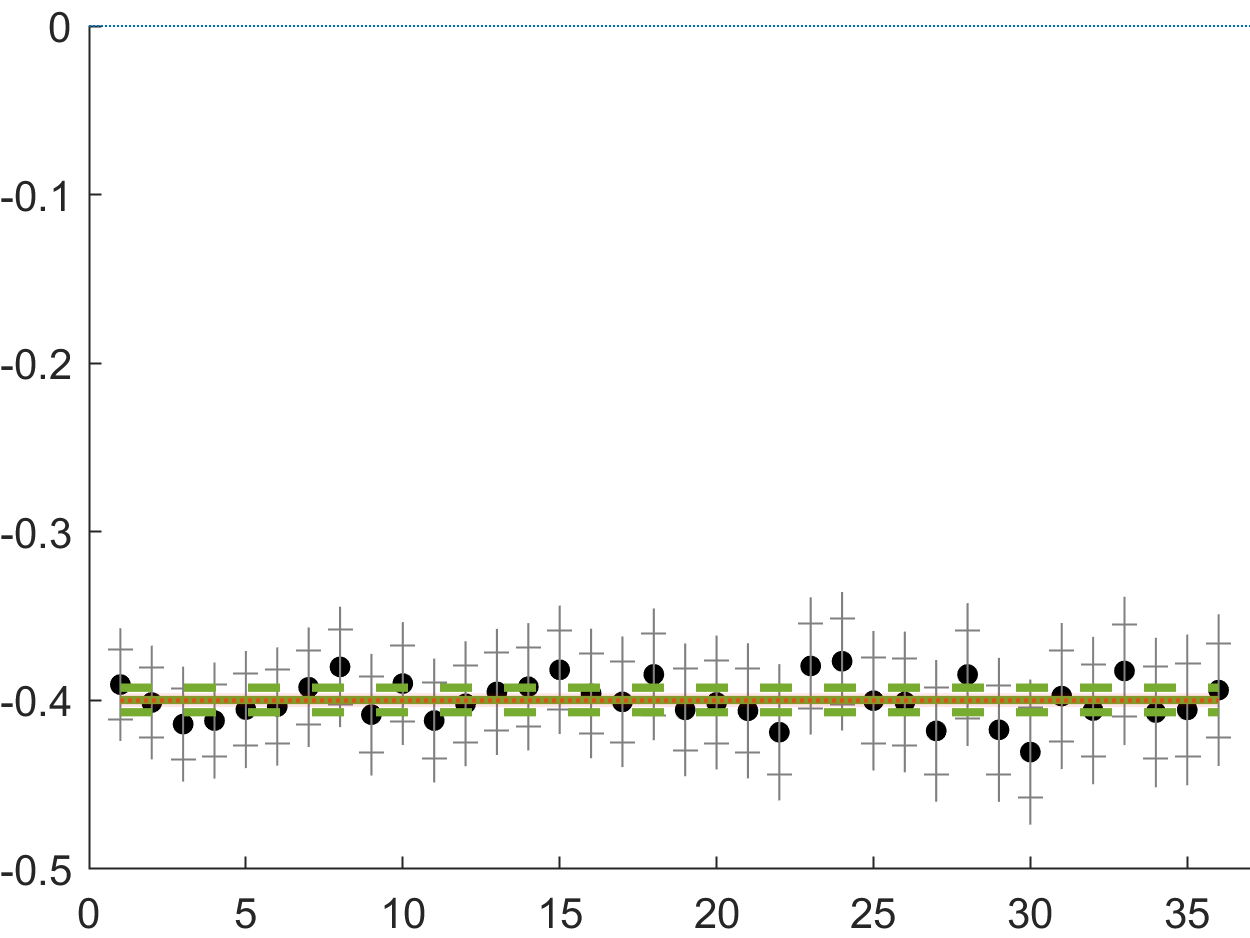}
		\caption{treatment path constant}
		\label{app-fig:new_viz_constant}
	\end{subfigure}\hfill
	\begin{subfigure}[t]{.48\linewidth}
		\centering
		\includegraphics[width=\linewidth]{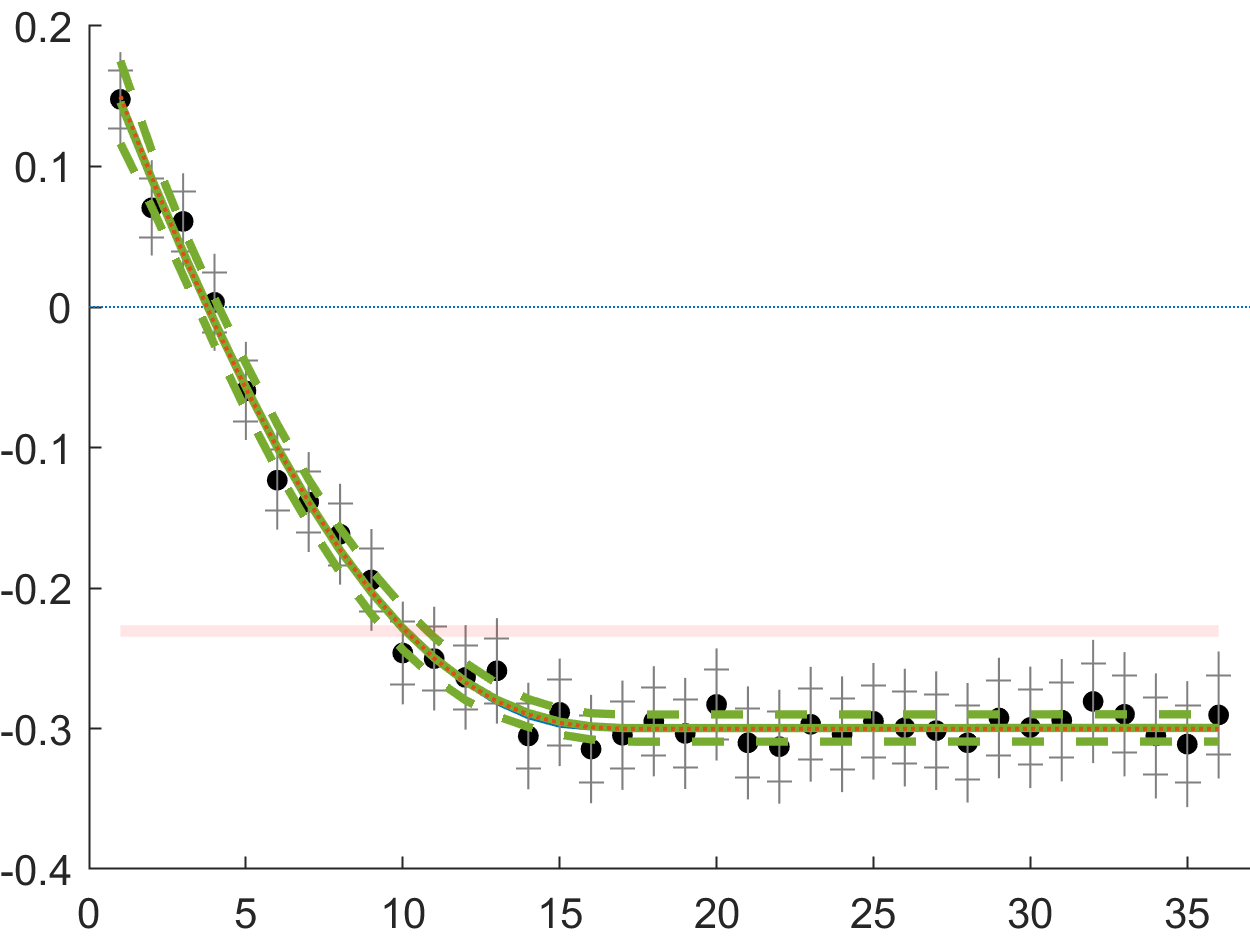}
		\caption{treatment path smooth, eventually flat}
		\label{app-fig:new_viz_quadratic}
	\end{subfigure}
	
	\begin{subfigure}[t]{.48\textwidth}
		\centering
		\includegraphics[width=\linewidth]{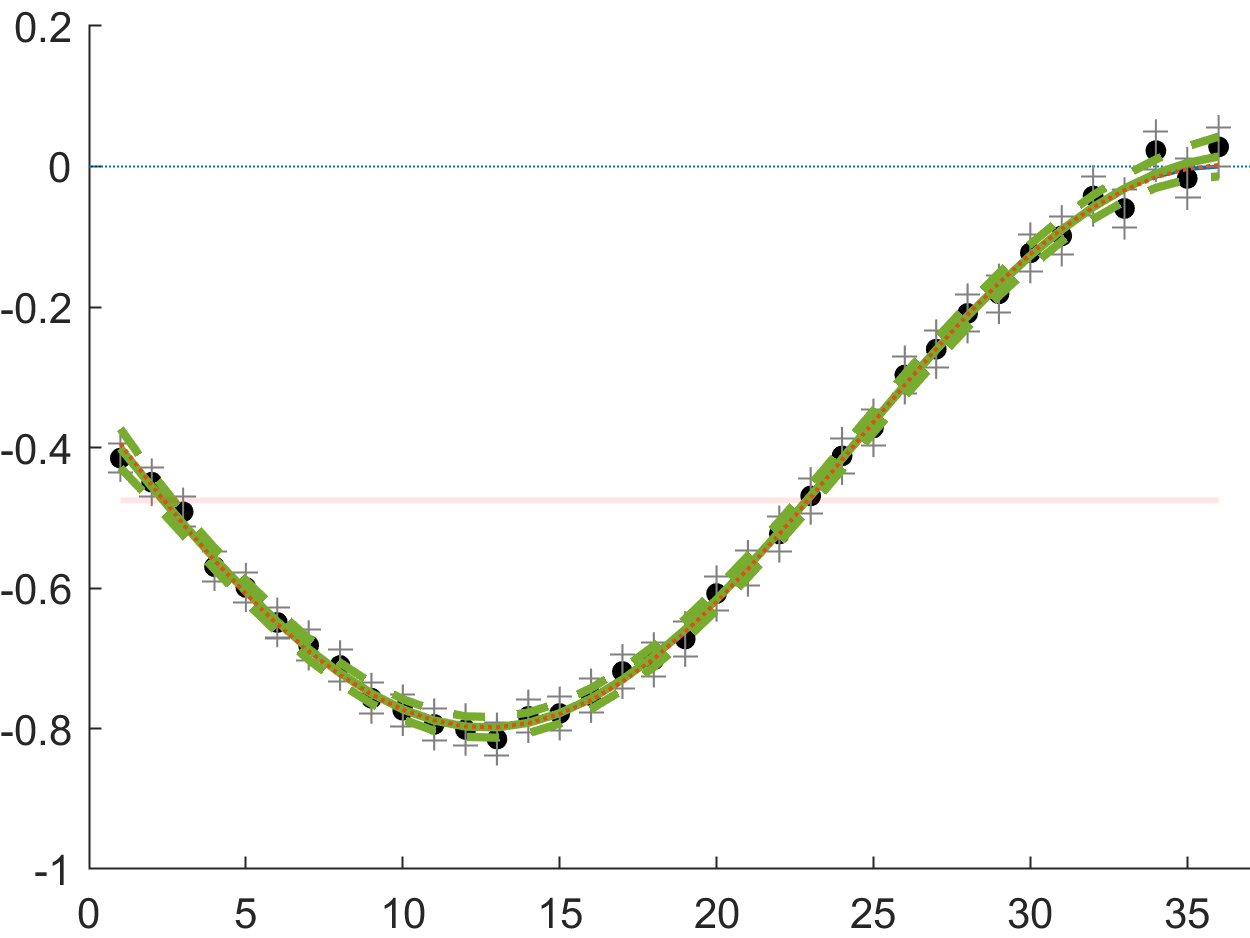}
		\caption{treatment path hump-shaped}
		\label{app-fig:new_viz_no_flat}
	\end{subfigure}\hfill
	\begin{subfigure}[t]{.48\linewidth}
		\centering
		\includegraphics[width=\linewidth]{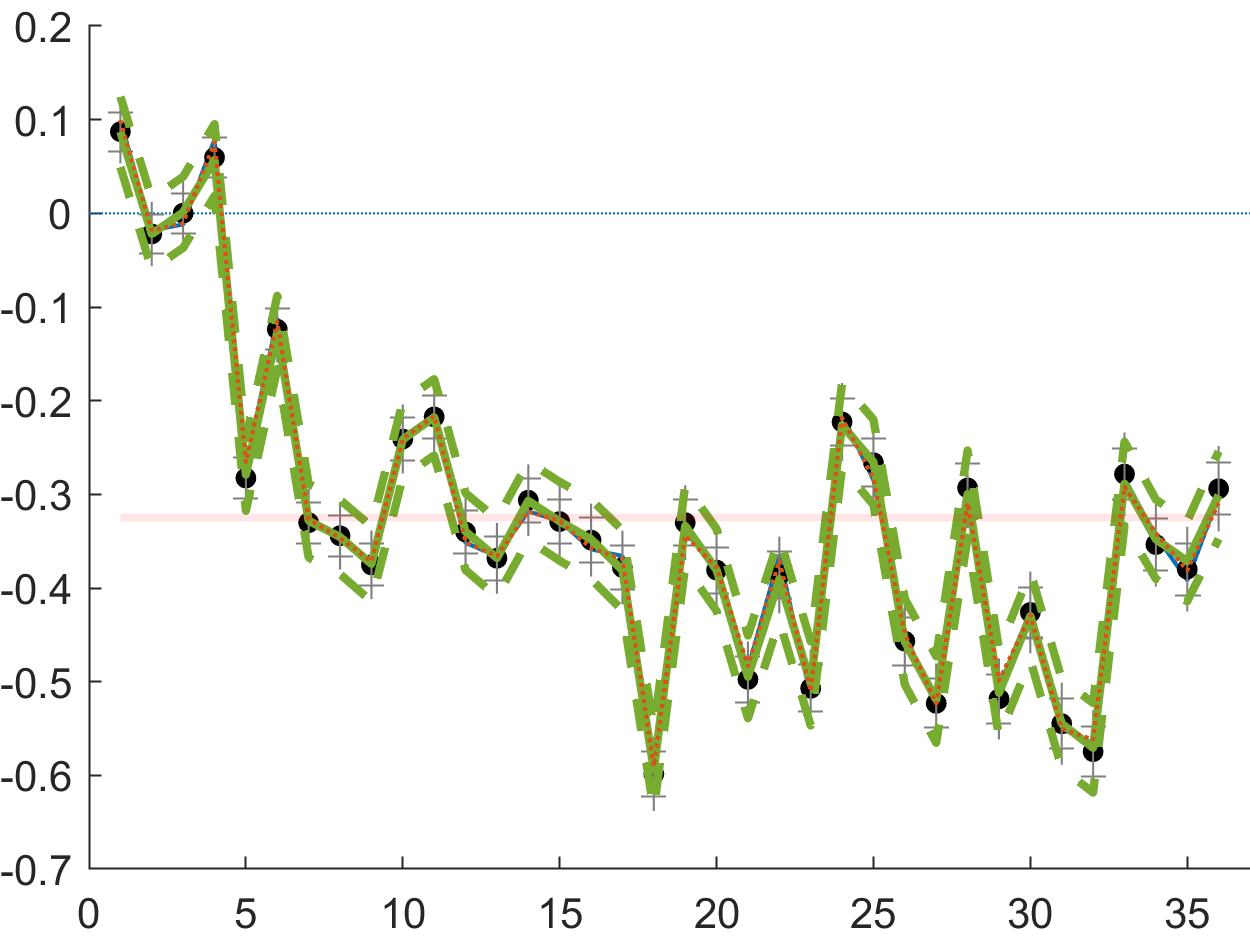}
		\caption{treatment path wiggly}
		\label{app-fig:new_viz_wiggly}
	\end{subfigure}
	\caption[Exemplary event-study plots including our proposals with smaller noise.]{Exemplary event-study plots including our proposals with smaller noise than in Figure \ref{fig:new_viz}.}
	\label{app-fig:new_viz_largen}
\end{figure}

\begin{figure}[tb!]
	\centering
	\begin{subfigure}[h]{0.49\textwidth}
		\includegraphics[width=\linewidth]{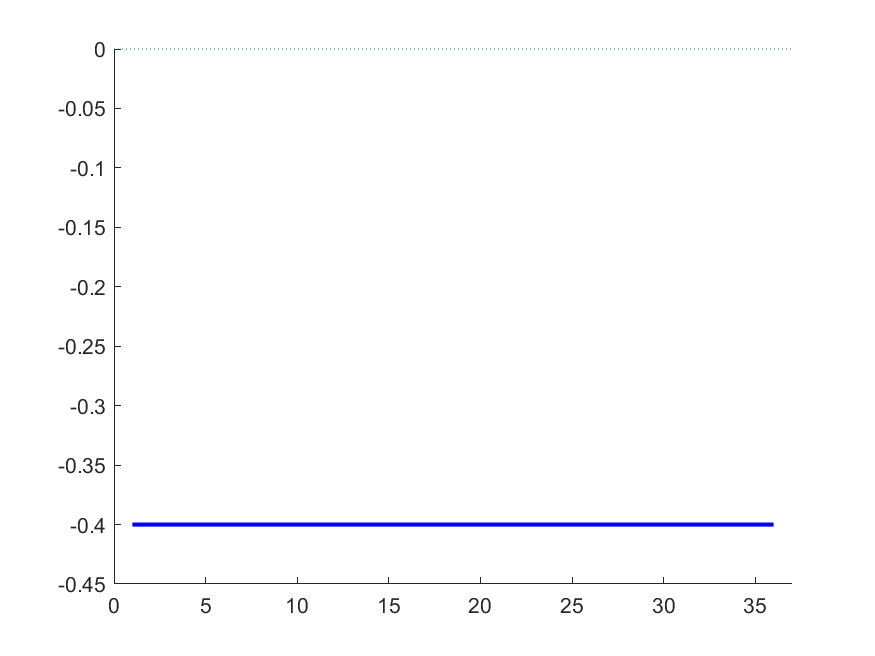}
		\caption{Constant Treatment Effect}
		\label{app-fig:model_universe_constant}
	\end{subfigure}
	\hfill
	\begin{subfigure}[h]{0.49\textwidth}
		\includegraphics[width=\linewidth]{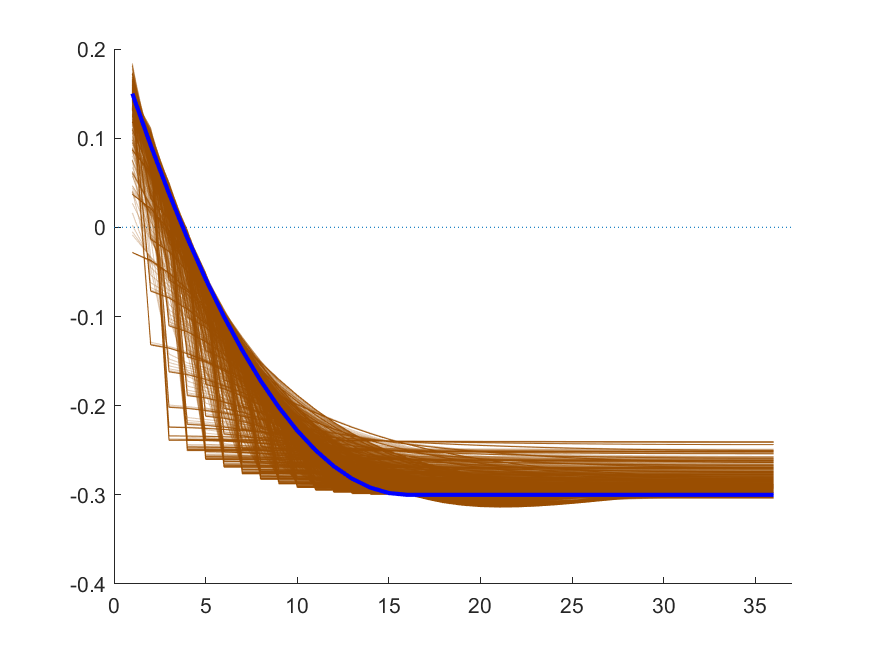}
		\caption{Smooth, eventually flat}
		\label{app-fig:model_universe_quadratic}
	\end{subfigure}
	
	\begin{subfigure}[h]{0.49\textwidth}
		\includegraphics[width=\linewidth]{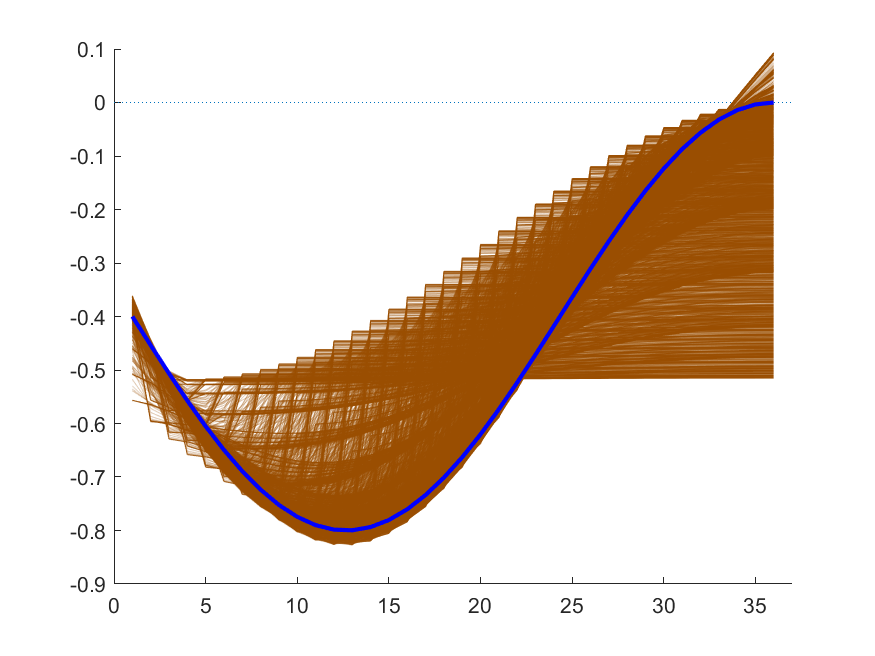}
		\caption{Hump-shaped}
		\label{app-fig:model_universe_noflat}
	\end{subfigure}
	\hfill
	\begin{subfigure}[h]{0.49\textwidth}
		\includegraphics[width=\linewidth]{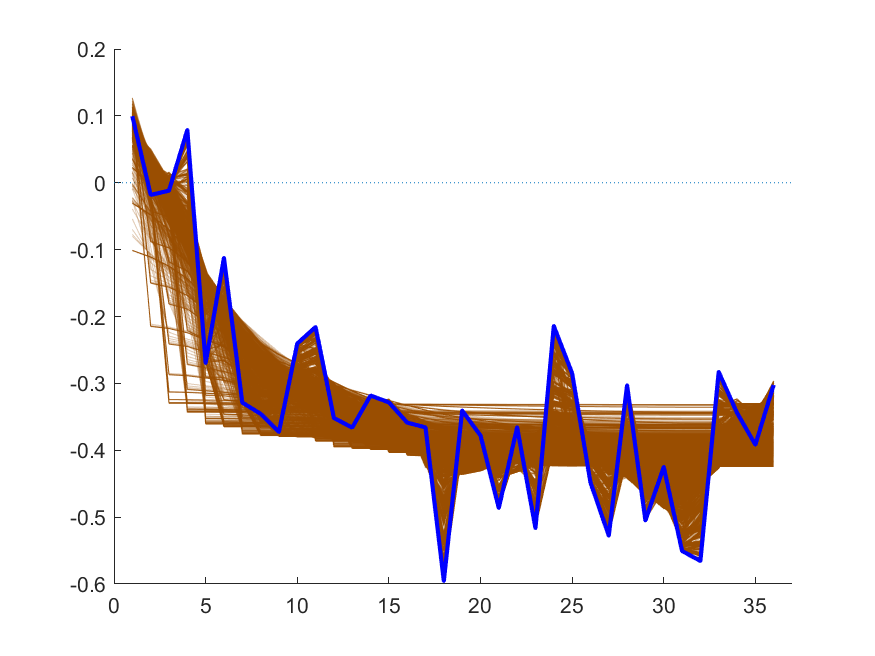}
		\caption{Wiggly}
		\label{app-fig:model_universe_wiggly}
	\end{subfigure}
	\caption[Illustration of Model universe $\mathcal{M}$]{Illustration of Model universe $\mathcal{M}$. Brown lines correspond to all considered models $M$ of the form $P(M)\beta=\beta(\lambda_1,\lambda_2,K)$ with $df \in [4, H-1]$. Blue line corresponds to true treatment effect $\beta$.}
	\label{app-fig:model_universe}
\end{figure}

\begin{figure}[tb!]
	\centering
	\begin{subfigure}[h]{0.49\textwidth}
		\includegraphics[width=\linewidth]{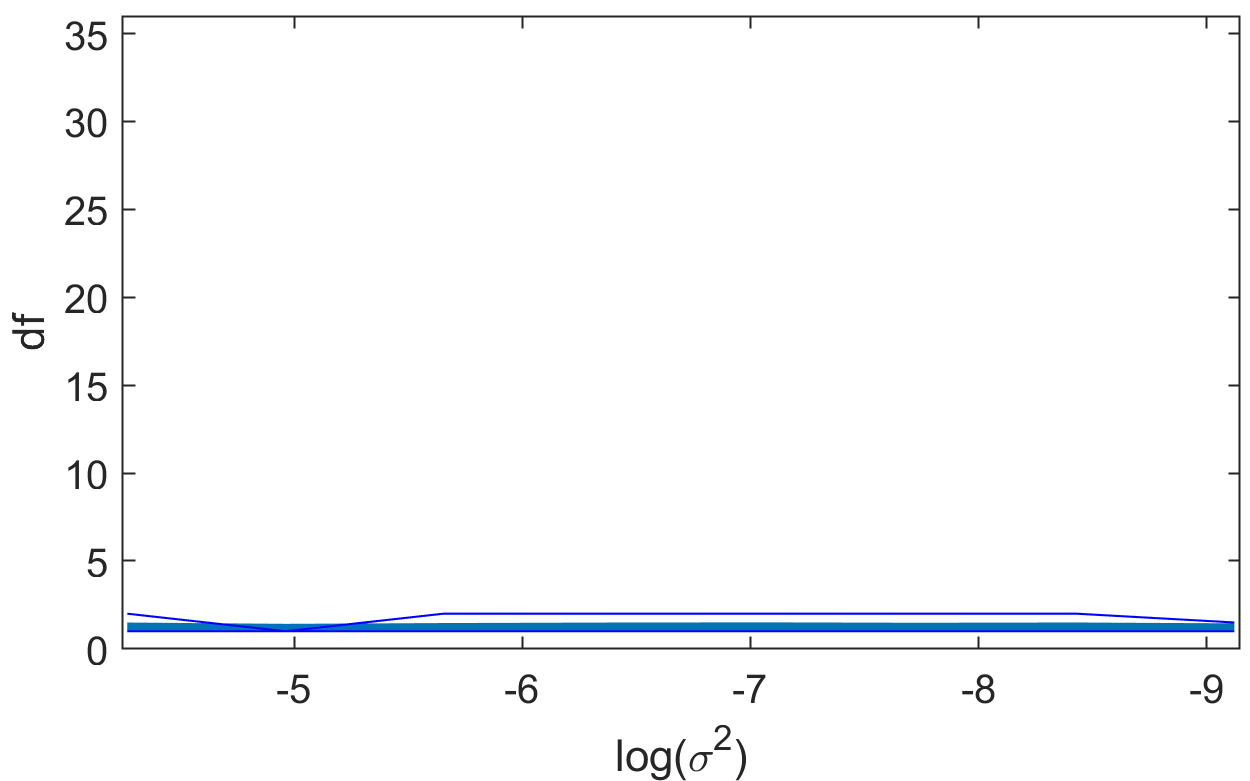}
		\caption{Constant Treatment Effect}
		\label{app-fig:df_constant}
	\end{subfigure}
	\hfill
	\begin{subfigure}[h]{0.49\textwidth}
		\includegraphics[width=\linewidth]{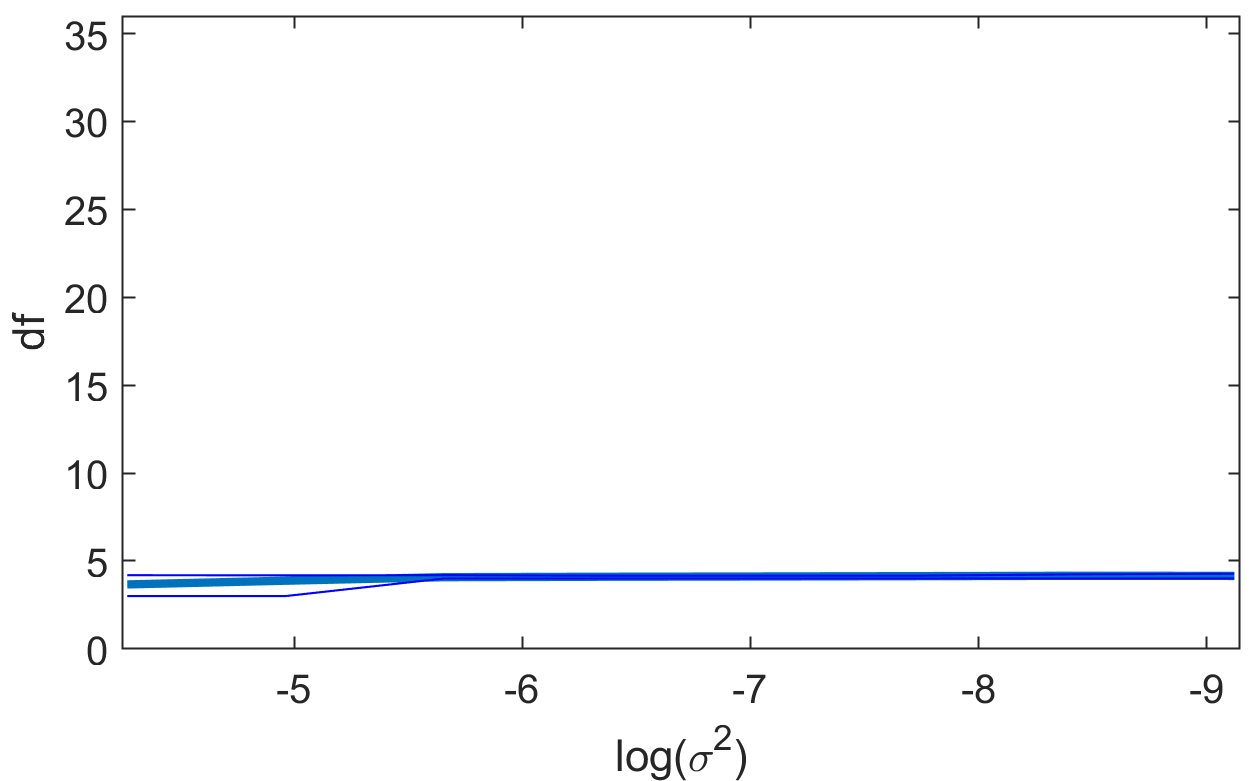}
		\caption{Smooth, eventually flat}
		\label{app-fig:df_quadratic}
	\end{subfigure}
	
	\begin{subfigure}[h]{0.49\textwidth}
		\includegraphics[width=\linewidth]{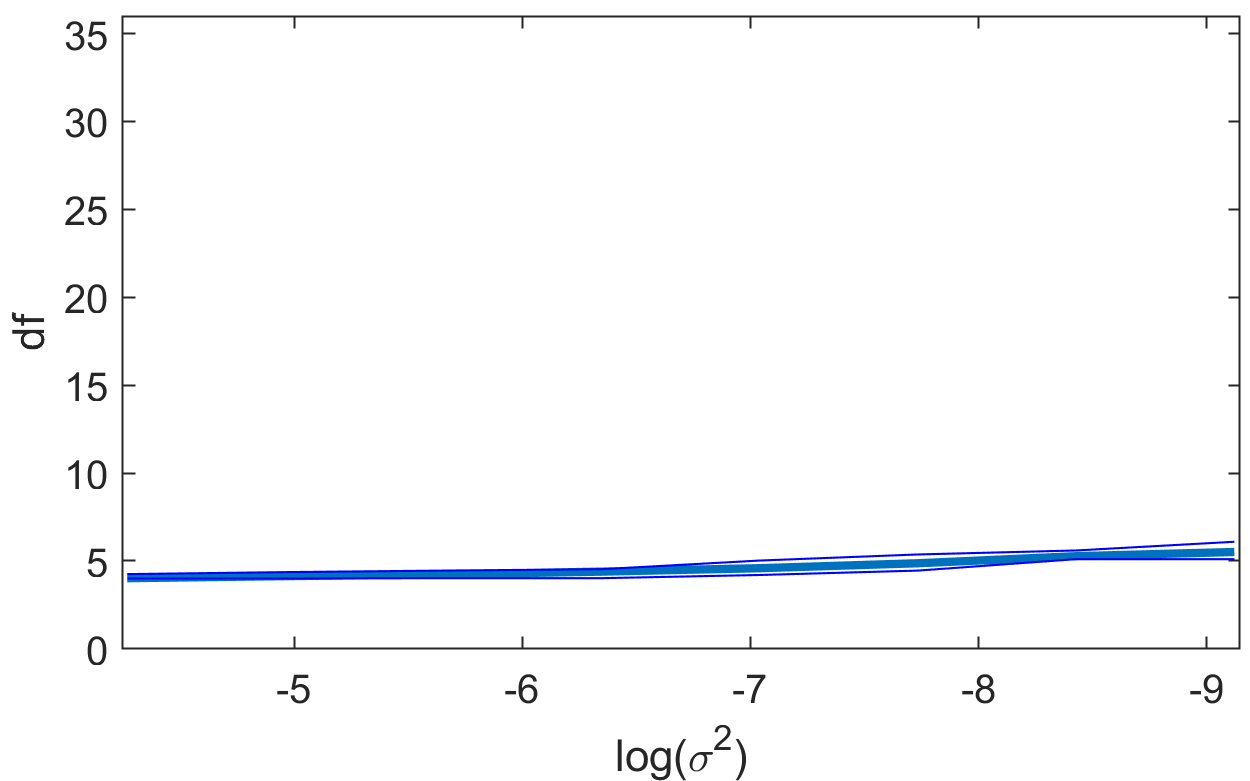}
		\caption{Hump-shaped}
		\label{app-fig:df_noflat}
	\end{subfigure}
	\hfill
	\begin{subfigure}[h]{0.49\textwidth}
		\includegraphics[width=\linewidth]{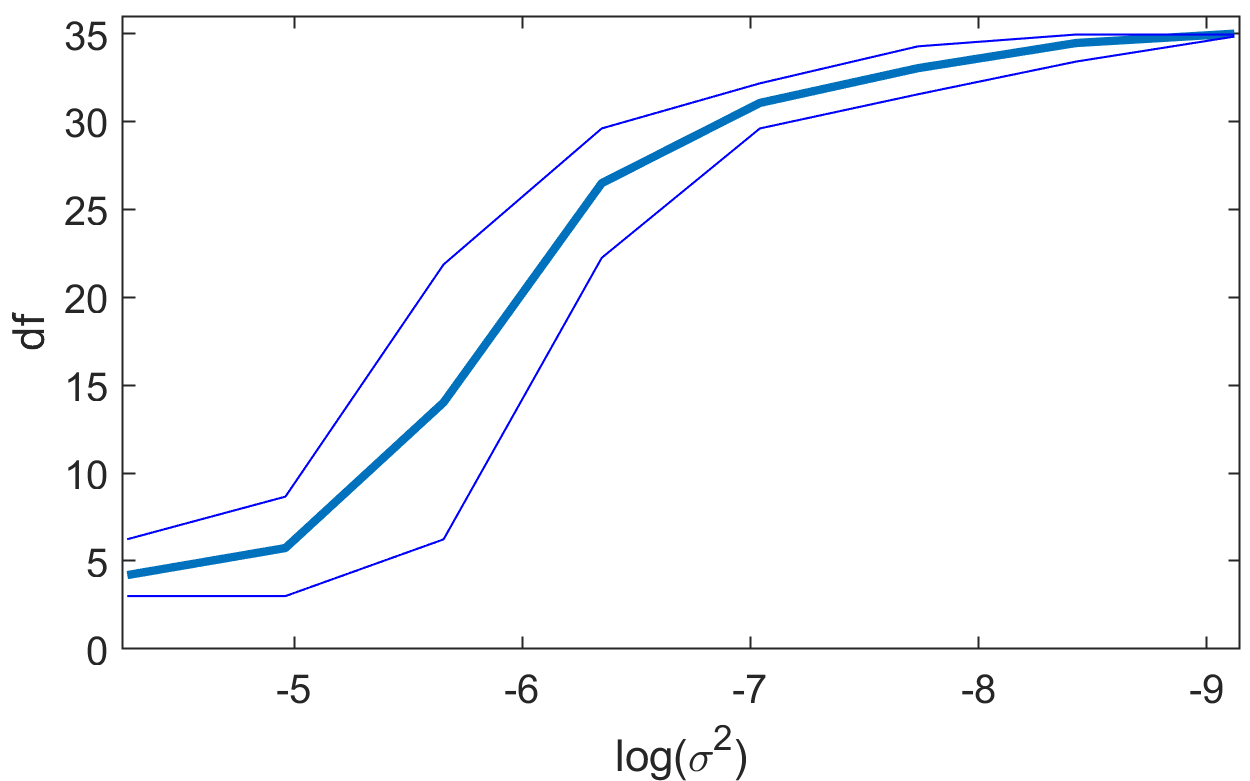}
		\caption{Wiggly}
		\label{app-fig:df_wiggly}
	\end{subfigure}
	\caption[Chosen $df$ for restricted estimates]{Chosen $df$ across realizations for restricted estimates as a function of the amount of noise $\sigma^2$ in the initial estimates $\hat{\beta}$.} \label{app-fig:df}
\end{figure}

\begin{figure}[tb!]
	\centering
	\begin{subfigure}[h]{0.4\textwidth}
		\includegraphics[width=\linewidth]{simulation/wiggly/no_cor/surrogates_iid.png}
		\caption{$log(\sigma^2)=-4.27$}
	\end{subfigure}
	\hfill
	\begin{subfigure}[h]{0.4\textwidth}
		\includegraphics[width=\linewidth]{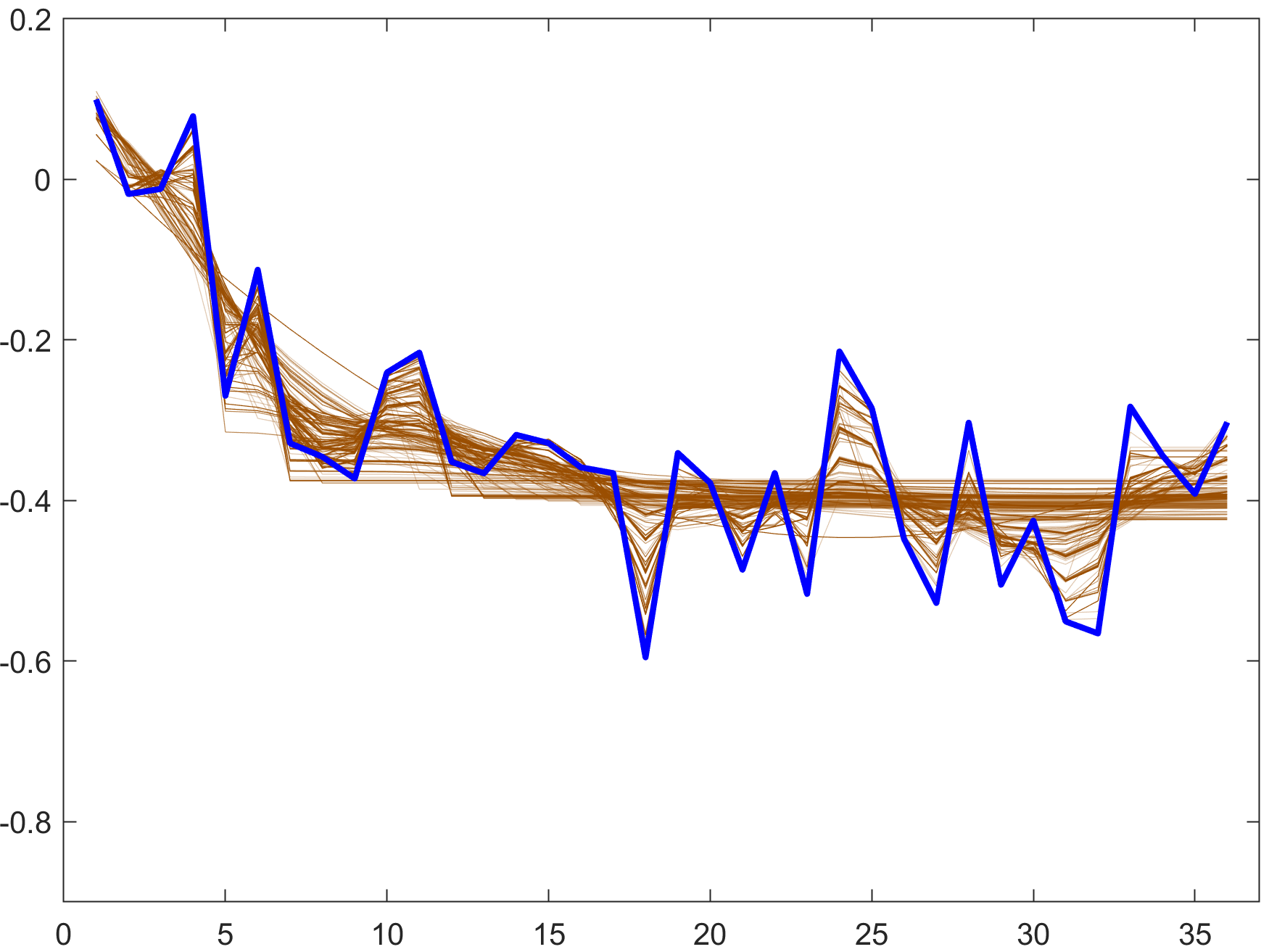}
		\caption{$log(\sigma^2)=-5.66$}
	\end{subfigure}
	
	\begin{subfigure}[h]{0.4\textwidth}
		\includegraphics[width=\linewidth]{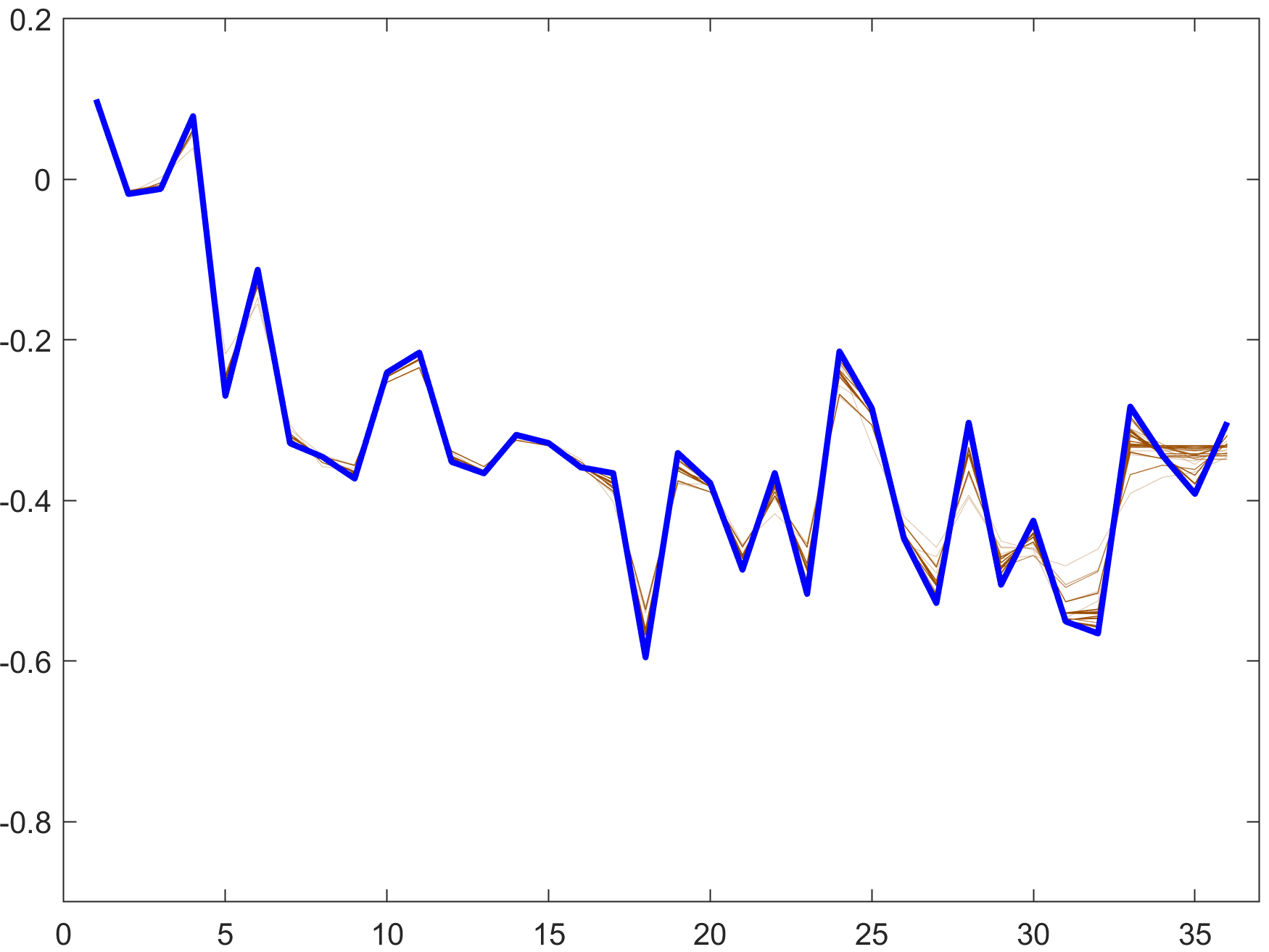}
		\caption{$log(\sigma^2)=-7.04$}
	\end{subfigure}
	\hfill
	\begin{subfigure}[h]{0.4\textwidth}
		\includegraphics[width=\linewidth]{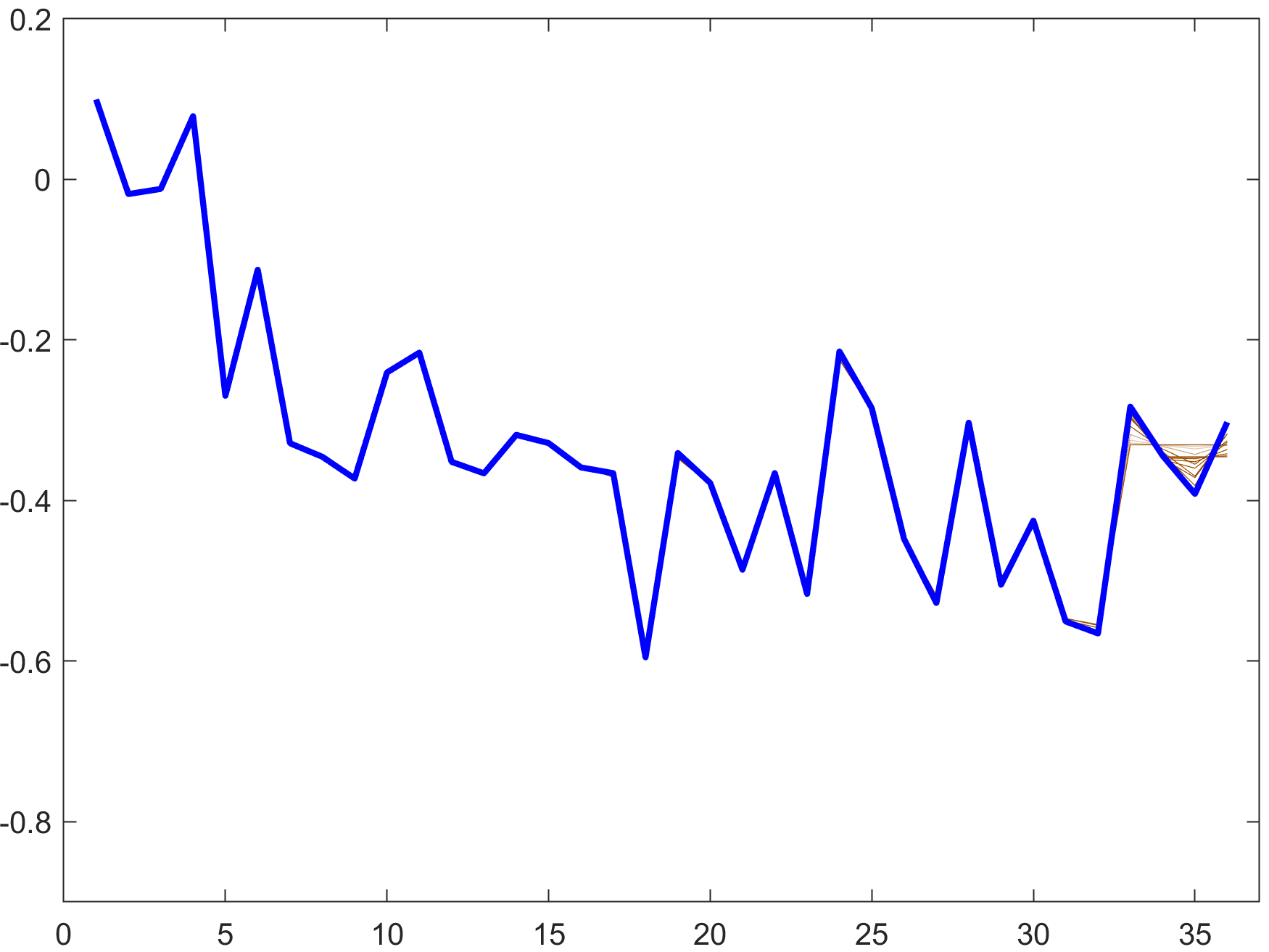}
		\caption{$log(\sigma^2)=-8.43$}
	\end{subfigure}
	\caption[Illustration of chosen surrogates for Wiggly DGP.]{Illustration of 1,000 chosen surrogates for Wiggly DGP for various levels of $\sigma^2$, the amount of noise in the initial estimates $\hat{\beta}$.}
	\label{app-fig:surrogates_wiggly_noise}
\end{figure}

\clearpage

\newpage

\section{Simulation results with positively correlated estimates}\label{app-sec:general_cov}

In this appendix, we repeat the simulation experiment reported in the main test using a covariance matrix $V_{\beta}$ capturing positively correlated estimates. In particular, recall that $V_{\beta} = \sigma^2*diag(S) R \ diag(S)$, where $S_h=(100+h)/100$, and $R$ is a $H \times H$ Toeplitz matrix with $R_{ij}=\rho^{|i-j|}$. In the results that follow, we set $\rho=0.8$.

\begin{figure}[hb!]
	\centering
	\begin{subfigure}[h]{0.49\textwidth}
		\includegraphics[width=\linewidth]{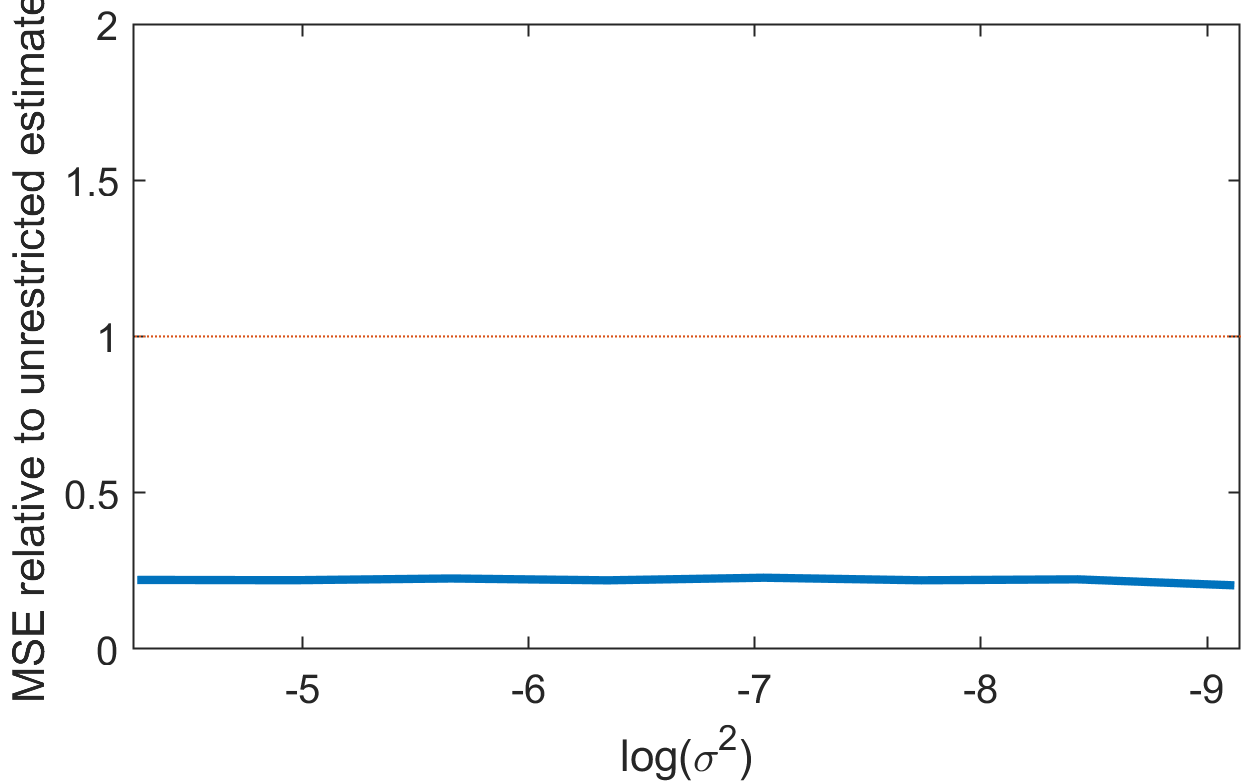}
		\caption{Constant Treatment Effect}
		\label{app-fig:mse_constant}
	\end{subfigure}
	\hfill
	\begin{subfigure}[h]{0.49\textwidth}
		\includegraphics[width=\linewidth]{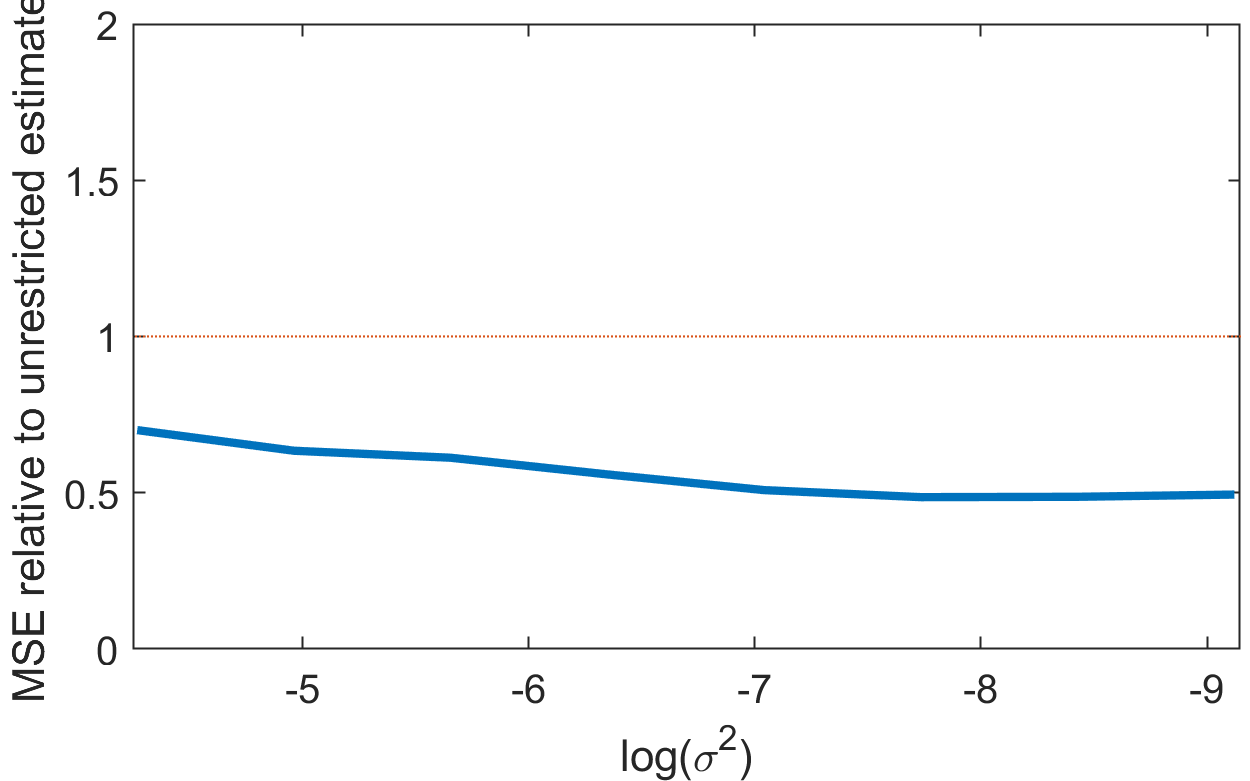}
		\caption{Smooth, eventually flat}
		\label{app-fig:mse_quadratic}
	\end{subfigure}
	\begin{subfigure}[h]{0.49\textwidth}
		\includegraphics[width=\linewidth]{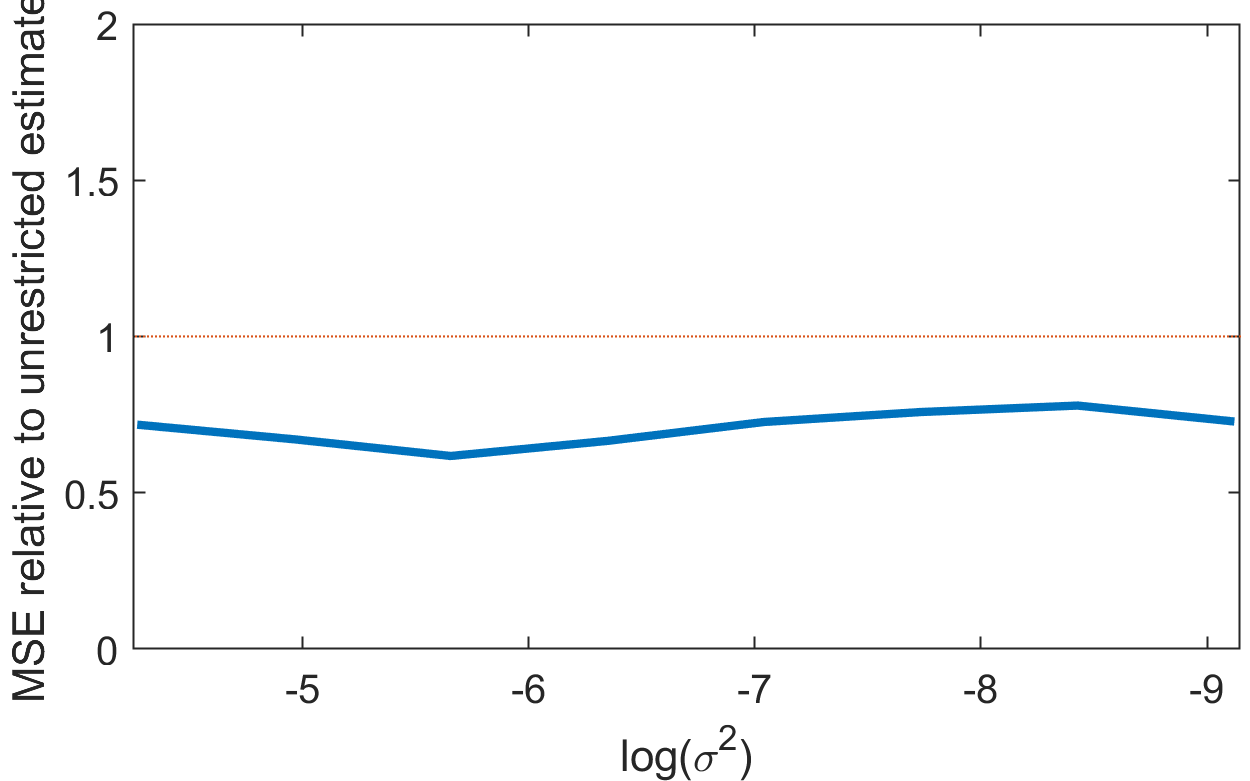}
		\caption{Hump-shaped}
		\label{app-fig:mse_noflat}
	\end{subfigure}
	\hfill
	\begin{subfigure}[h]{0.49\textwidth}
		\includegraphics[width=\linewidth]{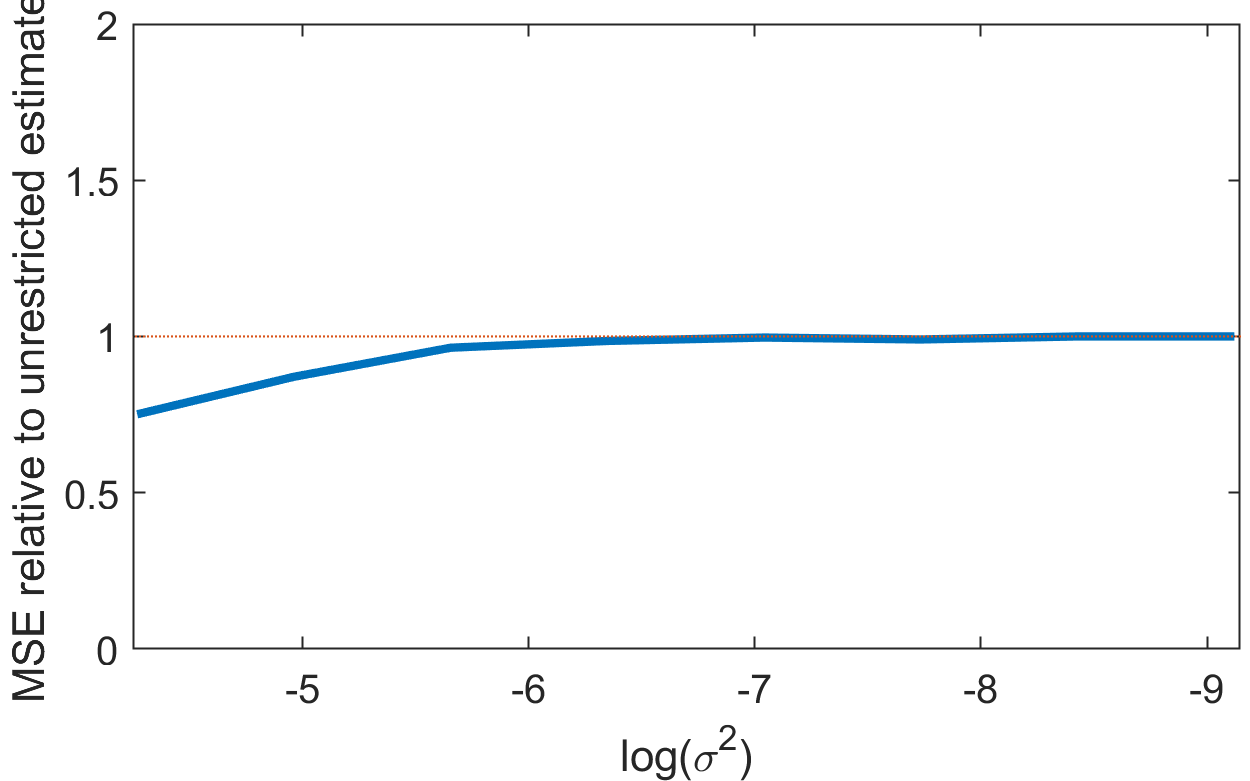}
		\caption{Wiggly}
		\label{app-fig:mse_wiggly}
	\end{subfigure}
	\caption[Relative performance of restricted and unrestricted estimators under correlation.]{Relative performance of restricted and unrestricted estimators. Depicted is the ratio $\frac{MSE(\tilde{\beta}(\hat{M}))}{MSE(\hat{\beta})}$ as a function of the amount of noise in the initial estimates $\hat{\beta}$.}
	\label{app-fig:MSEs}
\end{figure}

\begin{figure}[tb!]
	\centering
	\begin{subfigure}[h]{0.49\textwidth}
		\includegraphics[width=\linewidth]{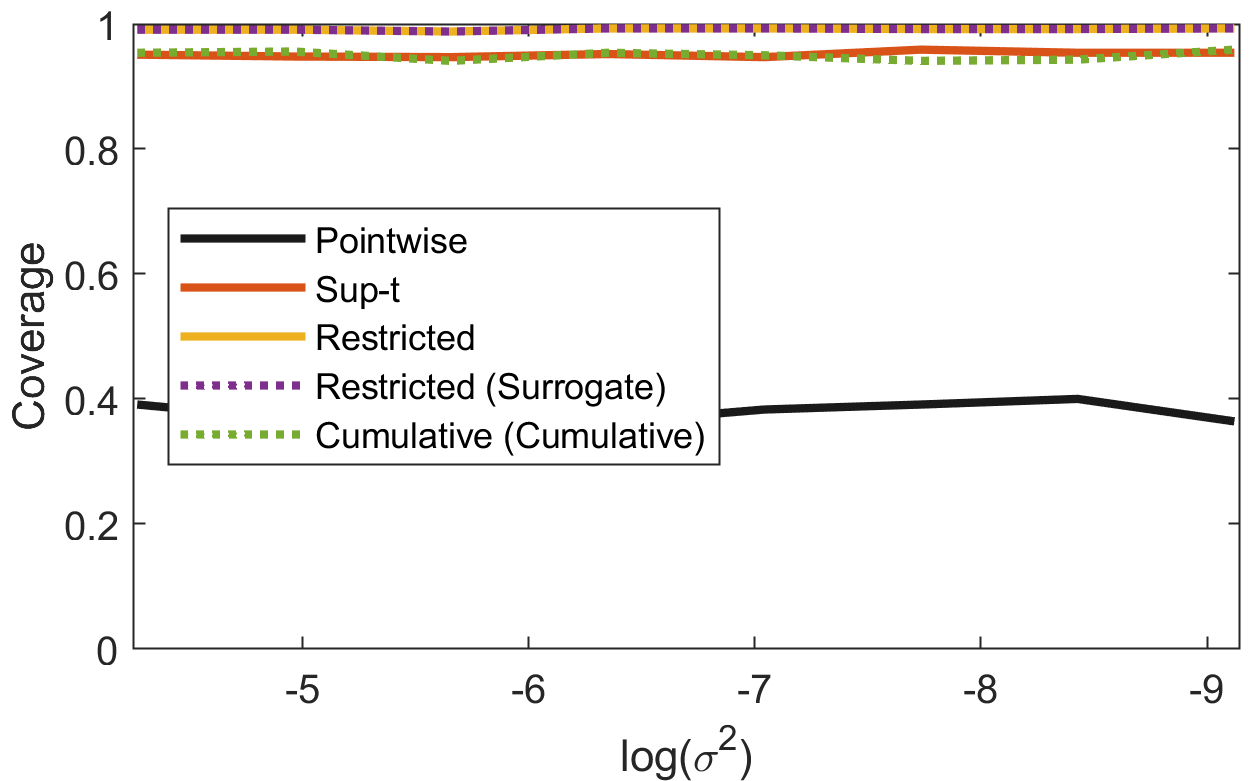}
		\caption{Constant Treatment Effect}
		\label{app-fig:coverage_constant}
	\end{subfigure}
	\hfill
	\begin{subfigure}[h]{0.49\textwidth}
		\includegraphics[width=\linewidth]{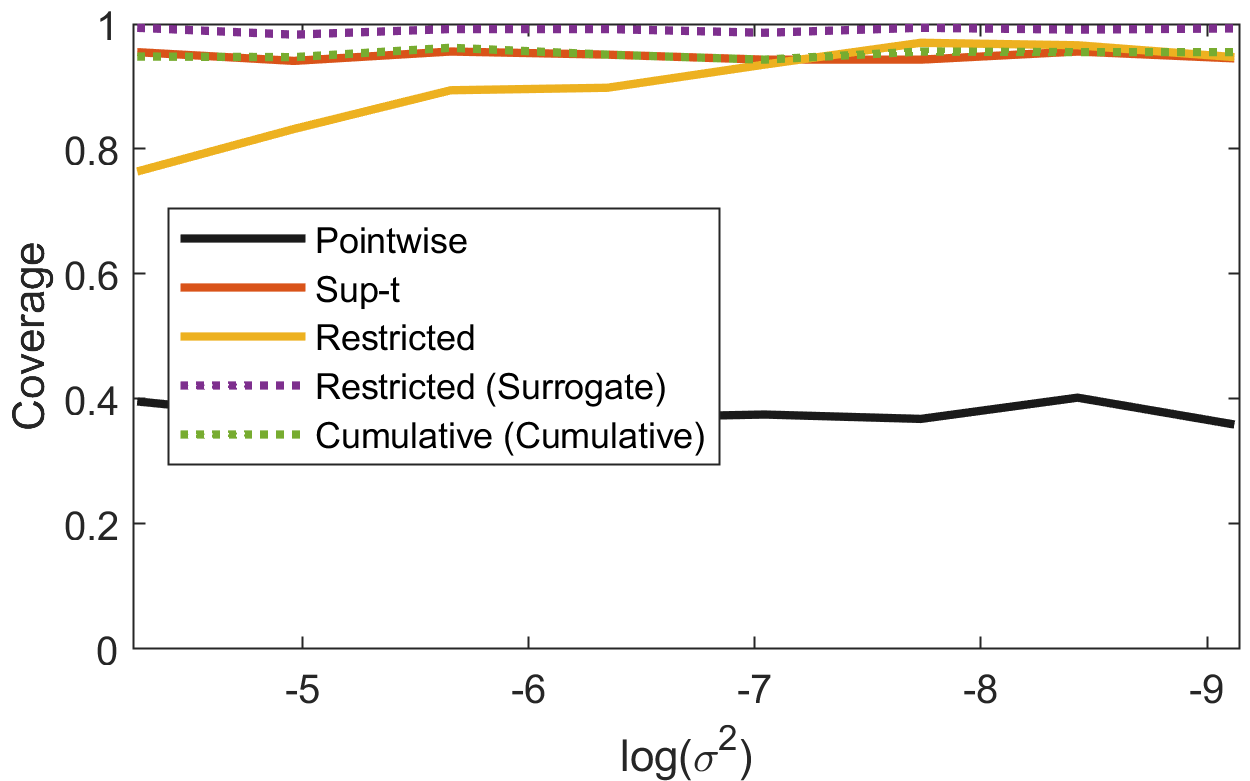}
		\caption{Smooth, eventually flat}
		\label{app-fig:coverage_quadratic}
	\end{subfigure}
	\begin{subfigure}[h]{0.49\textwidth}
		\includegraphics[width=\linewidth]{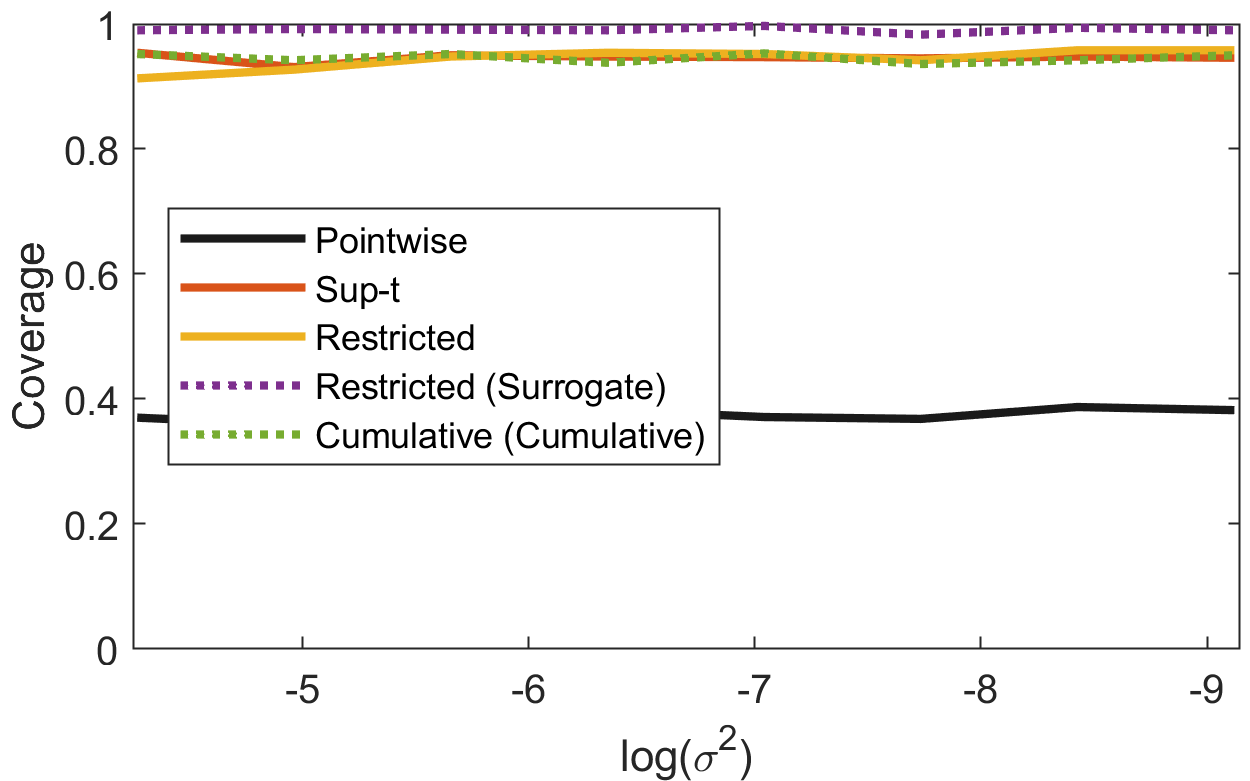}
		\caption{Hump-shaped}
		\label{app-fig:coverage_noflat}
	\end{subfigure}
	\hfill
	\begin{subfigure}[h]{0.49\textwidth}
		\includegraphics[width=\linewidth]{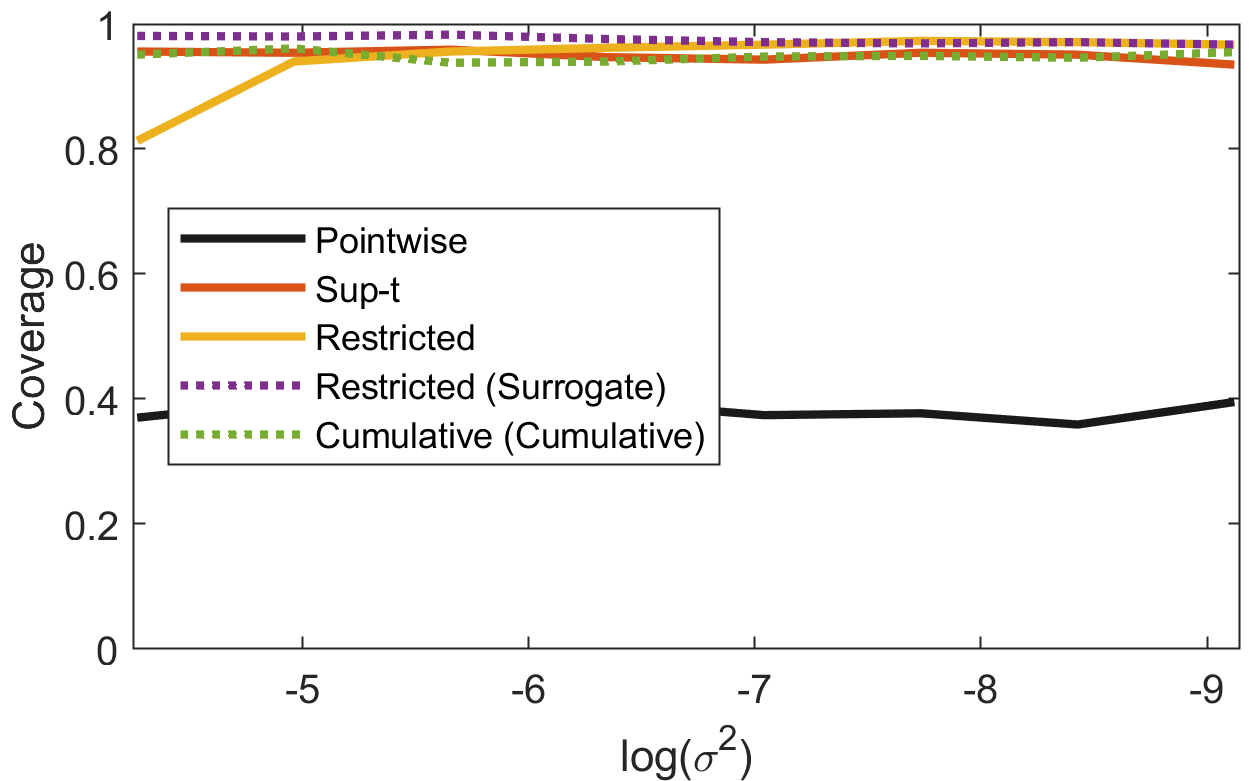}
		\caption{Wiggly}
		\label{app-fig:coverage_wiggly}
	\end{subfigure}
	\caption[Coverage properties of various confidence regions under correlation]{Coverage properties of various confidence regions as a function of the amount of noise in the initial estimates $\hat{\beta}$.}
	\label{app-fig:coverage}
\end{figure}

\begin{figure}[thb!]
	\centering
	\begin{subfigure}[h]{0.49\textwidth}
		\includegraphics[width=\linewidth]{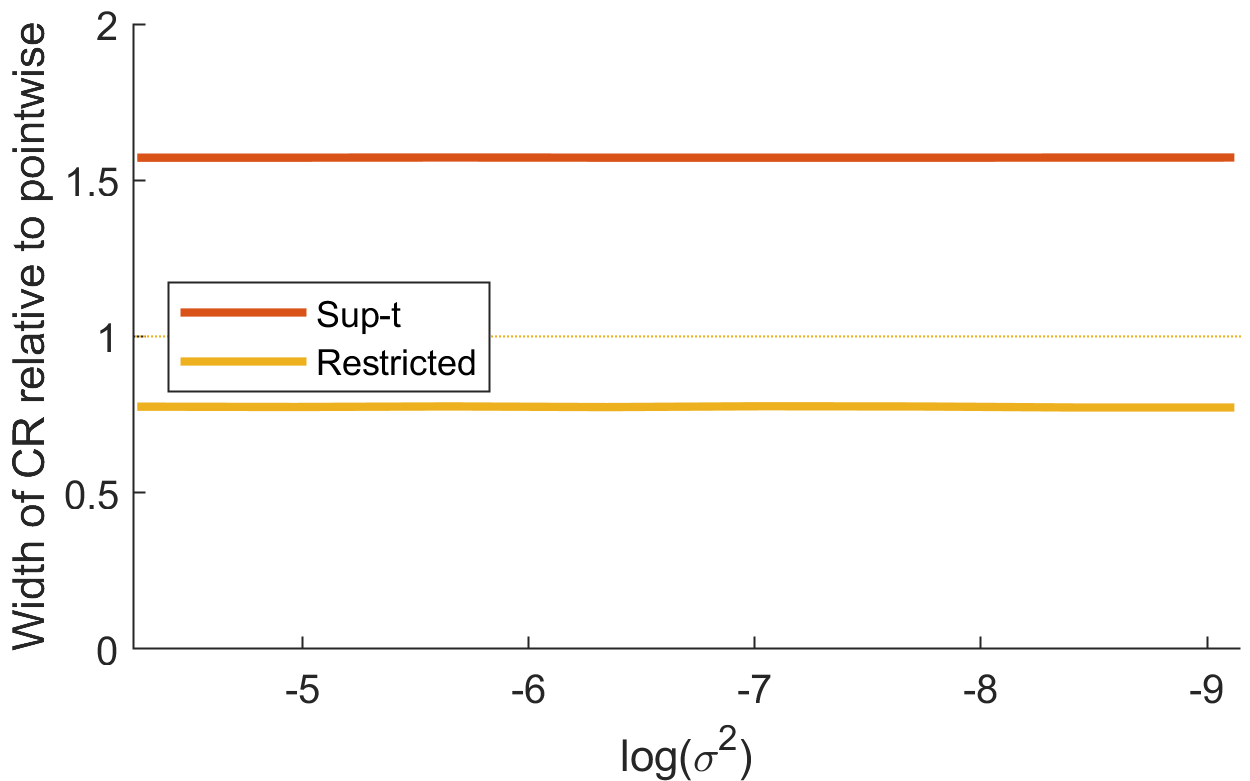}
		\caption{Constant Treatment Effect}
		\label{app-fig:relwidth_constant}
	\end{subfigure}
	\hfill
	\begin{subfigure}[h]{0.49\textwidth}
		\includegraphics[width=\linewidth]{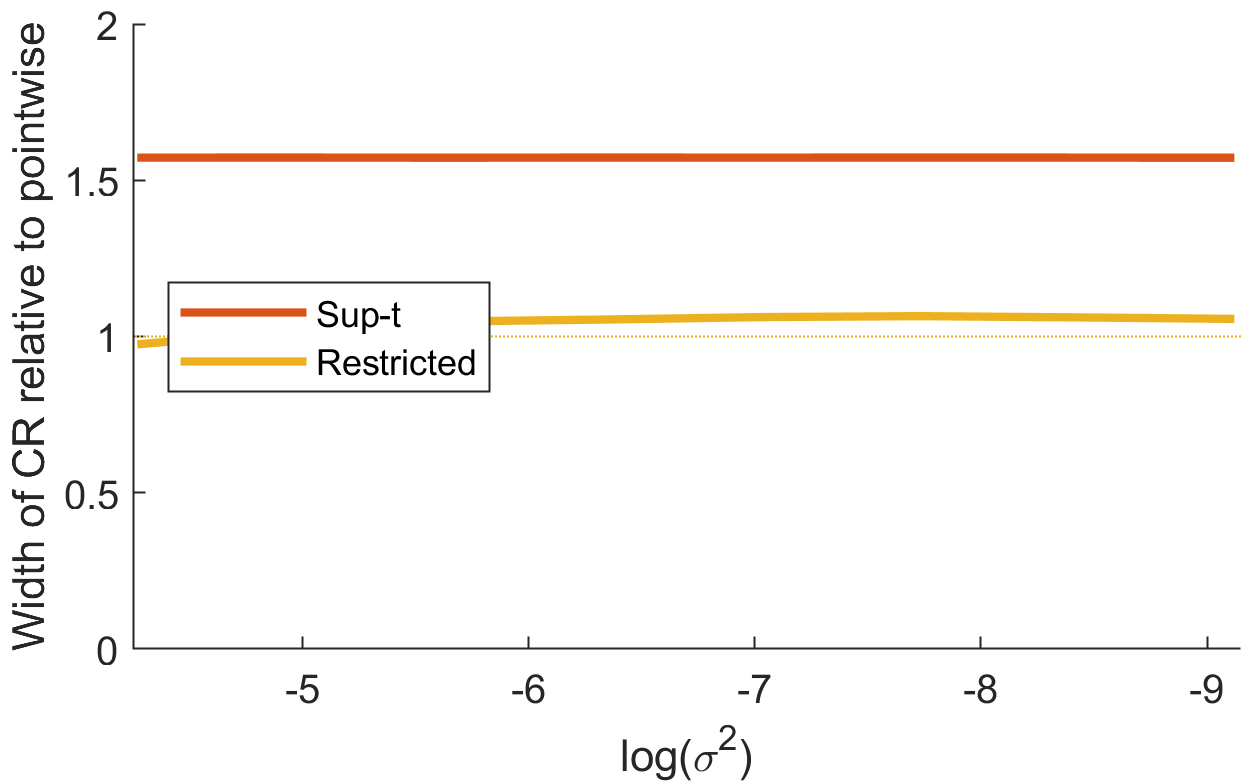}
		\caption{Smooth, eventually flat}
		\label{app-fig:relwidth_quadratic}
	\end{subfigure}
	
	\begin{subfigure}[h]{0.49\textwidth}
		\includegraphics[width=\linewidth]{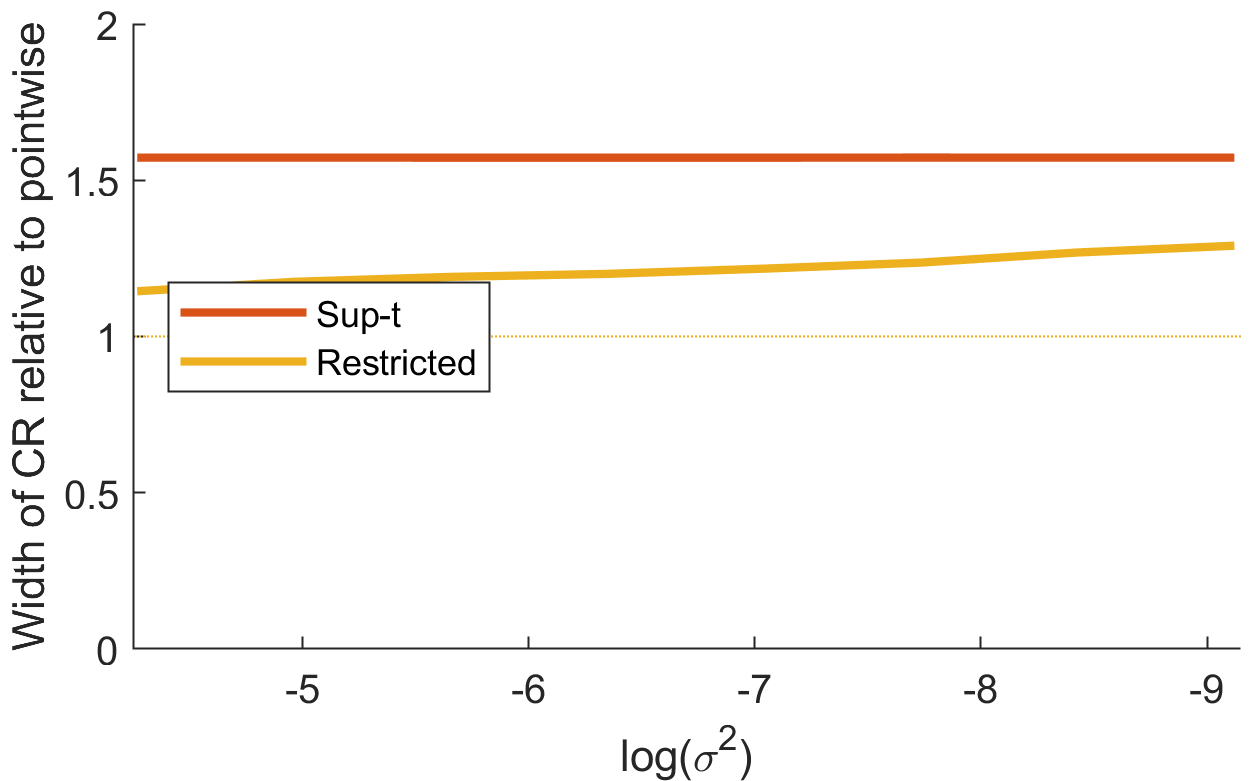}
		\caption{Hump-shaped}
		\label{app-fig:relwidth_noflat}
	\end{subfigure}
	\hfill
	\begin{subfigure}[h]{0.49\textwidth}
		\includegraphics[width=\linewidth]{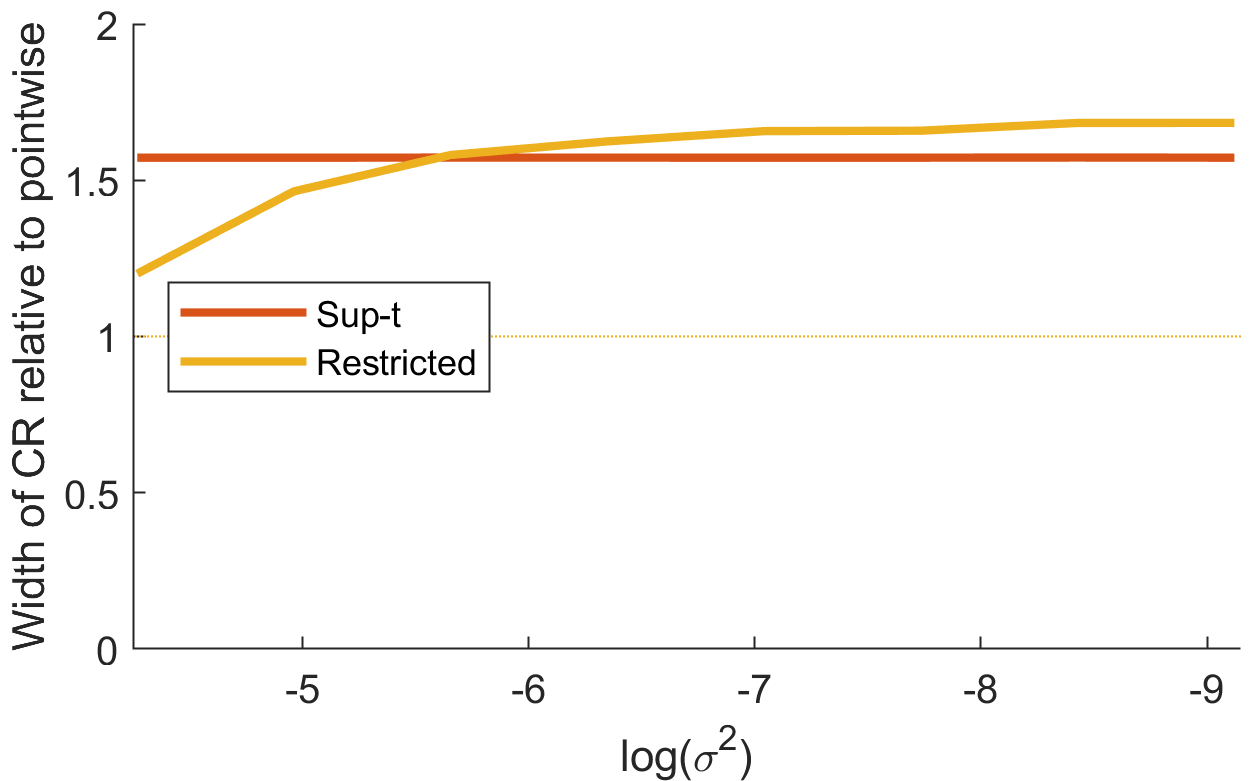}
		\caption{Wiggly}
		\label{app-fig:relwidth_wiggly}
	\end{subfigure}
	\caption[Width of confidence regions under correlation]{Average width of confidence regions relative to pointwise confidence intervals as a function of the amount of noise in the initial estimates $\hat{\beta}$.}
	\label{app-fig:ratio_widths}
\end{figure}

\begin{figure}[tb!]
	\centering
	\begin{subfigure}[h]{0.49\textwidth}
		\includegraphics[width=\linewidth]{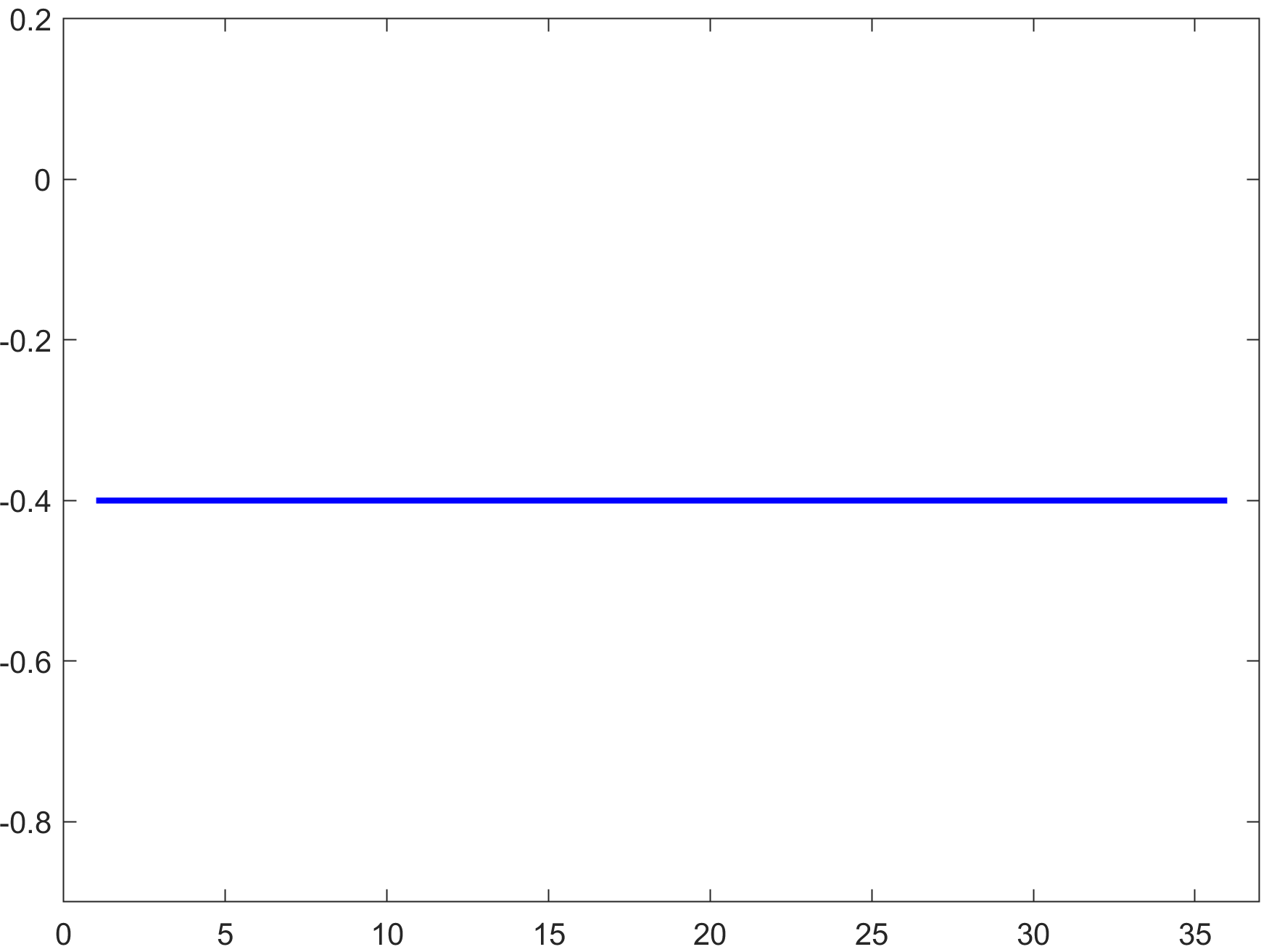}
		\caption{Constant Treatment Effect}
	\end{subfigure}
	\hfill
	\begin{subfigure}[h]{0.49\textwidth}
		\includegraphics[width=\linewidth]{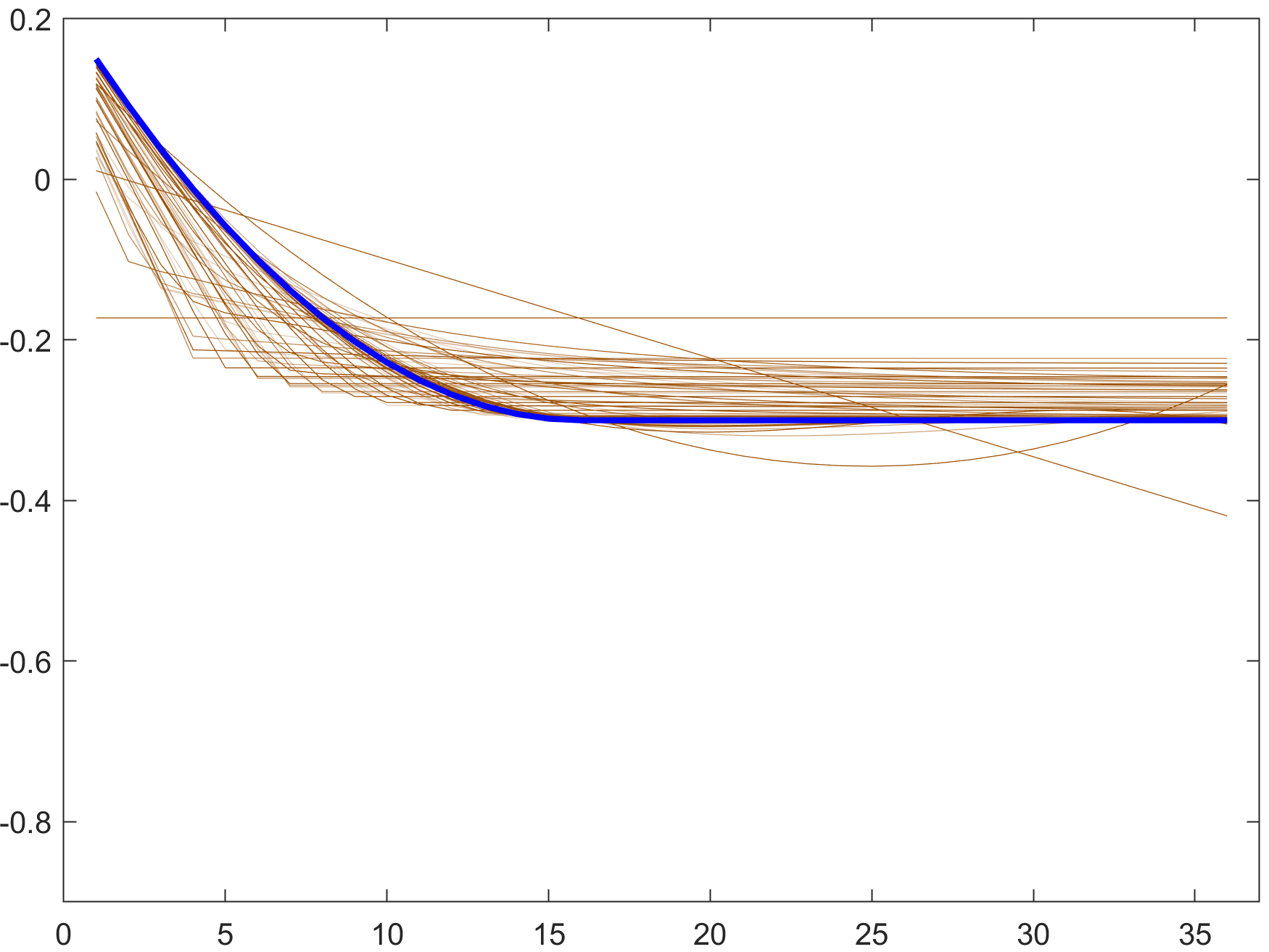}
		\caption{Smooth, eventually flat}
	\end{subfigure}
	
	\begin{subfigure}[h]{0.49\textwidth}
		\includegraphics[width=\linewidth]{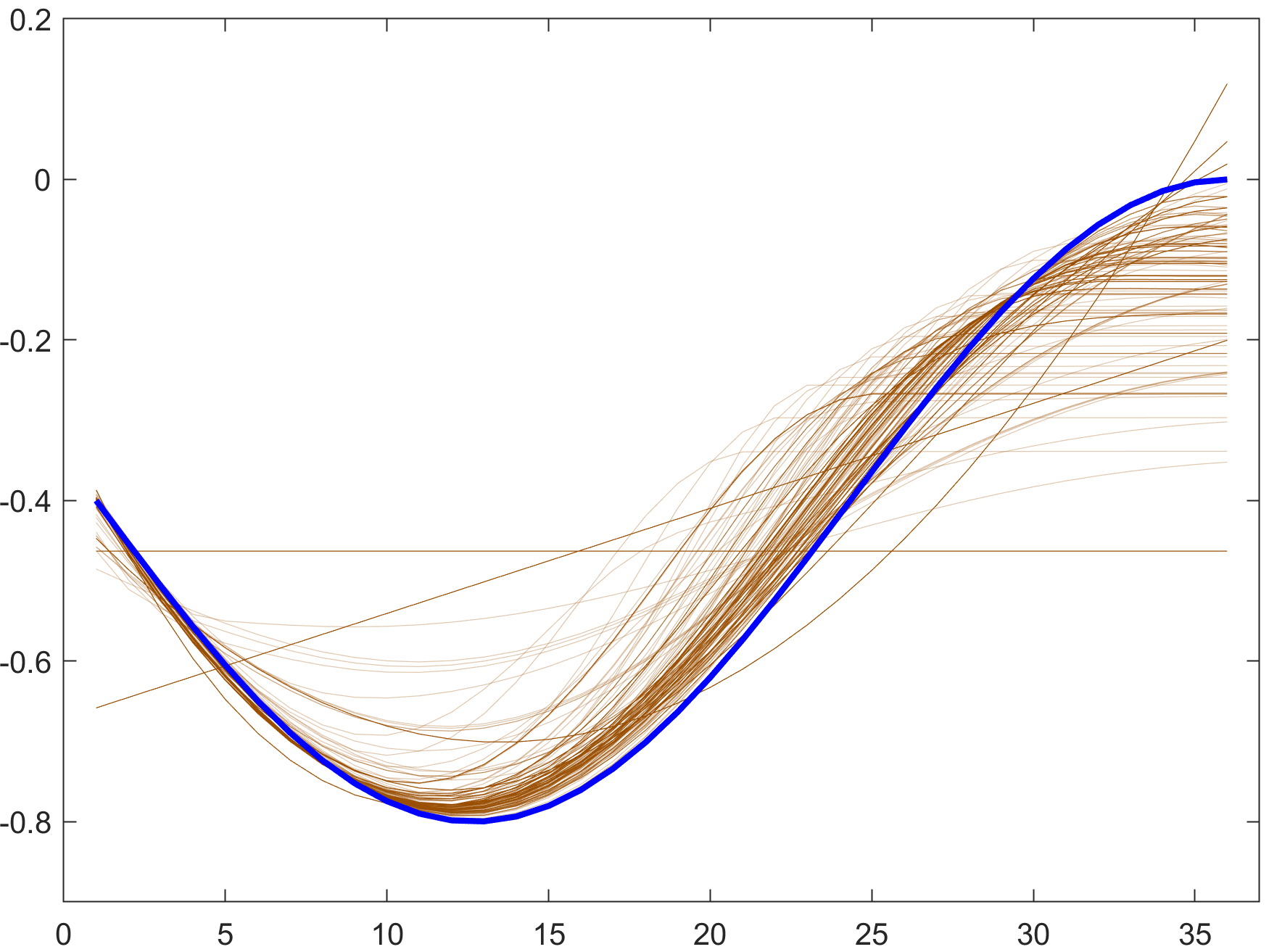}
		\caption{Hump-shaped}
	\end{subfigure}
	\hfill
	\begin{subfigure}[h]{0.49\textwidth}
		\includegraphics[width=\linewidth]{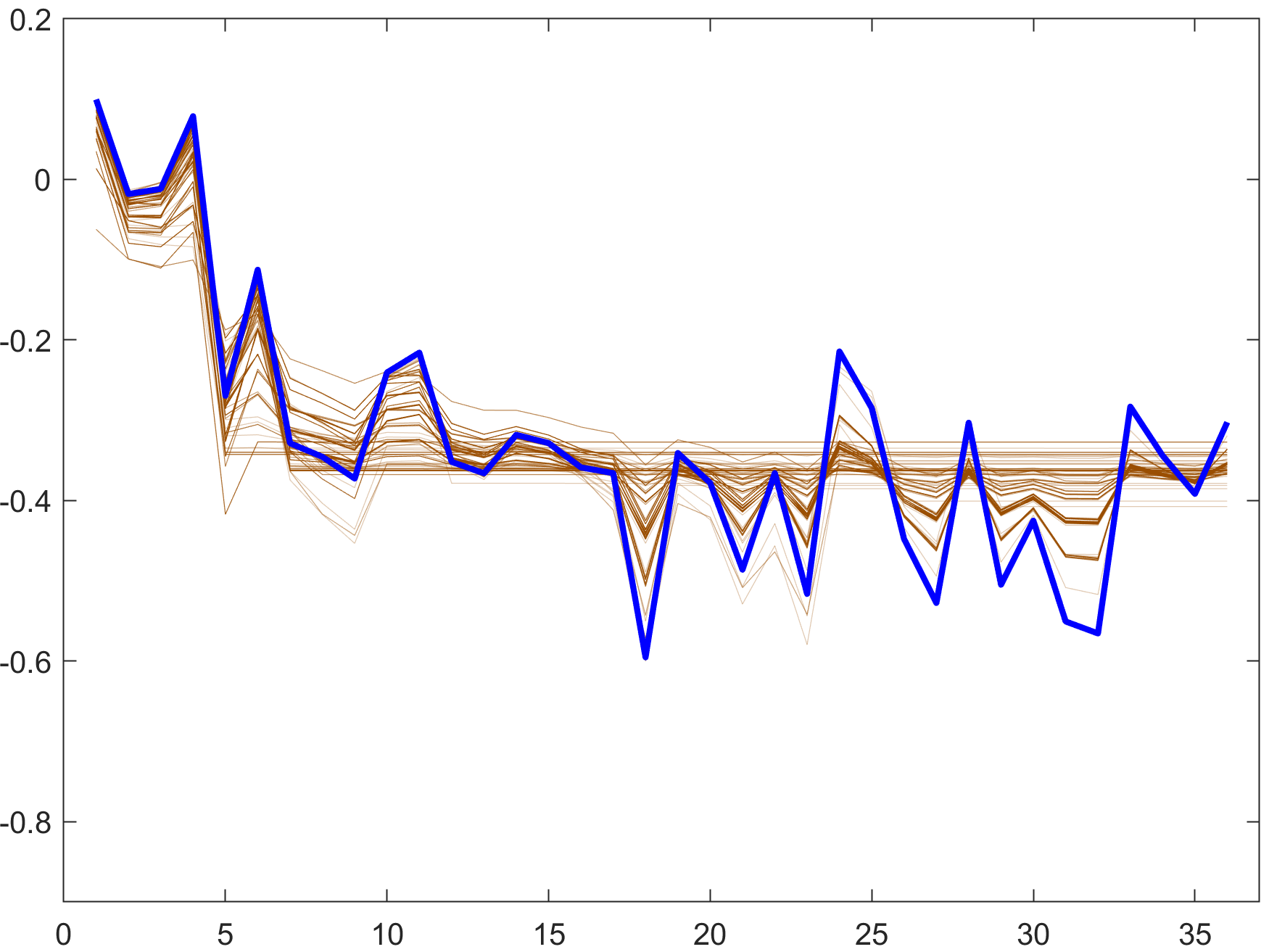}
		\caption{Wiggly}
	\end{subfigure}
	\caption[Illustration of chosen surrogates under correlation.]{Illustration of the 1,000 chosen surrogates for $\sigma^2=0.014$ ($\log(\sigma^2)=-4.27$) under positive correlation in the point estimates.}
	\label{app-fig:surrogates_cor}
\end{figure}

\end{document}